%% file: main.tex
\documentclass[11pt]{article}

\usepackage{geometry,showlabels}
\usepackage{xr-hyper}
\usepackage{times,enumitem}
\usepackage[colorlinks = true,linkcolor =blue,citecolor = blue, urlcolor = blue]{hyperref}
\usepackage{url}

\usepackage{bigints}
\usepackage{amsfonts,amsmath,amssymb,amsthm}
\usepackage{wrapfig}
\usepackage{verbatim,float,url}
\usepackage{graphicx,psfrag}
\usepackage{subcaption}
\usepackage[round]{natbib}   
\usepackage{cancel}
\usepackage{color}

\usepackage{mathtools}

\newtheorem{theorem}{Theorem}
\newtheorem{lemma}[theorem]{Lemma}

\usepackage{mdwlist}	
\theoremstyle{remark}

\theoremstyle{definition}
\newtheorem{definition}{Definition}
\newtheorem{example}{Example}

\usepackage{tikz}  
\usetikzlibrary{trees}

\graphicspath{{figs/}}

\usepackage{mathtools}
\DeclarePairedDelimiter\ceil{\lceil}{\rceil}
\DeclarePairedDelimiter\floor{\lfloor}{\rfloor}

\theoremstyle{plain}
\newtheorem{assumption}{Assumption}

\begin{document}

\title{Adaptive Non-Parametric Regression \\ With the $K$-NN Fused Lasso}

\author{
	Oscar Hernan Madrid Padilla$^{1}$\\
	{\tt omadrid@berkeley.edu}\\
	\and
	James Sharpnack$^{2}$ \\
	{\tt  jsharpna@ucdavis.edu} \\
	\and
	Yanzhen Chen $^{5}$ \\
	{\tt  imyanzhen@ust.hk} \\
	\and
	Daniela Witten$^{3}$ $^{4}$ \\
	{\tt  dwitten@uw.edu} \\
	\begin{tabular}{c}
		$^{1}$ Department of Statistics, UCLA\\
		$^{2}$ Department of Statistics, UC Davis\\
		$^{3}$ Department of Statistics, University of Washington\\
		$^{4}$ Department of Biostatistics, University of Washington\\ 
		$^{5}$  Department of  ISOM, Hong Kong University of Science and Technology
	\end{tabular}
}

\date{\today}

\maketitle

\begin{abstract}
The fused lasso, also known as total-variation denoising, is a locally-adaptive function estimator over a regular grid of design points.
In this paper, we extend the fused lasso to settings in which the points do not occur on a regular grid, leading to an approach for non-parametric regression. This approach, which we call the \emph{$K$-nearest neighbors ($K$-NN) fused lasso}, involves (i) computing the $K$-NN graph of the design points; and (ii) performing the fused lasso over this $K$-NN graph. We show that this procedure has a number of theoretical advantages over competing approaches: specifically, it inherits \emph{local adaptivity} from its connection to the fused lasso, and  it inherits \emph{manifold adaptivity} from its connection to the $K$-NN approach.  We show that excellent results are obtained in a simulation study  and on an application to flu data.  
For completeness, we also study an estimator that makes use of an $\epsilon$-graph rather than a $K$-NN graph, and contrast this with the $K$-NN fused lasso. 
\end{abstract}

\input{intro}

\input{method}

\input{local-adaptivity}

\input{manifold-adaptivity}

\input{prediction}
\input{data}

\input{appendix}

\subsection*{Acknowledgements} JS is partially supported by NSF Grant DMS-1712996. DW is partially supported by NIH Grant DP5OD009145, NSF CAREER Award DMS-1252624, and a Simons Investigator Award in Mathematical Modeling of Living Systems.

\bibliographystyle{abbrvnat} 
\bibliography{graphfused}

\end{document}

%% file: intro.tex
\section{Introduction}
\label{sec:introduction}

In this paper, we consider the non-parametric regression setting in which we have $n$ observations, $(x_1, y_1), \ldots,$ $(x_n, y_n)$, of the pair of random variables $(X,Y) \in \mathcal{X}\times \mathbb{R}$,  where $\mathcal{X}$ is a metric space with metric $d_{\mathcal{X}}$. 
We suppose that the model 
\begin{equation}
\label{eq:model}
y_i    =    f_0(x_i) \,\,+\,\,\varepsilon_i \,\,\,\,\,\,\,(i = 1,\ldots,n)
\end{equation}
holds, where  $f_0 \,:\,  \mathcal{X} \,\rightarrow \,\mathbb{R} $   is an unknown function  that we wish to estimate.  This problem  arises in many settings, including demographic applications \citep{petersen2016convex,sadhanala2017additive}, environmental data analysis \citep{hengl2007regression},  image processing \citep{rudin1992nonlinear}, and causal inference \citep{wager2017estimation}.   

 A substantial body of work has considered estimating the function $f_0$ in \eqref{eq:model} at the observations $X=x_1,\ldots,x_n$ (i.e.~denoising) as well as at other values of the random variable $X$ (i.e.~prediction).  
This includes  seminal papers by \cite{breiman1984classification},  \cite{duchon1977splines}, and  \cite{friedman1991multivariate}, as well as   more recent work by \cite{petersen2016fused}, \cite{petersen2016convex}, and \cite{sadhanala2017additive}. A number of previous papers  have focused in particular on manifold adaptivity, i.e, adapting to the dimensionality of the data; these include   work on local polynomial regression by \cite{bickel2007local} and \cite{cheng2013local}, $K$-NN  regression by \cite{kpotufe2011k},  Gaussian processes by \cite{yang2015minimax} and \cite{yang2016bayesian}, and  tree--based  estimators   such as those in \cite{kpotufe2009escaping} and \cite{kpotufe2012tree}. 
We refer the reader to \cite{gyorfi2006distribution} for a very detailed survey of other classical non-parametric regression methods. The vast  majority of this work performs well in  function classes with variation controlled uniformly throughout the domain, such as Lipschitz and $L_2$ Sobolev classes. \cite{donoho1998minimax} and \cite{hardle2012wavelets} generalize this setting by considering  functions of bounded  variation and  Besov classes. In this work, we focus on piecewise Lipschitz  and bounded variation functions, as these classes can have functions  with  non-smooth regions as well as smooth regions \citep{wang2016trend}.

Recently, interest has focused on  so-called \emph{trend filtering} \citep{kim2009ell_1}, which seeks to estimate 
$f_0(\cdot)$  under the assumption that its discrete derivatives are sparse, in a setting in which we have access to an unweighted graph that quantifies the pairwise 
relationships between the $n$  observations.  In particular, the fused lasso, also known as zeroth-order trend filtering or total variation denoising \citep{mammen1997locally,rudin1992nonlinear,tibshirani2005sparsity,wang2016trend}, solves the optimization problem
\begin{equation}
\text{minimize}_{\theta \in \mathbb{R}^n} \left\{  \frac{1}{2}  \sum_{i=1}^n (y_i \,-\, \theta_i)^2\,\,+\,\,\lambda\, \sum_{(i,j) \in E}  | \theta_i - \theta_j | \right\},
\label{eq:fusedlasso}
\end{equation}
where $\lambda$ is a non-negative tuning parameter, and where $(i,j) \in E$ if  and only if there is an edge between the $i$th and $j$th observations in the underlying graph. Then, $ \hat{f}(x_i) = \hat\theta_i $.  
Computational aspects of the fused lasso have been studied extensively in the case of chain graphs \citep{johnson2013dynamic,barbero2014modular,davies2001local} as well as general graphs \citep{chambolle2009total,chambolle2011first,landrieu2016cut,hoefling2010path,tibshirani2011solution,chambolle2009total}. 
Furthermore, the fused lasso is known to have excellent theoretical properties. In one dimension, \citet{mammen1997locally} and \cite{tibshirani2014adaptive} showed that the fused lasso attains  nearly minimax  rates in mean squared error (MSE) for estimating functions of bounded variation. More recently, also in one dimension, \citet{lin2017sharp} and \citet{guntuboyina2017spatial} independently  proved  that the fused lasso is nearly minimax under the assumption that $f_0$ is piecewise constant. In grid  graphs, \cite{hutter2016optimal}, \cite{sadhanala2016total}, and   \cite{sadhanala2017higher}   proved minimax  results  for the fused lasso when estimating signals of interest in applications of image denoising. In more general graph structures, \cite{padilla2016dfs} showed that the fused lasso  is  consistent  for  denoising problems, provided that the underlying signal has total variation along the graph that divided by $n$ goes to  zero. Other graph  models that have been studied in the literature include tree graphs in \cite{padilla2016dfs} and \cite{ortelli2018total}, and star and  Erd\H{o}s-R\'enyi graphs in \cite{hutter2016optimal}. 

In this paper, we extend the utility of the fused lasso approach by combining it with the \emph{$K$-nearest neighbors} ($K$-NN) procedure. 
$K$-NN has been well-studied from a theoretical   \citep{stone1977consistent,chaudhuri2014rates,von2010hitting,alamgir2014density}, methodological 
\citep{dasgupta2012consistency, kontorovich2016active,singh2016analysis,dasgupta2014optimal}, and algorithmic \citep{friedman1977algorithm,dasgupta2013randomized,zhang2012efficient} perspective.  One key feature of $K$-NN methods is that they automatically have a finer resolution in regions with a higher density of design points; this is particularly consequential when the underlying density is highly  non-uniform.  We study the extreme case in which the data are supported over multiple manifolds of mixed intrinsic dimension. An estimator that adapts to this setting is said to achieve {\em manifold adaptivity}. 


In this paper,  we exploit recent theoretical developments in the fused lasso and the $K$-NN procedure in order to obtain a single approach that inherits the advantages of both methods. In greater detail,  we extend the fused lasso to the general non-parametric setting of \eqref{eq:model}, by performing a two-step procedure.
\begin{list}{}
	\item{\emph{Step 1}.} We construct a $K$-nearest-neighbor ($K$-NN) graph, by placing an edge between each observation and the $K$ observations to which it is closest, in terms of  the metric $d_\mathcal{X}$. 
	\vspace{-0.01in}
	\item \emph{Step 2.} We apply the fused lasso to this $K$-NN graph. 
\end{list}
\noindent The resulting \emph{$K$-NN fused lasso} ($K$-NN-FL) estimator appeared  in the context of image processing in \cite{elmoataz2008nonlocal} and \cite{ferradans2014regularized}, and more recently in an application of graph trend filtering in \cite{wang2016trend}. We are the first to study its theoretical properties. 
We also consider a variant obtained by replacing the $K$-NN graph in Step 1 with 
an $\epsilon$-nearest-neighbor ($\epsilon$-NN) graph, which contains an edge between $x_i$ and $x_j$ only if $d_\mathcal{X}(x_i,x_j)<\epsilon$.

The main contributions of this paper are as follows:\\

\vspace{-0.1in}

	 \emph{Local adaptivity.} We show that provided that $f_0$ has bounded variation, along with
	an additional condition that generalizes piecewise  Lipschitz  continuity (Assumption \ref{as:sup}), then 
	the mean squared errors  of both the $K$-NN-FL estimator
	and the $\epsilon$-NN-FL estimator scale like $n^{-1/d}$, ignoring logarithmic factors; here, $d>1$  is the dimension of $\mathcal{X}$.
	In fact, this matches the minimax rate for estimating a two-dimensional Lipschitz function \citep{gyorfi2006distribution}, but over a much wider 
	function class.\\
	
	\vspace{-0.1in}
	
	\emph{Manifold adaptivity.}  
	Suppose that the covariates are i.i.d. samples from a mixture model $\sum_{l=1}^{\ell} \pi^*_l p_l $, where  $p_1,\ldots,p_{\ell}$ are unknown bounded densities, and the weights $\pi^*_l \in [0,1]$  satisfy $\sum_{l=1}^{\ell} \pi^*_l$ = 1. Suppose further that for $l=1,\ldots,\ell$,  the support $\mathcal{X}_l$ of $p_l$  is 
		homeomorphic (see Assumption \ref{as3}) to  $[0,1]^{d_l} \,=\,  [0,1] \times [0,1] \times \ldots \times [0,1]$, where $d_l > 1$ is the intrinsic
		dimension of $\mathcal{X}_l$. We show that under mild conditions, if the restriction of $f_0$ to $\mathcal{X}_l$  is  a function of bounded
		variation, then the  $K$-NN-FL estimator  attains the rate  \smash{ $ \sum_{l=1}^{\ell}  \pi_l^* (\pi_l^* n )^{-1/d_l} $}. We can obtain intuition for this rate by noticing that $\pi_l^* n$ is the expected number of samples from the $l$th component, and hence $(\pi_l^* n)^{-1/d_l}$ is the expected rate for the  $l$th component. Therefore, our rate is the weighted average of the expected rates for the different components.  
	



%% file: method.tex
\section{Methodology}
\label{sec:method}

\subsection{The $K$-Nearest-Neighbor and $\epsilon$-Nearest-Neighbor Fused Lasso}
\label{sec:Construction_of_approach}
Both the $K$-NN-FL and $\epsilon$-NN-FL approaches are simple two-step procedures. The first step
 involves constructing a graph on the $n$ observations. 
The $K$-NN graph, $G_K=(V, E_K)$, has vertex set  $V \,=\,\{1,\ldots,n\}$, and its edge set $E_{K}$ 
 contains the pair $(i,j)$  if and only if $x_i$ is among  the $K$-nearest neighbors (with respect to the metric $d_{\mathcal{X}}$) of $x_j$, or vice versa.  
By contrast, the $\epsilon$-graph, $G_\epsilon=(V, E_\epsilon)$, contains the edge
 $(i,j)$ in $E_\epsilon$ if and only if $d_{\mathcal X}(x_i,x_j) < \epsilon$.

After constructing the graph, we apply the fused lasso to $y=\begin{pmatrix} y_1 \ldots y_n \end{pmatrix}^T$ over the graph $G$ (either $G_K$ or $G_\epsilon$). We can re-write the fused lasso optimization problem  \eqref{eq:fusedlasso} as
\begin{equation}
\label{eqn:fl}
\hat{\theta}  \,\,\,=\,\,\underset{\theta \in \mathbb{R}^n}{\arg \min} \left\{ \,\, \frac 12   \sum_{i=1}^n (y_i - \theta_i)^2  \,\,+\,\, \lambda\,\| \nabla_{G} \theta  \|_1 \right\},    
\end{equation}
where $\lambda >0 $  is a tuning parameter, and $\nabla_{G}$ is an oriented incidence matrix of $G$; each row of $\nabla_G$ corresponds to an edge in $G$. For instance, if the $k$th edge in $G$ connects the $i$th and $j$th observations, then 
\begin{equation*}
(\nabla_G)_{k,l} = \begin{cases}  1 &  \text{ if } l=i \\
 -1 &  \text{ if } l=j \\
0 & \text{ otherwise }
\end{cases},
\end{equation*} 
and so $(\nabla_G \theta)_k \,=\,  \theta_i - \theta_j$. This definition of $\nabla_G$ implicitly assumes an ordering of the nodes and  edges, which may be chosen arbitrarily without loss of generality. 
In this paper, we mostly focus on the setting where  $G = G_K$ is the $K$-NN graph. 
We also include an analysis of the $\epsilon$-graph, which results from using $G=G_\epsilon$, as a point of contrast.

Given the estimator 
$\hat{\theta}$ defined in \eqref{eqn:fl}, we  predict the response at a new observation
 $x  \in \mathcal{X}  \backslash  \{x_1,\ldots,x_n\}$ according to 
\begin{equation}
\label{eqn:interpolation_rule}
\hat{f}(x)\,\,=\,\,   \frac{1}{\sum_{j=1}^n k(x_j,x)}\sum_{i=1}^n \hat \theta_i \, k(x_i,x).
\end{equation}
In the case of $K$-NN-FL, we take $k(x_i,x)=\mathbf 1_{\{ x_i \in \mathcal N_K(x)\}}$, where $\mathcal N_K(x)$ is the set of $K$ nearest neighbors of $x$ in the 
training data. In the case of $\epsilon$-NN-FL, we take $k(x_i,x)=\mathbf 1_{\{ d_{\mathcal X} (x_i,x) < \epsilon\}}$.
(Given a set $A$,  $\boldsymbol{1}_{A}(x)$ is the indicator function that takes on a value of $1$ if $x \in A$, and $0$ otherwise.) Note that for the $\epsilon$-NN-FL estimator, the prediction rule in (\ref{eqn:interpolation_rule})  might not be well-defined if all the training points are farther than $\epsilon$ from $x$. When that is the case, we set $\hat{f}(x)$ to equal the  fitted value of  the nearest training point.

\begin{figure}[t!]
	\begin{center}
		\captionsetup{width=\linewidth}		
		(a) \hspace{80mm} (b) \\
		\includegraphics[width=2.6in,height= 2.6in]{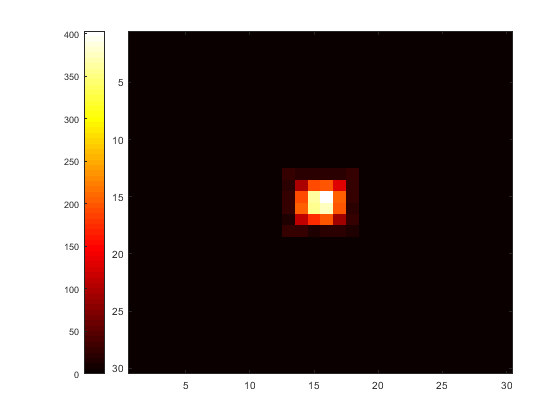}
		\includegraphics[width=2.6in,height= 2.6in]{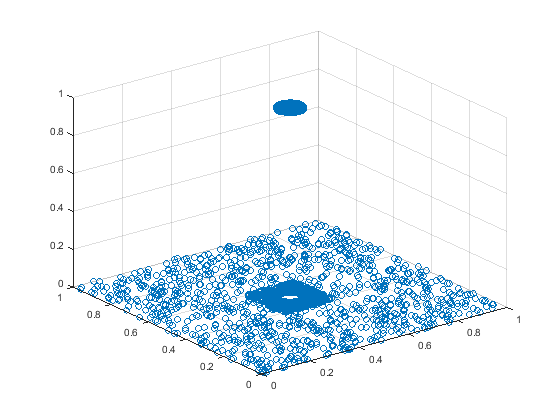}
		\caption{\label{fig:density} \emph{(a):} A heatmap of $n=5000$ draws from \eqref{eqn:form}.  \emph{(b):} $n = 5000$  samples  generated as in
			\eqref{eq:model}, with $\varepsilon_i \overset{\text{ind}}{\sim}  N(0,0.5)$,  $X$ has probability density function as in (\ref{eqn:form}), and $f_0$ is given in  (\ref{eq:f0_example}). The vertical axis  corresponds to $f_0(x_i)$,
			and the other two axes display the two covariates.}
	\end{center}
\end{figure}

We construct  the $K$-NN  and $\epsilon$-NN    graphs  using standard  \verb=Matlab= functions such 
as \verb=knnsearch= and \verb=bsxfun=; this results in a   computational complexity of  $O(n^2)$. 
We solve the fused lasso  with the parametric max-flow  algorithm from \cite{chambolle2009total}, for which  software
 is available from the authors' website, \url{http://www.cmap.polytechnique.fr/~antonin/software/}; it  is in practice much faster than
its  worst-case complexity  of  $O(m\,n^2)$,  where $m$  is the number of edges  in the graph  \citep{boykov2004experimental,chambolle2009total}. 

In $\epsilon$-NN and $K$-NN, the values of $\epsilon$ and $K$ directly affect the sparsity of the graphs, 
and hence the computational performance of the $\epsilon$-NN-FL and $K$-NN-FL estimators. 
Corollary 3.23 in \cite{miller1997separators}  provides  an upper bound
on the maximum degree of arbitrary $K$-NN  graphs in $\mathbb{R}^d$.

\subsection{Example}
\label{sec:example}
To illustrate the main advantages of $K$-NN-FL, we construct a simple example.
We refer to the ability to adapt to the local smoothness of the regression function as local adaptivity, and the ability to adapt to the density of the design points as manifold adaptivity.
The performance gains of $K$-NN-FL are most pronounced when these two effects happen in concert, namely, when the regression function is less smooth where design points are denser.
These properties are manifested in the following example.

We generate  $X \in \mathbb{R}^2$ according to the probability density function
\begin{equation}
	\label{eqn:form}
	 p(x)\,\,=\,\,  \frac{1}{5}\,  \boldsymbol{1}_{\left\{ [0,1]^2   \backslash   [0.4,0.6]^2 \right\}}(x) \,\,+\,\, \frac{16}{25}\, 
 \boldsymbol{1}_{\left\{ [0.45,0.55]^2 \right\} }(x) \,\,+\,\, \frac{4}{25}\,  \boldsymbol{1}_{\left\{ [0.4,0.6]^2 \backslash [0.45,0.55]^2 \right\}}(x).
\end{equation}
Thus, $p$ concentrates 64$\%$ of its mass in the small interval $[0.45,0.55]^2$, and 80$\%$ in $[0.4,0.6]^2$. 
The left-hand panel of  Figure \ref{fig:density} displays a heatmap of  $n =  5000$ observations drawn from \eqref{eqn:form}. 

We define $f_0:\mathbb{R}^2 \rightarrow \mathbb{R}$ in \eqref{eq:model} to be the piecewise constant function 
\begin{equation}
\label{eq:f0_example}
  f_0(x)\,\,=\,\,\     \boldsymbol{1}_{\left\{   \left\| x - \frac{1}{2}(1,1)^T\right \|_2^2   \,\leq \, \frac{2}{1000}      \right \}    }(x).
\end{equation}
We then generate  $\{(x_i,y_i)\}_{i=1}^{n}$ with $n=5000$  from  (\ref{eq:model}); 
 the regression function is displayed in the right-hand panel of Figure~\ref{fig:density}. 
This simulation study has the following characteristics: 
(a) the function $f_0$ in \eqref{eq:f0_example} is not Lipschitz, but does have low total variation, and 
(b) the probability density function $p$ is non-uniform with higher density in the region where $f_0$ is less smooth.

\begin{figure}[t!]
	\begin{center}
		\captionsetup{width=\linewidth}		
		\includegraphics[width=3.0in,height= 2.55in]{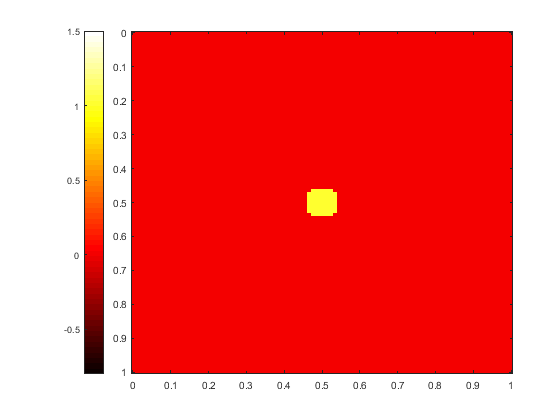}
		\includegraphics[width=2.59in,height= 2.55in]{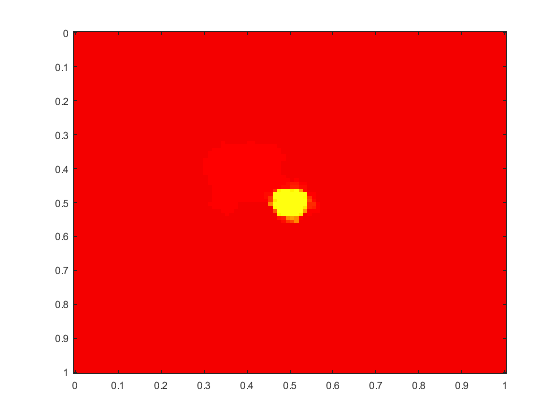}\\
		\hspace{2.3em} 
		\includegraphics[width=2.59in,height= 2.55in]{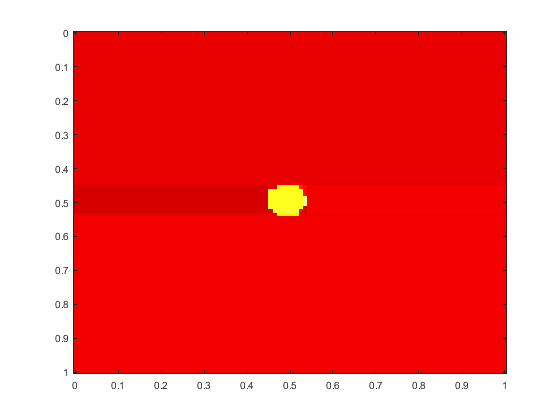}		     	
		\includegraphics[width=2.59in,height= 2.55in]{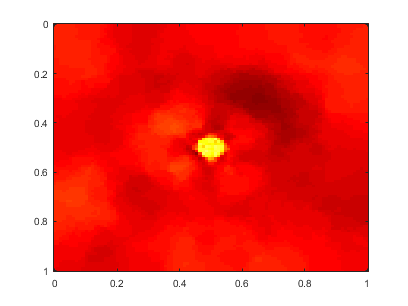}\\
		\caption{\label{fig:scenarios} 
			\emph{Top Left:} The function $f_0$ in \eqref{eq:f0_example}, evaluated at an evenly-spaced  grid of size $100\times100$ in $[0,1]^2$. 
			\emph{Top Right:} The estimate of $f_0$ obtained via $K$-NN-FL. \emph{Bottom Left:} The estimate of $f_0$ obtained via CART. \emph{Bottom Right:} The estimate of $f_0$ obtained via $K$-NN regression. 
		}
	\end{center}
\end{figure}


We compared the following methods in this example:

\begin{enumerate}
\item 
$K$-NN-FL, with the number of neighbors set to $K=5$, and the tuning parameter $\lambda$ chosen to minimize the average MSE over 100 Monte Carlo replicates.
\item 

CART  \citep{breiman1984classification}, with the complexity parameter chosen to minimize the average MSE over 100 Monte Carlo replicates.

\item

 $K$-NN regression   \citep[see e.g.][]{stone1977consistent}, with the number of neighbors $K$ set to minimize the average MSE 
  over 100 Monte Carlo replicates.
\end{enumerate}
\vspace{-0.05in}
The estimated regression functions resulting from these three approaches are displayed in Figure~\ref{fig:scenarios}. 
We see that $K$-NN-FL  can  adapt to low-density and high-density regions of the distribution of covariates, as well as
  to the local structure of the regression function. By contrast, CART displays some artifacts due to the binary splits that make up the decision tree, and $K$-NN regression  undersmoothes in large areas of the domain.

	In practice, we anticipate that $K$-NN-FL will outperform its competitors when the data are highly-concentrated around a low-dimensional manifold and the regression function is non-smooth in that region, as in the above example. In our theoretical analysis, we will consider the special case in which the data lie precisely on a low-dimensional manifold, or a mixture of low-dimensional manifolds.

%% file: local-adaptivity.tex
\section{Local Adaptivity of $K$-NN-FL and $\epsilon$-NN-FL}
\label{sec:local-adaptivity}

\subsection{Assumptions}
\label{sec:compact_support}

We  assume that,  in \eqref{eq:model}, the elements of $\varepsilon \,=\, (\varepsilon_1,\ldots,\varepsilon_n)^T$ are 
independent and identically-distributed  mean-zero sub-Gaussian random variables, 
 \begin{equation}
 \label{eqn:sub_gaussian_errors}
 E(\varepsilon_i) \,=\, 0,\,\,\,\,\,\,\,\,\text{pr}\left( \vert \varepsilon_i\vert \,>\, t  \right)\,\leq \, C\,\exp\left\{-t^2/(2\sigma^2) \right\}, \,\,\,\,i = 1,\ldots, n,  \,\,\,\text{for all}\,\,  t > 0,
 \end{equation}
 for some positive constants $\sigma$  and $C$. Furthermore, we assume that $\varepsilon$ is independent of $X$.
 
 In addition,  for a set $A \subset \mathcal{A}$ with $(\mathcal{A}, d_{\mathcal{A}} )$  a metric space, we write $B_{\epsilon}(A) = \{a \,:\,  \text{exists  }\,\, a^{\prime}  \in A ,\,\,\text{with }\, d_{\mathcal{A}}(a,a^{\prime}) \leq \epsilon  \}$. We let  $\partial A$ denote  
the boundary of the set $A$. Moreover, the MSE of $\hat{\theta}$ is defined as  
 $\|\hat{\theta} - \theta^* \|_n^2 = n^{-1}\sum_{i=1}^n (\hat{\theta}_i - \theta^*_i)^2$. The Euclidean norm of a vector  $x \in \mathbb{R}^d$ is denoted by $\|x\|_2 \,=\,  (x_1^2 + \cdots + x_d^2 )^{1/2}$.  For  $s \in \mathbb{N}$, we set $\boldsymbol{1}_s \,=\,  (1,\ldots,1)^T  \in \mathbb{R}^{s}$.
 In the  covariate  space $\mathcal{X}$, we consider the Borel sigma algebra, $\mathcal{B}(\mathcal{X})$, induced by the metric $d_{\mathcal{X}}$. 
We let $\mu$ be a measure on $\mathcal{B}(\mathcal{X})$. We complement  the model in   \eqref{eq:model}  by  assuming
  that  the  covariates satisfy $x_i  \overset{\text{ind}}{\sim} p(x)$. Thus,  $p$ is the probability  density function  associated  with the distribution of $x_i$, with respect to the measure space $(\mathcal{X},\mathcal{B}(\mathcal{X}),\mu)$.
Note that $\mathcal{X}$ can be a manifold of dimension $d$ in a space of much higher dimension.

We begin by stating assumptions on the distribution of the covariates $p(\cdot)$, and on the metric space $(\mathcal X, d_{\mathcal X})$.  
 In the theoretical results in Section 3  of \cite{gyorfi2006distribution}, 
 it is assumed that $p$  is the probability density function 
of the uniform distribution in $[0,1]^d$. In  this section,  we will require only that $p$ is bounded
   above and below.  This condition appeared in the framework for  studying    $K$-NN graphs  in \cite{von2010hitting}, 
and in the work on density quantization by \cite{alamgir2014density}.

\begin{assumption}
	\label{as1}
	The density $p$  satisfies    $0 \, <\,  p_{\min}   \,<\,  p(x)\,   <\,   p_{\max}$,  for all $x \in \mathcal{X}$, where $p_{\min}, p_{
		\max}  \in \mathbb{R}$.
\end{assumption}

Although we do not  require that $\mathcal{X}$ be a Euclidean space,  we do require that  balls in $\mathcal{X}$  have 
volume (with respect to $\mu$) that behaves similarly to the Lebesgue measure of balls in $\mathbb{R}^d$.  This is  expressed in the next assumption,
 which appeared as part of the definition  of a valid region  (Definition 2) in \cite{von2010hitting}.

\begin{assumption}
	\label{as2}
	The base measure $\mu$  in $\mathcal{X}$ satisfies
	\[
c_{1,d} 	r^{d} \,\,\leq   \mu\left\{ B_{r}(x) \right\} \,\,\leq c_{2,d} r^d,  \,\,\,\,\,\,\,\,\text{for all}\,\,  x  \in   \mathcal{X},       
	\] 
	for all $0 < r < r_0$, where $r_0$, $c_{1,d}$, and $c_{2,d}$  are positive constants, and $d \in \mathbb{N} \backslash \{0,1\} $ is the intrinsic dimension of $\mathcal{X}$.
\end{assumption}

Next, we make an assumption about the topology of the space $\mathcal{X}$. We require that the space have no holes, 
and  is topologically equivalent to  $[0,1]^d$, in the sense that there exists a  continuous bijection
  between $\mathcal{X}$ and $[0,1]^d$.

\begin{assumption}
	\label{as3}
	There exists a homeomorphism (a continuous bijection with a continuous inverse)  $h  \,:\,  \mathcal{X}   \rightarrow  [0,1]^d $, such that
	\[
	L_{\min}\,d_{\mathcal{X}}(x,x^{\prime})  \,\leq \, \| h(x) - h(x^{\prime}) \|_2  \,\leq \, L_{\max}d_{\mathcal{X}}(x,x^{\prime}),  \,\,\,\text{for all}\,\,x, x^{\prime} \in \mathcal{X},
	\]
	for some positive constants $ L_{\min},  L_{\max}$, where  $d \in \mathbb{N} \backslash \{0,1\} $ is the intrinsic dimension of $\mathcal{X}$.
\end{assumption}

Note that Assumptions \ref{as2} and \ref{as3}  immediately hold if  $\mathcal{X}  \,=\, [0,1]^d$,  with  $d_{\mathcal{X}}$   the Euclidean distance, $h$  the identity mapping in $[0,1]^d$, and $\mu$ the Lebesgue measure in $[0,1]^d$. 
A metric space $(\mathcal X, d_{\mathcal X})$ that satisfies Assumption \ref{as3} is a special case of a differential manifold;  
the intuition is that the space $\mathcal X$ is 
 a chart of the atlas for said differential manifold. 
 
In Assumptions \ref{as2} and \ref{as3},  we assume $d>1$, since local adaptivity in  non-parametric regression is well understood in one dimension. For instance, see \cite{tibshirani2014adaptive},\cite{wang2016trend}, \cite{guntuboyina2017spatial}, and references therein.

We now proceed to state conditions on the regression function $f_0$ defined in (\ref{eq:model}). The first assumption  simply requires   bounded variation  of the composition of the regression function with the homeomorphism $h$  from Assumption \ref{as3}.

\begin{assumption}
	\label{as:_rf}
	The function  $g_0 \,\,=\,\,f_0\circ h^{-1}$  has bounded variation, i.e. $g_0  \in \text{BV}\{(0,1)^d\}$, and $g_0$ is also  bounded.
  Here  $(0,1)^d$  is the interior of  $[0,1]^d$,  and $\text{BV}\{(0,1)^d\}$   is the class of 
functions in $(0,1)^d$ with bounded variation. We refer the reader to Section S2  in the Supplementary Material for the explicit construction of the $\text{BV}\{(0,1)^d\}$ class.  The function $h$ was defined in Assumption \ref{as3}.
\end{assumption} 
It is worth mentioning that if $\mathcal{X}  \,=\,  [0,1]^d$  and $h(\cdot)$ is 
the identity function in $[0,1]^d$, then  Assumption \ref{as:_rf} simply  states that $f_0$  has bounded variation. However,
  in order to allow for more general scenarios, the condition is stated in terms of the  function $g_0$  which has domain 
in the unit box, whereas the domain of $f_0$  is the more general set $\mathcal{X}$.

 We now  recall the definition of a  piecewise Lipschitz function,  which induces  a  much larger   class than the set of Lipschitz functions, as it  allows  for  discontinuities.

\begin{definition}
	\label{def:pie_lip}	
Let  \smash{$\Omega_{\epsilon} :=  [0,1]^d \backslash B_{\epsilon}(\partial [0,1]^d) $}. 
	We say that a bounded function $g  : [0,1]^d  \rightarrow \mathbb{R}$ is \emph{piecewise Lipschitz} if there exists a  set $\mathcal{S}   \,\subset\,   (0,1)^d $ that has the following properties:
\begin{enumerate}
 \item  The set $\mathcal{S}$  has Lebesgue measure zero.
\item For some constants $C_{\mathcal{S}}, \epsilon_0 >0$, we have that	$\mu(h^{-1}\{  B_{\epsilon}(\mathcal{S}) \cup ([0,1]^d\backslash \Omega_{\epsilon}) \} ) \,\,\leq \,\,C_{\mathcal{S} }\,\epsilon$   	  for all  $0  < \epsilon < \epsilon_0$.
 
\item There exists a positive constant $L_0$  such that  if  $z$ and $z^{\prime}  $  belong to the same connected component of 
	$\Omega_{\epsilon}  \backslash B_{\epsilon}(\mathcal{S})$,  then $\vert  g(z) \,-\, g(z^{\prime})   \vert   \,\leq \,   L_0 \,\|  z - z^{\prime} \|_2$.
\end{enumerate}
\end{definition}
\noindent Roughly speaking, Definition~\ref{def:pie_lip} says that $g$ is piecewise Lipschitz if there exists a small
 set $\mathcal{S}$ that partitions $[0,1]^d$ in 
such a way that $g$ is Lipschitz within each connected component of the partition. 
Theorem 2.2.1  in \cite{ziemer2012weakly} implies that if $g$ is piecewise Lipschitz,  then $g$ has bounded variation on any open set within a connected
component. 

Theorem~\ref{thm:upper_bound} will require Assumption~\ref{as:sup}, which is a milder assumption 
on $g_0$ than piecewise Lipschitz continuity (Definition~\ref{def:pie_lip}). We now present some notation that is needed in order to introduce  
 Assumption~\ref{as:sup}.

 For $\epsilon > 0 $ small enough, we denote by $\mathcal{P}_{\epsilon}$  a rectangular partition of $(0,1)^{d}$ induced by  $0,\epsilon,2\epsilon,\ldots, \epsilon( \floor{1/\epsilon}   -1),1$, so that all the elements of $\mathcal{P}_{\epsilon}$  have volume of order $\epsilon^d$. 
We define    \smash{$\Omega_{2\epsilon} :=  [0,1]^d \backslash B_{2\epsilon}(\partial [0,1]^d) $}.
Then, for a set $\mathcal{S}   \subset (0,1)^d$,  we define
 $$ \mathcal{P}_{\epsilon,  \mathcal{S} }  \,\,:=\,\, \{ A \cap  \Omega_{2\epsilon}\backslash B_{2\,\epsilon}(\mathcal{S} )\,\,:\,\,\,A \in \mathcal{P}_{\epsilon},\,\,\,\,A \cap  \Omega_{2\epsilon}\backslash B_{2\epsilon}(\mathcal{S} )   \,\neq \, \emptyset  \}; $$   
 this is the partition induced in $\Omega_{2\epsilon}\backslash B_{2 \epsilon}(\mathcal{S} )$  by the grid $\mathcal{P}_{\epsilon}$. 
 
 For a 
 function $g$ with domain $[0,1]^d$, we define 
 \begin{equation}
 \label{summation1}
 S_1(g,\mathcal{P}_{\epsilon,  \mathcal{S} })      := \underset{A \in \mathcal{P}_{\epsilon,  \mathcal{S} } }{\sum} \,\,
 \underset{  z_{A }  \in A   }{\sup} 
 \,\, \frac{   1  }{\epsilon\,}\, \underset{ B_{\epsilon}(z_{A }  ) }{\int}\,\vert g(z_{A})\,-\,  g( z)\vert\,dz
 . \,\,
 \end{equation}
 If $g$ is piecewise Lipschitz, then  $S_1(g,\mathcal{P}_{\epsilon,  \mathcal{S} })$ is  bounded; see Section S3.1 in the Supplementary Material.

 Next, we define
 \begin{equation}
 \label{summation2}
 S_2\left(g,  \mathcal{P}_{\epsilon,  \mathcal{S} }  \right) 
 :=  \underset{A  \in \mathcal{P}_{\epsilon,  \mathcal{S}}}{\sum} \,\, \underset{ z_{A }  \in A  }{\sup}  T(g,z_{A})\,\, \epsilon^{d}, \,\,
 \end{equation}
 where
 \begin{equation}
 \label{summation3}
 T(g,z_{A }) \,\,=\,\,   \underset{ 
 	z \in B_{\epsilon}(z_{A})
 }{\sup}\,\,  \sum_{l=1}^d \,  \left\vert  \underset{ \| z^{\prime}\|_2 \leq \epsilon }{\int}   \frac{\partial\psi(z^{\prime}/\epsilon)}{\partial z_{l}}  \left\{ \frac{ g(z_{A } - z^{\prime})\,-\,  g(z - z^{\prime})}{ \|z -z_{A }\|_2 \,\epsilon^{d} }\right\}\,dz^{\prime}   \right\vert,
 \end{equation}
 and  $\psi$ is  a test function (see Section S1 in the Supplementary Material). 
 Thus, \eqref{summation2} is the summation, over evenly-sized rectangles of volume $\epsilon$ that intersect 
 $\Omega_{2\,\epsilon}\backslash B_{2\,\epsilon}(\mathcal{S} )$, of the supremum  of the function in  \eqref{summation3}.  
 The latter,  for  a function $g$, can be thought as the  average Lipschitz constant  near  $z_A$ ---  see the expression
 inside curly braces in   \eqref{summation3} --- weighted by the derivative of a test function. The scaling factor $\epsilon^d$  in  
 \eqref{summation3} arises because the integral is taken over  a  set of measure proportional to $\epsilon^d$. 
 
 As with   $ S_1\left(g,  \mathcal{P}_{\epsilon,  \mathcal{S} }  \right) $,   one can verify   that  if   $g$  is  a
 piecewise  Lipschitz function,  then   $ S_2\left(g,  \mathcal{P}_{\epsilon,  \mathcal{S} }  \right) $  is bounded.

 We now make use of \eqref{summation1} and \eqref{summation2}
 in order to   state our next condition on $g_0   =   f_0 \circ h^{-1}$.
 This next condition  is milder than assuming that $g_0$ is piecewise Lipschitz, see Definition~\ref{def:pie_lip}.
 \begin{assumption}
 	\label{as:sup}
 	Let  \smash{$\Omega_{\epsilon} :=  [0,1]^d \backslash B_{\epsilon}(\partial [0,1]^d) $}.  
 	There  exists a  set $\mathcal{S}   \,\subset\,   (0,1)^d $ that satisfies the following:
 	\begin{enumerate}
 		\item The set $S$  has Lebesgue measure zero.
 		
 		\item 	For some constants $C_{\mathcal{S}}, \epsilon_0 >0$, we have that 	$\mu \{h^{-1}[  B_{\epsilon}(\mathcal{S}) \cup \{(0,1)^d\backslash \Omega_{\epsilon}\} ] \} \,\,\leq \,\,C_{\mathcal{S} }\,\epsilon$  for all $0  < \epsilon < \epsilon_0$.  
 		
 		\item The summations $S_1(g_0, \mathcal{P}_{\epsilon,\mathcal{S} })$ and $S_2(g_0, \mathcal{P}_{\epsilon,\mathcal{S} })$ are bounded:  
 		$$\underset{ 0  <  \epsilon  < \epsilon_0 }{\sup}\,\, \max\{ S_1(g_0, \mathcal{P}_{\epsilon,\mathcal{S} }), S_2(g_0, \mathcal{P}_{\epsilon,\mathcal{S} })\}  < \infty.$$
 	\end{enumerate}
 \end{assumption}

Finally, we refer the reader to Section S3 in the Supplementary Material for a discussion  on Assumptions \ref{as:_rf}--\ref{as:sup}. In particular, we
	 present an example illustrating that the class of piecewise Lipschitz functions  is, in general, different from the class of functions for which  Assumptions \ref{as:_rf}--\ref{as:sup}  hold. However,  both classes  contain  the class of  Lipschitz functions, where the latter amounts to $\mathcal{S} \,=\,\emptyset$ in Definition  \ref{def:pie_lip}.

\subsection{Results }
\label{sec:upper_bound}

Letting $\theta_i^*=f_0(x_i)$, we express the MSEs of $K$-NN-FL and $\epsilon$-NN-FL in terms of the total variation of $\theta^*$ with respect to the $K$-NN 
and $\epsilon$-NN graphs. 
\begin{theorem}
	\label{thm:upper_bound}
	Let  $K \,\,\asymp\,\,   \log^{1+2r} n $ for some $r > 0$.  Then under  Assumptions \ref{as1}--\ref{as3}, with an 
appropriate choice of the tuning parameter $\lambda$,  the $K$-NN-FL estimator $\hat{\theta}$ satisfies
	\[
	\| \hat{\theta }\,-\, \theta^* \|_n^2  \,\, = \,\, O_{\textup{pr} }\left( \frac{\log^{1+2r} n  }{n}\,+\, \frac{\log^{1.5 + r} n}{n} \|\nabla_{G_K}\,\theta^{*}\|_1 \right).
	\] 
This upper bound also holds for the $\epsilon$-NN-FL estimator  with $\epsilon  \,\asymp\,  \log^{(1+2r)/d} n  / n^{1/d}$,  if we replace  $\|\nabla_{G_K}\,\theta^{*}\|_1 $   with $\|\nabla_{G_{\epsilon}}\,\theta^{*}\|_1 $, 
and  with an appropriate choice of $\lambda$.  
\end{theorem}
Clearly, the upper bound in Theorem \ref{thm:upper_bound} is a function of $\| \nabla_{G_K}\,\theta^{*} \|_1$ or $\| \nabla_{G_{\epsilon} }\,\theta^{*} \|_1$, for the $K$-NN or $\epsilon$-NN  graph, respectively. 
For the grid graph considered in  \cite{sadhanala2016total}, $\| \nabla_{G}\,\theta^{*} \|_1 \,  \asymp \, n^{1-1/d}$,
leading to the rate $n^{-1/d}$.  However,  for a general  graph, there is no  a priori reason to expect that
$\| \nabla_{G}\,\theta^{*} \|_1 \,  \asymp \, n^{1-1/d}$. Notably, 
our next  result shows that   $\| \nabla_{G}\,\theta^{*} \|_1 \,  \asymp \, n^{1-1/d}$   for  $G \in \{ G_{K}, G_{\epsilon}\}$, under
 the assumptions discussed in Section \ref{sec:compact_support}.
\begin{theorem}
	\label{thm:penalty}
	Under Assumptions   \ref{as1}--\ref{as:sup}, or   Assumptions   \ref{as1}--\ref{as3} and  piecewise Lipschitz continuity of  $f_0\circ h^{-1}$,
	if $K   \asymp \log^{1+2r} n  $   for some $r>0$,
	then for an appropriate choice of the tuning parameter $\lambda$,  the  $K$-NN-FL estimator defined in (\ref{eqn:fl})    satisfies 
	\begin{equation}
		\label{eqn:bound1}
		    \| \hat{\theta} \,-\, \theta^* \|_{n}^2  \,\,=\,\,  O_{ \textup{pr} } \left(    \frac{\log^{\alpha} n}{n^{1/d}}  \right),		
	\end{equation}
	with $\alpha  \,=\, 3r  + 5/2  + (2r+1)/d$. Moreover, under Assumptions   \ref{as1}--\ref{as3}  and   piecewise Lipschitz continuity of  $f_0\circ h^{-1}$, then   $\hat{f}$ defined  in (\ref{eqn:interpolation_rule})  with the $K$-NN-FL  estimator satisfies
	\begin{equation}
		\label{eqn:bound2}
			E_{X  \sim p}\left\{\left\vert   f_0(X) \,-\,\hat{f}(X)  \right\vert^2\right\}    	\,\,=\,\, O_{ \textup{pr} }\left\{   \frac{ \log^{  \alpha} n   }{n^{1/d}}     \right\}.
	\end{equation}
	Furthermore, under the same assumptions, (\ref{eqn:bound1}) and (\ref{eqn:bound2}) hold for the $\epsilon$-NN-FL 
 estimator with  $\epsilon   \,\asymp\,  \log^{(1+2r)/d} n  / n^{1/d}$.
	\end{theorem}
Theorem~\ref{thm:penalty} indicates  that under Assumptions \ref{as1}--\ref{as:sup} or Assumptions \ref{as1}--\ref{as3} and piecewise Lipschitz continuity of  $f_0\circ h^{-1}$, both the $K$-NN-FL and $\epsilon$-NN-FL estimators
	attain the convergence rate $n^{-1/d}$, ignoring logarithmic terms.   Importantly, Theorem 3.2 from \cite{gyorfi2006distribution}  
	shows that in the two-dimensional setting, this rate is actually minimax for estimation of  Lipschitz continuous functions, 
	when  the design points are uniformly drawn from $[0,1]^2$. Thus, when  $d=2$ both $K$-NN-FL  and $\epsilon$-NN-FL  are  minimax  for  estimating   functions  in the class  implied by  Assumptions \ref{as1}--\ref{as:sup}, and also in the class of piecewise Lipschitz  functions implied by Assumptions \ref{as1}--\ref{as3} and Definition \ref{def:pie_lip}. In  higher dimensions ($d>2$), by the lower bound in Proposition 2 from \cite{willett2006faster}, we can  conclude that $K$-NN-FL  and $\epsilon$-NN-FL  attain nearly minimax rates  for estimating  piecewise Lipschitz functions, whereas it is unknown if the same is true  under Assumptions  \ref{as1}--\ref{as:sup}. Notably, a different method, similar in spirit to CART, was introduced  in Appendix E of \cite{willett2006faster}. \cite{willett2006faster} showed that this approach is also nearly minimax for  estimating elements in the class of  piecewise  Lipschitz  functions, although is unclear whether a computationally  feasible implementation of their algorithm is available. 


We see from Theorem~\ref{thm:penalty} that both $\epsilon$-NN-FL and $K$-NN-FL are locally adaptive,
 in the sense that they can adapt to the form 
of the function $f_0$. Specifically, these estimators  do not require  knowledge  of  the set  $\mathcal{S}$ in Assumption \ref{as:sup} or Definition \ref{def:pie_lip}. This is  similar in spirit to the one-dimensional  fused lasso, which  does not
 require knowledge of the breakpoints when estimating  a piecewise Lipschitz function.  

 However, there is an important difference between the applicability of Theorem~\ref{thm:penalty} for $K$-NN-FL and $\epsilon$-NN-FL.
 In order to attain the rate in Theorem~\ref{thm:penalty}, $\epsilon$-NN-FL  requires  knowledge of the dimension
 $d$, since this quantity appears in the rate of decay of 
$\epsilon$. But in practice, the value of $d$ might not be clear: for instance, suppose that $\mathcal{X}  \,=\,  [0,1]^2  \times  \{0
	\}    $; this is a subset of $[0,1]^3$,  but it is  homeomorphic  to $[0,1]^2$, so that $d=2$. 
If $d$ is unknown, then it can be challenging to choose $\epsilon$ for $\epsilon$-NN-FL.  
 By contrast,    the choice of $K$ in $K$-NN-FL only involves  the sample size $n$.  Consequently,
 local adaptivity of $K$-NN-FL may be much easier to achieve in practice.

%% file: manifold-adaptivity.tex
\section{Manifold Adaptivity of $K$-NN-FL} \label{sec:manifold-adaptivity}
In this section, we allow the observations $\{ (x_i, y_i) \}_{i=1}^n$ to be 
	drawn from a mixture distribution,
	 in which each mixture component satisfies the assumptions in Section~\ref{sec:local-adaptivity}. Under these assumptions, we show that the $K$-NN-FL estimator can still achieve a desirable rate. 

We assume 
	\begin{equation}
\label{mixture_model}
\begin{array}{lll}
y_i   & =  &  \theta^*_i  \,\,+\,\,\varepsilon_i,\,\,\,\,\,\,i=1,\ldots,n,\\
\theta^{*}_i   &  = & f_{0,z_i}(x_i),  \,\,\,   \\
x_i & \sim &   p_{z_i}(x),  \\
\text{pr}(z_i =l)& \sim& \pi_l^*,  \,\,\,\text{for}  \,\,l = 1,\ldots,\ell. 
\end{array}
\end{equation}
where $\varepsilon$  satisfies  (\ref{eqn:sub_gaussian_errors}), $\pi_l^* \in [0,1]$ with $\sum_{l=1}^{\ell} \pi_l^*=1$,  $p_l$  is a density with support
  $\mathcal{X}_l  \subset   \mathcal{X}$,   $f_{0,l} \,:\,  \mathcal{X}_l  \rightarrow  \mathbb{R} $,  and 
$\{\mathcal{X}_l\}_{l=1,\ldots,\ell}$ is a collection of subsets of $\mathcal{X}$. For simplicity, we will assume
  that $\mathcal{X}   \subset   \mathbb{R}^{d}$ for some   $d >1$,
 and $d_{\mathcal{X}}$ is the Euclidean distance. In  (\ref{mixture_model}), the  observed data is $\{(x_i,y_i)\}_{i=1}^n$. The remaining ingredients in (\ref{mixture_model}) are either latent or unknown.  

We further assume that 
each set $\mathcal{X}_{l}$    is homeomorphic to a Euclidean box  of dimension depending on $l$, as follows:

\begin{assumption}
	\label{as:partition}
			For $l=1,\ldots,\ell$, the set	 $\mathcal{X}_{l}$    satisfies Assumptions \ref{as1}--\ref{as3} with metric given by $d_{\mathcal{X}}$,  with  dimension $d_l \in \mathbb{N} \backslash \{ 0,1\}$, and with  $\mu$ equal to  some measure $\mu_l$. In addition:   \\
				\vspace{-0.2in}
	\begin{enumerate}
		\item	  There exists a positive constant $\tilde{c}_{l}$ such that the set $\partial X_l \,\,:= \,\, \bigcup_{l^{\prime} \neq l }   \mathcal{X}_{l^{\prime}} \cap  \mathcal{X}_l$
		satisfies  
		\begin{equation}
		\label{eqn:measure_condition}
		\begin{array}{lll}
		\mu_l\left\{   B_{\epsilon} \left(  \partial X_l  \right)   \bigcap  X_l  \right\}    \,\,\leq \,\, \tilde{c}_l\epsilon,
		\end{array} 
		\end{equation}
		for any small enough $\epsilon  > 0 $. 
		\item  
		There exists a positive constant $r_l$ such that for any $x \in \mathcal{X}_l$, either 
			\begin{equation}
			\label{eqn:connected}
			 \underset{x^{\prime \prime} \in \partial X_l }{\inf}\, \,d_{\mathcal{X}}(x,x^{\prime \prime}  )   \,\,<\,\,  d_{\mathcal{X}}(x,x^{\prime} ) \;\;\; \text{for all} \;\;\; 
 x^{\prime}  \in \mathcal{X}\backslash \mathcal{X}_l,
			\end{equation}
			or  $B_{\epsilon}(x)  \subset  	\mathcal{X}_l$  for all $\epsilon  < r_l.$

	\end{enumerate}

\end{assumption}


The  constraints   implied by  Assumption  \ref{as:partition}  are very natural. Inequality  (\ref{eqn:measure_condition})  states that the intersections  of the  manifolds  $\mathcal{X}_1,\ldots,\mathcal{X}_{\ell}$  are small.  To put this in perspective, if the  extrinsic space ($\mathcal{X}$) were $[0,1]^d$   with  the Lebesgue  measure, then  balls  of   radius of  $\epsilon$  would have measure  $\epsilon^{d}   $  which is  less than   $\epsilon$  for all $d>1$, and the set  $B_{\epsilon}( \partial [0,1]^d ) \cap  [0,1]^d$  would have  measure  that scales  like $\epsilon$, which is the same  scaling  appearing (\ref{eqn:measure_condition}). Furthermore, 
	\eqref{eqn:connected}
	holds if  $\mathcal{X}_1,\ldots,\mathcal{X}_{\ell}$  are   compact and convex  subsets of $\mathbb{R}^d$   whose interiors are disjoint. 

We are now ready to  extend Theorem \ref{thm:penalty}  to the framework   described in this section.
\begin{theorem}
	\label{thm:adaptivity_dimension}
	Suppose that the data  are generated as in (\ref{mixture_model}),     and Assumption \ref{as:partition}  
holds. Suppose also that the functions   $f_{0,1},\ldots,f_{0,\ell}$  either satisfy   Assumptions \ref{as:_rf}--\ref{as:sup} or are piecewise Lipschitz in the domain $\mathcal{X}_l$. Then for an appropriate choice of the tuning parameter $\lambda$,
 the $K$-NN-FL estimator defined in  (\ref{eqn:fl}) satisfies
	\[
	\|\hat{\theta}   \,-\, \theta^* \|_n^2  \,\,=\,\,  O_{\textup{pr} } \left\{   {\rm poly}(\log n)\,\displaystyle \sum_{l=1}^{\ell} \frac{\pi^*_l }{(\pi^*_l n)^{ 1/d_l   }  } \right\} ,
	\]
	provided that $n \min\{ \pi_l^* \,:\,l \in [\ell]\}  \,\geq  \, c_0 n^{r_0}$, and $K \,\asymp \, \log^{1+2r} n$ for some  constants $c_0,r_0,r > 0 $, where ${\rm poly}(\cdot)$  is a polynomial function. Here,  the  $\pi_l^*$'s are  allowed  to change with $n$.
\end{theorem} 
 Notice that when  $d_l  =  d$ for all $l \in [\ell]$ in Theorem \ref{thm:adaptivity_dimension},  then we obtain, ignoring logarithmic factors, the  rate $ n^{-1/d} $ which  is  minimax when the functions $f_{0,l}$ are piecewise Lipschitz. The rate is also minimax when $d=2$ and the functions  $f_{0,l}$  satisfy  Assumptions \ref{as:_rf}--\ref{as:sup}. In addition, our rates can be compared  with the existing literature on manifold adaptivity. Specifically, when $d=2$, the rate   $n^{-1/2}$  is  attained  by  local polynomial regression \citep{bickel2007local} and   Gaussian process regression \citep{yang2016bayesian} for the class of differentiable functions with bounded  partial derivatives,  and by $K$-NN  regression for Lipschitz functions  \citep{kpotufe2011k}. In higher dimensions,  \cite{bickel2007local}, \cite{yang2016bayesian} and  \cite{kpotufe2011k}    attain better  rates  than $n^{-1/d}$ on smaller classes of functions that do not allow for discontinuities.  

 Finally, we refer the reader to Section S11 of the Supplementary Material for an example suggesting that the $\epsilon$-NN-FL estimator may not be  manifold adaptive.

%% file: data.tex
\section{Experiments}
\label{sec:data}

Throughout this section, we take $d_\mathcal{X}$ to be Euclidean distance.

\begin{figure}[t!]
	\begin{center}
	\captionsetup{width=0.97\linewidth}
(a) \hspace{80mm} (b) \\
		\includegraphics[width=2.6in,height= 2.6in]{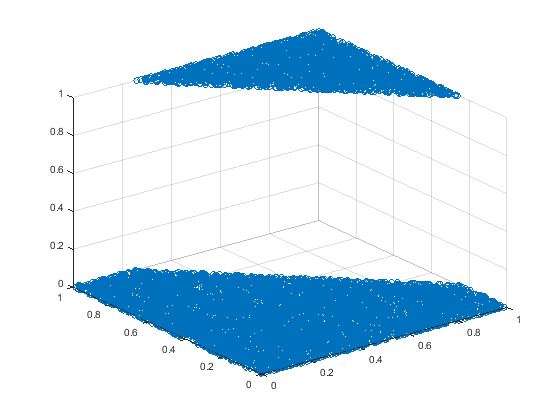}
		\includegraphics[width=2.6in,height= 2.6in]{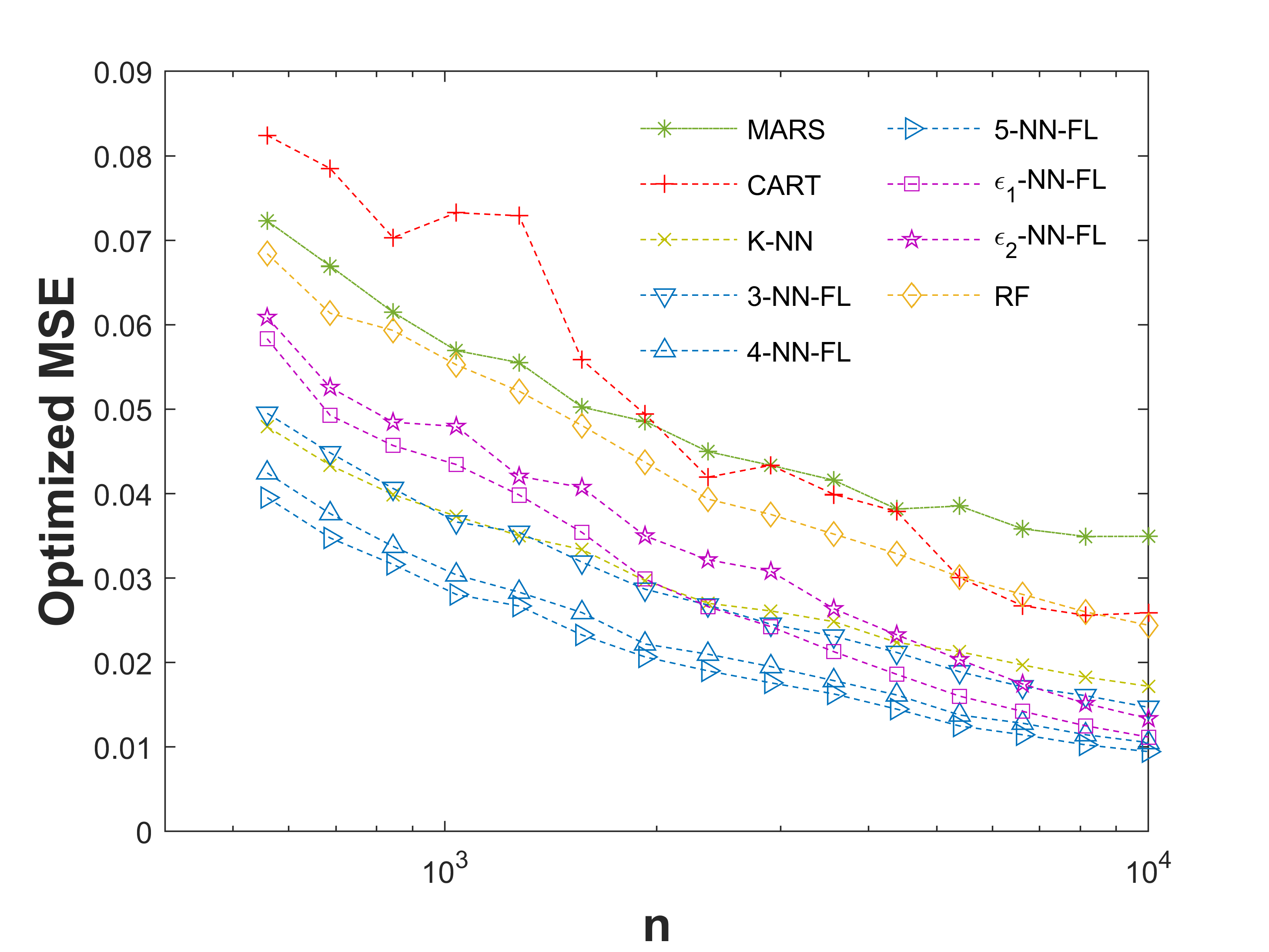}\\
(c) \hspace{80mm} (d) \\
		\includegraphics[width=2.6in,height= 2.6in]{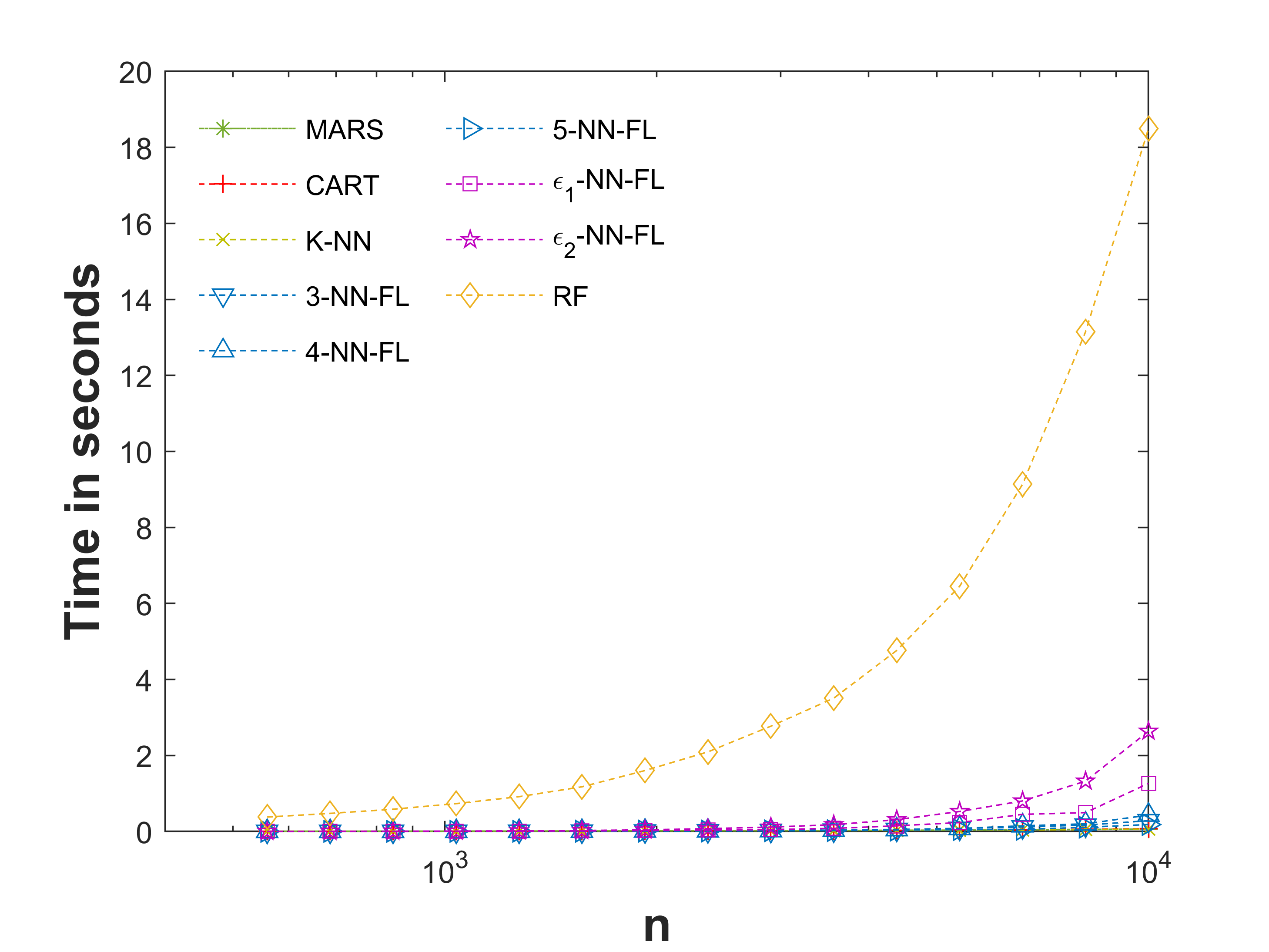}
		\includegraphics[width=2.6in,height= 2.6in]{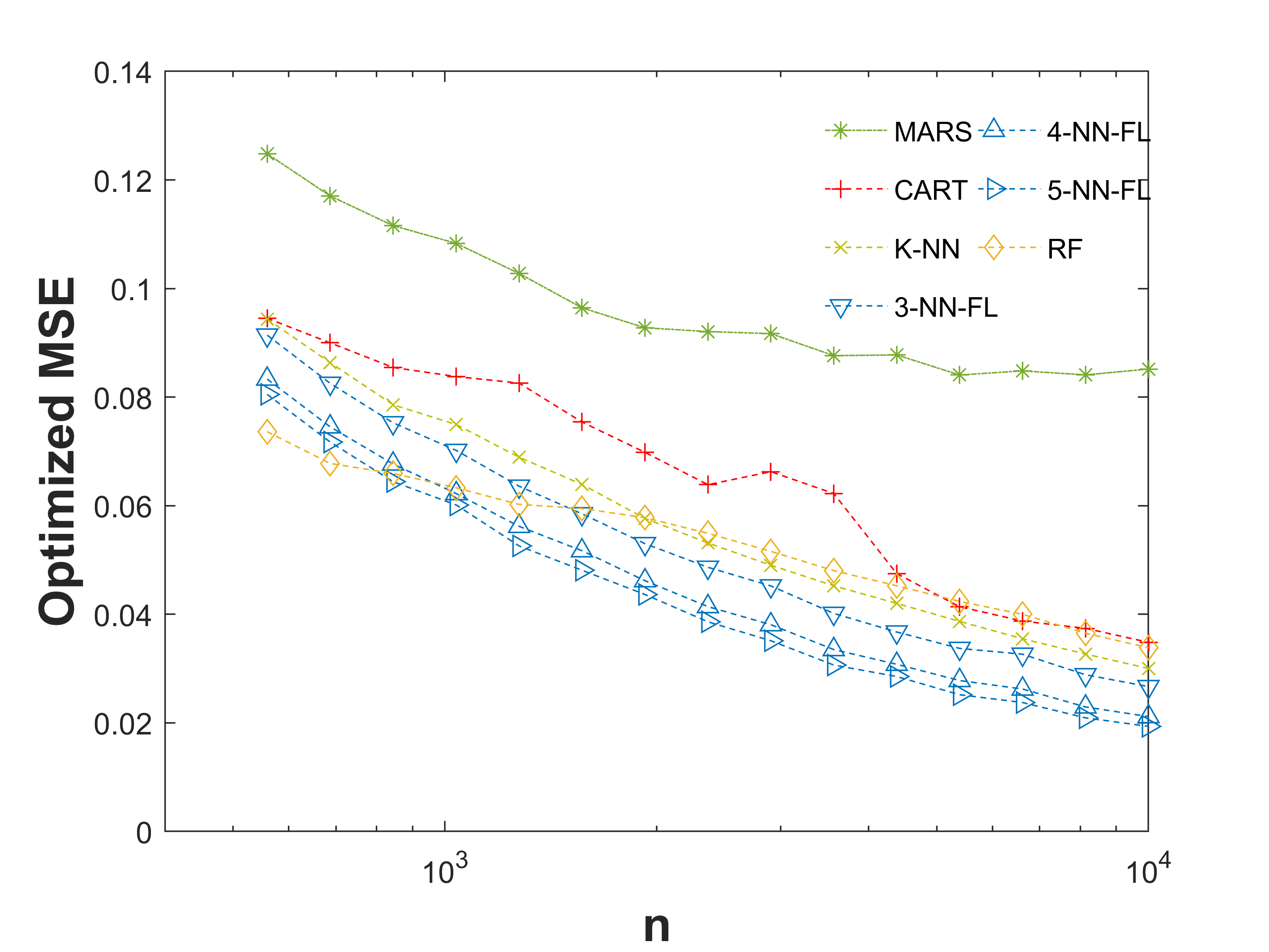}
		\caption{\label{fig:fix_p_var_d} \emph{(a):} A scatterplot of data generated under Scenario 1.  The vertical axis displays $f_0(x_i)$, 
while the other two axes display the two covariates. 
 \emph{(b):} Optimized MSE (averaged over 150 Monte Carlo  simulations) of competing methods under Scenario 1. Here $\epsilon_1 = \frac{3}{4}(\log n/n)^{1/2}$ and $\epsilon_2  \,=\,  (\log n/n)^{1/2} .$
\emph{(c):} Computational time (in seconds)  for Scenario 1, averaged over 150 Monte Carlo simulations. 
\emph{(d):} Optimized MSE (averaged over 150 Monte Carlo simulations) of competing methods under Scenario 2. }
	\end{center}
\end{figure}


\subsection{Simulated Data}\label{sec:sim}

In this section, we compare the following approaches:


\begin{itemize}
\item $\epsilon$-NN-FL, with $\epsilon$ held fixed, and $\lambda$ treated as a tuning parameter.
\item $K$-NN-FL, with $K$ held fixed, and $\lambda$ treated as a tuning parameter. 
\item CART \citep{breiman1984classification}, implemented in the \verb=R= package \verb=rpart=, 
with the complexity parameter treated as a tuning parameter.
\item MARS  \citep{friedman1991multivariate}, implemented in the \verb=R= package \verb=earth=, 
with the penalty parameter treated as a tuning parameter. 
\item Random forests \citep{breiman2001random}, implemented in the \verb=R= package \verb=randomForest=,
 with the number of trees fixed at 800, and with 
the minimum size of each terminal node treated as a tuning parameter.
\item $K$-NN regression \citep[e.g.][]{stone1977consistent}, implemented in \verb=Matlab= using the 
function \verb=knnsearch=, with $K$ treated as a tuning parameter. 
\end{itemize}

We evaluate each method's performance using the MSE, as defined in Section \ref{sec:compact_support}. Specifically, we apply 
 each method to 150 Monte Carlo data sets with a range of tuning parameter values. 
For each method, we then identify the tuning parameter value that leads to the smallest average MSE 
over the 150 data sets. We refer to this smallest average MSE as the 
\emph{optimized MSE} in what follows. 



 In our   first two scenarios we consider $d=2$ covariates, and  let  the sample  size $n$ vary.

\emph{Scenario 1.}
The function $f_0  \,:\, [0,1]^2   \rightarrow \mathbb{R}$  is piecewise constant,
\begin{equation}
	\label{eqn:true_function2}
	f_0(x)  \,\,=\,\,  \boldsymbol{1}_{ \left\{t  \in  \mathbb{R}^2  \,\,:\,\,   \left\|t -   \frac{3}{4}(1,1)^T \right\|_2  \, <\, \left\|t -  \frac{1}{2}(1,1)^T \right\|_2  \right\}   }(x).
\end{equation} 
The covariates are drawn from a uniform distribution on  $[0,1]^2$. The data are generated as in  \eqref{eq:model}  with $N(0,1)$ errors.  

\emph{Scenario 2.} The function $f_0  \,:\, [0,1]^2   \rightarrow \mathbb{R}$  is  as  in \eqref{eq:f0_example}, with generative density for $X$ as in   (\ref{eqn:form}). The data are generated as in  \eqref{eq:model}  with $N(0,1)$ errors. 

Data generated under Scenario 1 are displayed in  Figure~\ref{fig:fix_p_var_d}(a).  
Data generated  under Scenario 2 are displayed in 
Figure~\ref{fig:density}(b).

 Figure~\ref{fig:fix_p_var_d}(b) and Figure~\ref{fig:fix_p_var_d}(d) display the optimized MSE, 
as a function of the sample 
size, for Scenarios 1 and 2, respectively. $K$-NN-FL gives the best results in both scenarios.
 $\epsilon$-NN-FL performs a bit worse than $K$-NN-FL in Scenario 1,
 and very poorly in Scenario 2 (results not shown).

Timing results for all approaches under Scenario 1 are  reported in Figure \ref{fig:fix_p_var_d}(c). 
For all  methods, the times reported are averaged over a range of tuning parameter values. 
 For instance, for $K$-NN-FL,    we fix $K$  and       compute  the time  for different  choices
  of $\lambda$; we then report the average of those times.


\begin{figure}[t!]
	\begin{center}
(a) \hspace{80mm} (b) \\
		\includegraphics[width=2.6in,height= 2.6in]{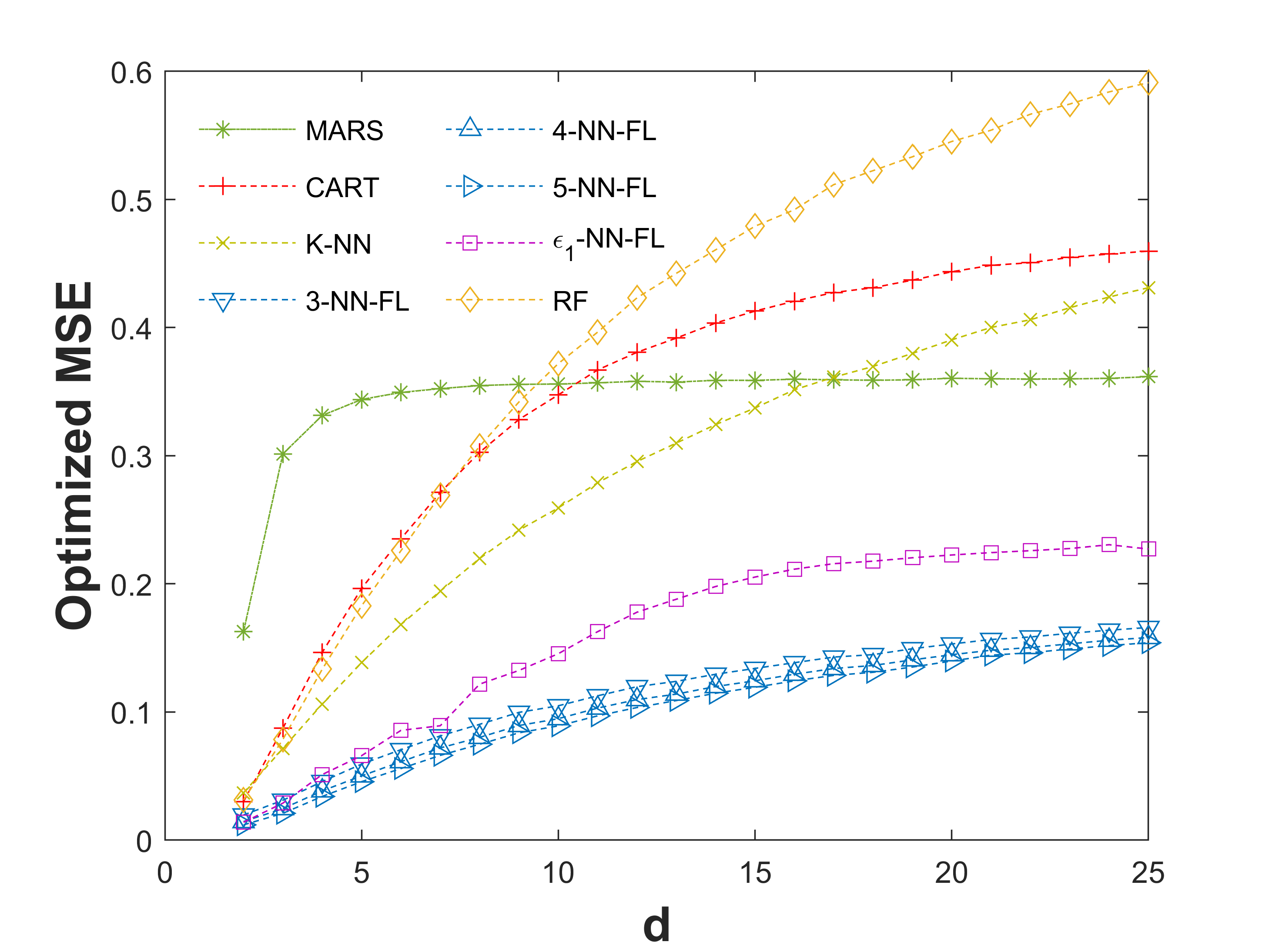}
			\includegraphics[width=2.6in,height= 2.6in]{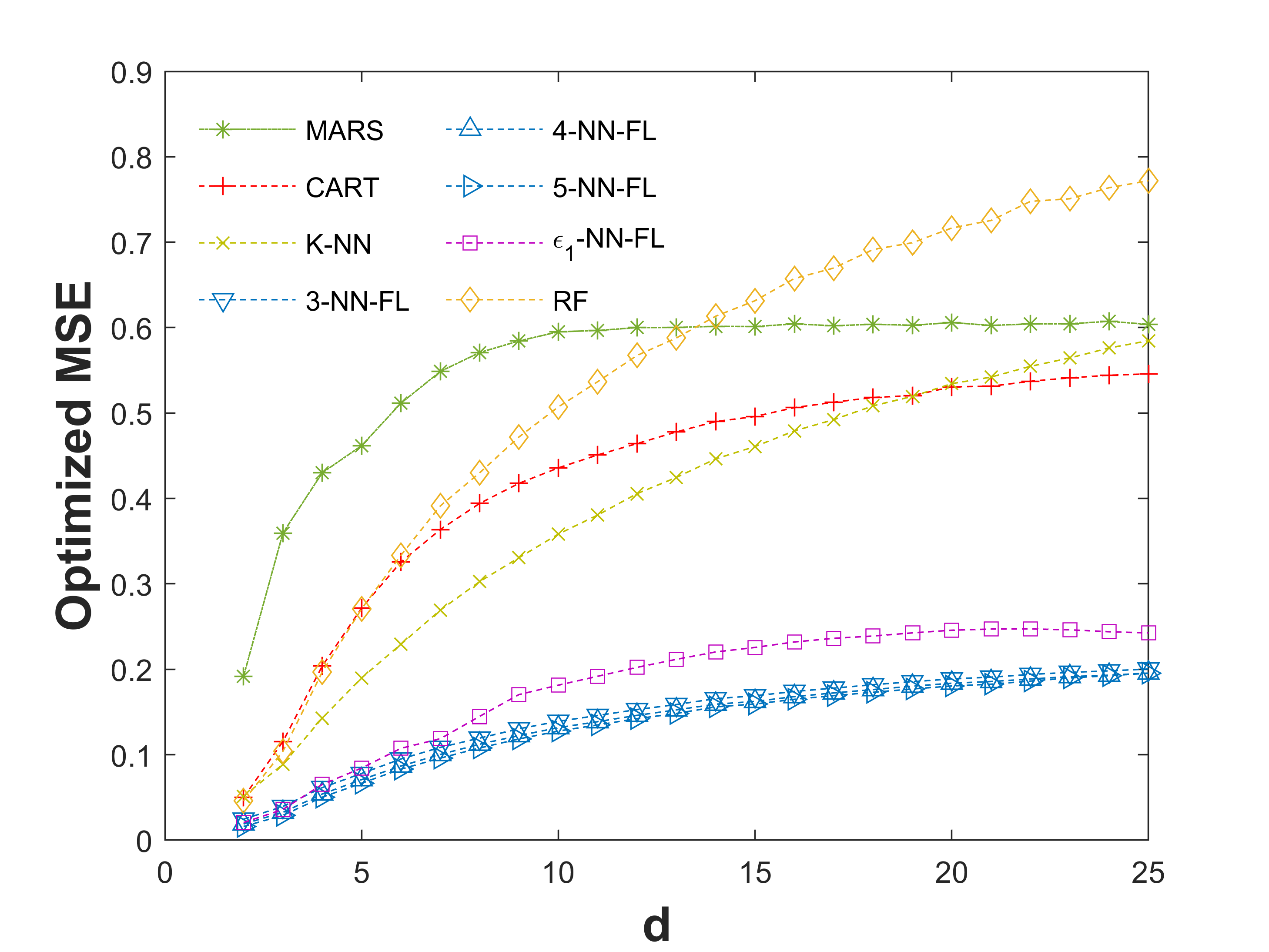}
		\caption{\label{fig:fix_n_var_p_ind10}  \emph{(a)}  Optimized MSE, averaged over 150 
Monte Carlo simulations, for Scenario 3. 
\emph{(b):}  Optimized MSE, averaged over 150  Monte Carlo simulations, for Scenario 4. In both Scenario 3 and Scenario  4,  $\epsilon_1$ is chosen to be the largest  value  such that the total number of edges in the graph $G_{\epsilon_1}$  is at most  50000.}
	\end{center}
\end{figure}

For our  next two  scenarios, we consider  $n=8000$  and values of  $d$ in  $\{2,\ldots,25\}$.

\emph{Scenario 3.}     The function  $f_0  \,:\, [0,1]^d   \rightarrow \mathbb{R}$  is defined as
\[
   f_0(x)     \,\,=\,\,\begin{cases}
     1  & \,\,\,\text{if}\,\,\,  \left\|   x    -  \frac{1}{4} \boldsymbol{1}_d \right\|_2\,<\,\,\left\|   x    -  \frac{3}{4} \boldsymbol{1}_d \right\|_2\\
      -1  &  \,\,\,\text{otherwise},
   \end{cases},
\]
and the  density $p$   is  uniform in $[0,1]^d$. The data are generated as in  \eqref{eq:model}  with $\varepsilon_i \overset{\text{ind}}{\sim}  N(0,0.3)$.

\emph{Scenario 4.} The function  $f_0  \,:\, [0,1]^d   \rightarrow \mathbb{R}$   is defined as 
\[
f_0(x)     \,\,=\,\,\begin{cases}
2  & \,\,\,\text{if}\,\,\,     \|x - q_1\|_2 < \min\{ \|x - q_2\|_2, \|x - q_3\|_2,\|x - q_4\|_2   \}      \\
1& \,\,\,\text{if}\,\,\,    \|x - q_2\|_2 < \min\{ \|x - q_1\|_2, \|x - q_3\|_2,\|x - q_4\|_2   \} \\
0 & \,\,\,\text{if}\,\,\,   \|x - q_3\|_2 < \min\{ \|x - q_1\|_2, \|x - q_2\|_2,\|x - q_4\|_2   \} \\
-1  & \,\,\,\text{otherwise}\,\,\,    \\
   \end{cases},
\]
where  $q_1 \,=\, \left(\frac{1}{4}\boldsymbol{1}_{ \floor{d/2} }^T ,   \frac{1}{2}\boldsymbol{1}_{d -  \floor{d/2}}^T    \right)^T       $, $q_2 \,=\, \left(\frac{1}{2}\boldsymbol{1}_{ \floor{d/2} }^T ,   \frac{1}{4}\boldsymbol{1}_{d -  \floor{d/2}}^T    \right)^T   $,  $q_3 \,=\, \left(\frac{3}{4}\boldsymbol{1}_{ \floor{d/2} }^T ,   \frac{1}{2}\boldsymbol{1}_{d -  \floor{d/2}}^T    \right)^T  $  and
$q_4 \,=\, \left(\frac{1}{2}\boldsymbol{1}_{ \floor{d/2} }^T ,   \frac{3}{4}\boldsymbol{1}_{d -  \floor{d/2}}^T    \right)^T  $. Once again,  the generative density for $X$ is uniform in $[0,1]^d$. The data  are generated  as in  \eqref{eq:model}  with $\varepsilon_i \overset{\text{ind}}{\sim}  N(0,0.3)$.

Optimized MSE for each approach is displayed in 
 Figure~\ref{fig:fix_n_var_p_ind10}.  
When $d$ is small, most methods perform well; however, 
  as $d$ increases, the competing methods quickly deteriorate, whereas $K$-NN-FL  continues to perform well.

  \begin{figure}[t!]
  	\begin{center}
  		\includegraphics[width=4.4in,height= 2.5in]{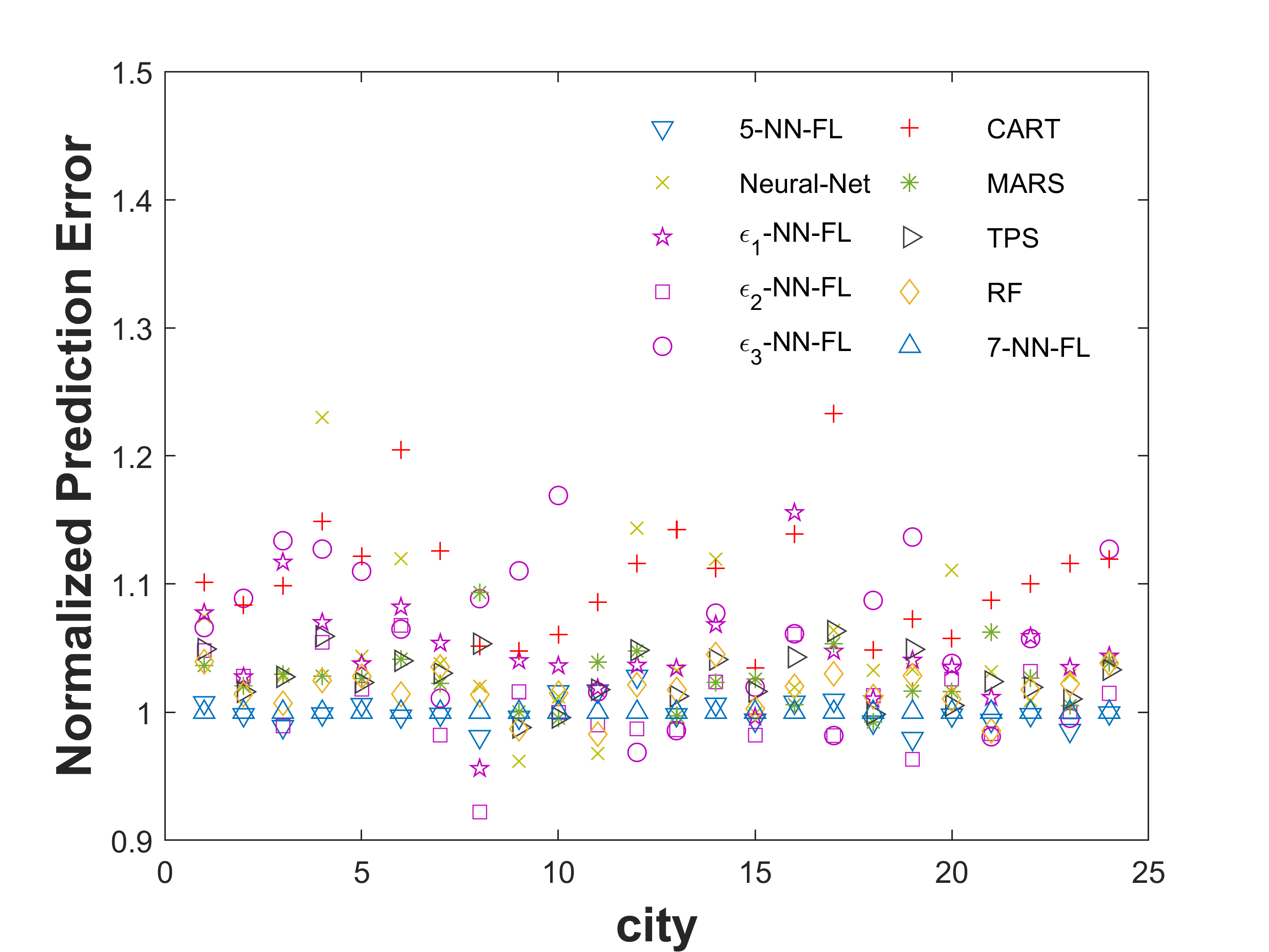}
  		\caption{\label{fig:flu_comparisons} 
Results for the flu data. ``Normalized Prediction Error'' was obtained by dividing each method's test set prediction error by the test set  prediction error
  of $K$-NN-FL, with $K=7$. }
  	\end{center}
  \end{figure}
  
\subsection{Flu Data}
\label{sec:flu}

Our data consists of 
 flu activity and
atmospheric conditions   between
January 1, 2003 and December 31, 2009 across different cities in Texas.  Our data-use agreement
does not permit dissemination of the flu activity, which comes from medical records. The atmospheric conditions, which include temperature and air
quality (particulate matter), can be obtained directly from CDC Wonder
(\url{http://wonder.cdc.gov/}).   Using  the number of flu-related doctor's office visits  as the 
dependent variable, we fit a 
separate non-parametric regression
 model to each of 24 cities; each day was treated as a separate observation, so that  the number 
of samples is $n=2556$ in each city. Five  independent variables are included in the regression:   
maximum and average observed concentration of particulate matter, 
maximum and minimum temperature, and day of the year.  All variables  are scaled to lie in $[0,1]$. 
We performed 50 75\%/25\% splits of the data into a training set and a test set. All models were fit on 
the training data, with 5-fold cross-validation to select tuning parameter values.
Then prediction performance was evaluated on the test set.

We apply $K$-NN-FL with $K \in \{5,7\}$,  and $\epsilon$-NN-FL with  
$\epsilon   =    j /  n^{1/d} $  for $j \in \{1,2,3\}$, which is motivated by Theorem \ref{thm:penalty}, and  with larger choices of $\epsilon$ leading to worse performance. We also fit neural networks 
(\citealp{hagan1996neural}; implemented in \verb=Matlab= using the functions \verb=newfit= and \verb=train=),
    thin plate splines (\citealp{duchon1977splines}; implemented using the \verb=R= package \verb=fields=),  and 
MARS, CART, and random forests, as described in Section~\ref{sec:sim}.

  Average test set prediction error (across the 50 test sets)  
is displayed   in Figure \ref{fig:flu_comparisons}. We see that $K$-NN-FL and $\epsilon$-NN-FL
  have the best performances. 
In particular,  $K$-NN-FL  performs best in 13  out the 24 cities, and  second best in 6  cities. 
In 8 of the 24 cities, 
 $\epsilon$-NN-FL  performs best. 

We contend that $K$-NN-FL achieves superior performance because it adapts to heterogeneity in the density of design points, $p$ (manifold adaptivity), and heterogeneity in the smoothness of the regression function, $f_0$ (local adaptivity).
In our theoretical results, we have substantiated this contention through prediction error rate bounds for a large class of regression functions of heterogeneous smoothness and a large class of underlying measures with heterogeneous intrinsic dimensionality.
Our experiments demonstrate that these theoretical advantages translate into practical performance gains.

%% file: appendix.tex
\appendix

\section{Notation}
\label{sec:notation}

Throughout, given $m \in \mathbb{N}$, we denote by $[m]$ the set $\{1,\ldots,m\}$. For   $a \in A$, $A  \subset \mathcal{A}$, with  $(\mathcal{A},d_{\mathcal{A}} )$ a metric space, we 
write
\[
\begin{array}{lll}
d_{\mathcal{A}}(a,A) &  = &\underset{b \in A }{\inf }\,  d_{\mathcal{A}}(a,b). \\
\end{array}
\]
In the case of the space $\mathbb{R}^d$, we will use the notation $\text{dist}(\cdot,\cdot)$ for the metric induced by the norm $\|\cdot\|_{\infty}$, and we will write $B_{\epsilon}(x) $ for the ball $B_{\epsilon}(x,\|\cdot\|_2)$.
We  use the notation $\|\cdot\|_n$  for the rescaling of the $\|\cdot\|_2$  norm, such that for  $x \in \mathbb{R}^n$,
\[
\|x\|_n^2  \,\,=\,\,  \frac{1}{n} \sum_{i=1}^{n} \,x_i^2.
\]
Furthermore,  we write
\begin{equation}
\label{eqn:set_epsilon}
\Omega_{\epsilon}   \,\,=\,\, [0,1]^d  \backslash  B_{\epsilon}(\partial [0,1]^d ).
\end{equation}  
Thus, $\Omega_{\epsilon}$ is the set of points in the interior of $[0,1]^d$  such that  balls of radius $\epsilon$ with center in such points are also contained in $[0,1]^d$.

Given  an open set  $\Omega \subset \mathbb{R}^d$,  as  is usual in mathematical analysis, we denote by $C_c^{1}(\Omega,\mathbb{R}^d )$  the set of functions  $\phi  \,:\,\Omega  \rightarrow  \mathbb{R}^d $ that  have  compact support  and continuous first derivative. We also write  $C^{\infty}(\Omega)$ for the class of functions   $\phi  \,:\, \Omega   \rightarrow  \mathbb{R} $  which have derivatives of all orders. 
The function   $\psi  \,:\, \mathbb{R}^d   \rightarrow \mathbb{R}$ given as
\begin{equation}
\label{eqn:mollifier}
\psi(z) \,\,=\,\,  \begin{cases}
C_1\,\exp\left(  \frac{-1}{1  -  \|x  \|_2^2}  \right)  & \,\,\,\,\,\,\,\text{if} \,\,\,\|x\|_2  \,<\,1,\\
0 &\,\,\,\,\,\,\,\text{otherwise},
\end{cases}
\end{equation}
is  called a test function if  $C_1$  is chosen such that  $ \int \psi(z) \, dz  \,=\,1$.
We also denote, for $s \in \mathbb{N}$,  by $\boldsymbol{1}_{s}$  the vector $\boldsymbol{1}_{s}   \,:=\, (1,\ldots,1)^{T}   \in \mathbb{R}^s$. Moreover, for vectors $u \in \mathbb{R}^s$ and $v \in \mathbb{R}^t$, we write $w   =  (u^T, v^T)^T \in \mathbb{R}^{s+t}$  for the concatenation of $u$ and $v$.

Furthermore,  we denote by  $L_1(\Omega)$  the set of Lebesgue measurable functions   $f \,:\,  \Omega  \,\rightarrow \, \mathbb{R}$ such that
\[
\displaystyle \int_{\Omega }  \, \vert  f(x)  \vert \mu(dx)  \,< \, \infty,   
\]
where  $\mu$  is the Lebesgue  measure  in $\Omega$. We also denote by $L_{\infty}(
\Omega)$ the set of measurable functions  $f \,:\,  \Omega  \,\rightarrow \, \mathbb{R}$, such that  $\vert f\vert$ is bounded except perhaps in set of $\mu$-measure zero. The supremum norm in
$L_{\infty}(
\Omega)$ is denoted by $\|\cdot\|_{L_{\infty}(\Omega)}$. 


\section{Definition of Bounded  Variation}

We now review some notation regarding general spaces of functions of bounded variation. Let  $\Omega   \,\,\subset  \,\,\mathbb{R}^d$ be an open set.    Recall that the divergence of a function $g \,:\, \Omega \rightarrow \mathbb{R} $ is defined  as
\[
\text{div}(g)(x)\,\,=\,\,\sum_{j = 1}^d  \frac{\partial g(x)}{\partial x_j},  \,\,\,\,\,\forall  x \in \Omega,   
\]
provided that  the partial derivatives involved exist. 

\begin{definition}
	A function $f  \,:\,   \Omega  \rightarrow \mathbb{R}$   has bounded variation if
	\[
	\vert  f \vert_{\text{BV}(\Omega)}  \,\,:=\,\,    \sup \left\{  \underset{\Omega }{\int} f(x) \,\text{div}(g)(x)\,dx   \,; \,\,   g \in  C_c^{1}(\Omega,\mathbb{R}^d ),\,\,       \|g\|_{\infty}    \,\leq  \,1        \right\},
	\]
	is finite, where 
	\[
	\|  g\|_{\infty}  \,\,:=\,\,  \left\|   \left(   \sum_{j=1}^d   g_j^2  \right)^{1/2}  \right\|_{L_{\infty}(\Omega)  }. 
	\]
	The space of functions of bounded variation on $\Omega$ is denoted by $\text{BV}(\Omega)$. We refer the reader to  \cite{ziemer2012weakly}  for a full description of the class of functions of bounded  variation.
\end{definition}

\section{Discussion of Assumptions \ref{as:_rf}--\ref{as:sup} }

\subsection{Piecewise Lipschitz Condition and Assumption \ref{as:sup}  }


To better understand the definition  of $S_1(g,\mathcal{P}_{\epsilon,  \mathcal{S} })$ (see  \eqref{summation1} in the main paper), notice that we can view  $S_1(g,\mathcal{P}_{\epsilon,  \mathcal{S} })$
as a discretization  over the set $\Omega_{2\epsilon}\backslash B_{2\epsilon}(\mathcal{S} )$ of the integral  
of the function   
\begin{equation}
\label{eqn:lebesgue}
m(z_A,\epsilon)=   \frac{1}{\epsilon^{d+1}}   \underset{ B_{\epsilon}(z_{A}) }{\int}\,\vert g(z_A)- g( z)\vert dz. 
\end{equation}
Recall that $z_A$ is a Lebesgue point of $g$ if $\lim_{\epsilon \rightarrow 0}  \, \, \epsilon\, m(z_A,\epsilon)\,=\,0$
\cite[see  Definition  2.18  in][]{giaquinta2010}. Furthermore,  if $g  \in L_1(\mathbb{R}^d)$  
then  almost all points are Lebesgue points. This follows from the Lebesgue differentiation theorem; see Theorem  2.16 in \cite{giaquinta2010}. 
Thus, if $g  \in L_1(\mathbb{R}^d)$, then loosely speaking each term  in the right hand side of (\ref{summation1}) in the main paper  is  $O(\epsilon^{d-1})$,
for  any configuration of  points $z_A$.  
In Assumption  \ref{as:sup},   we further require  
$ S_1(f_0 \circ  h^{-1} ,\mathcal{P}_{\epsilon,  \mathcal{S} }) $  to be bounded.


Next, we show that  piecewise Lipschitz continuity  implies  Assumption  \ref{as:sup}. First, note that if $g$ is piecewise Lipschitz, then  $S_1(g,\mathcal{P}_{\epsilon,  \mathcal{S} })$,  defined in \eqref{summation1} of the paper, is bounded.  To see  this,  assume that $g$ is piecewise Lipschitz   with the set  $\mathcal{S}$.  Then 
\begin{equation}
\label{eqn:lip_s1}
S_1(g,\mathcal{P}_{\epsilon,  \mathcal{S} })   \,\leq \, \underset{A \in \mathcal{P}_{\epsilon,  \mathcal{S} } }{\sum} 
\,\, \underset{  z_{A }  \in A   }{\sup} 
\,\,\, \underset{ B_{\epsilon}(z_{A }  ) }{\int}\, \frac{\vert g(z_{A})\,-\,  g( z)\vert}{\|z_A  -  z\|_2   }\,dz 
\,\leq \,   \frac{L_0\,  \pi^{ \frac{d}{2}  } }{ \Gamma \left(  \frac{d}{2} + 1 \right)}  \,<\,  \infty,
\end{equation}
where the second-to-last inequality follows from the fact that  the volume of a $d$-dimensional ball 
of radius $\epsilon$ is $\pi^{d/2} \epsilon^d/\Gamma(d/2+1)$, as well as the fact that there are
at most on the order of $\epsilon^{-d}$ elements of $\mathcal{P}_{\epsilon,S}$.

Furthermore, with  a similar argument to that in \eqref{eqn:lip_s1},  we can show that if $g$ is piecewise Lipschitz, then  $S_2(g,\mathcal{P}_{\epsilon,  \mathcal{S} })$ is bounded. Therefore, if   $f_0 \circ  h^{-1}$ is piecewise Lipschitz, then $f_0 \circ  h^{-1}$ satisfies Assumption \ref{as:sup}.


\subsection{  Generic Functions that Do Not Satisfy Definition  \ref{def:pie_lip}}

\label{sec:not_pl}
We now  present a general condition  on $f_0$   that implies that $f_0$  is not piecewise Lipschitz. 

Let $f_0  \,:\,  [0,1]^d \,\rightarrow  \mathbb{R}$ such that   $f_0$ is differentiable in $(0,1)^d$. Suppose  that the  gradient of $f_0$ is continuous and unbounded. Then, we claim that $f_0$ is not piecewise Lipschitz. To see this,  notice that there  exists $t^{(0)} \in [0,1]^d$ such that
\[
\underset{t \rightarrow t^{(0)} }{\lim }  \, \|\nabla f_0(t)\|_{\infty} \,=\,\infty.
\]
Hence, without lost of generality 
\[
\underset{t \rightarrow t^{(0)} }{\lim }  \, \left\vert  \frac{\partial f_0(t)}{\partial t_1} \right\vert \,=\,\infty.
\]
Now, let $S \subset [0,1]^d$ satisfy all but the third  condition of Definition 1. Then there exists a positive decreasing sequence $\{\epsilon_m\}_{m=1}^{\infty}$, 
and  $t^{
	(m,1)},t^{(m,2)} \in \{(0,1)^d \backslash B_{\epsilon_m}(S)\}  \cap B_{C\epsilon_m^{1/d}}(t^{(0)})$ such that   $\|t^{(m,1)}-t^{(m,2)}\|_2 <  C \epsilon_m^{1/d} < 1$ for some  constant  $C$, $t^{(m,1)}_j  \,=\,t^{(m,2)}_j $ for $j \in   \{2, \ldots, d\}$, and such that the segment  connecting $t^{(m,1)}$ to $t^{(m,2)}$ does not contain elements in $B_{\epsilon_m}(S)$. Hence, by the mean value theorem,
\[
\frac{  \vert  f_0(t^{(m,1)}) -f_0(t^{(m,2)}) \vert }{  \| t^{(m,1)} -  t^{(m,2)}\|_2 } \,=\,  \left\vert\frac{\partial f_0(t^{(m,1,2)} )}{\partial t_1} \right\vert, 
\]  
where $t^{(m,1,2)}$ is a point in the segment connecting $t^{(m,1)}$ and $t^{(m,2)}$. By the continuity of $\nabla f_0$,  we can make the right-hand side of the previous inequality as large as desired. Therefore,   the function  $f_0$ is not piecewise  Lipschitz.

\subsection{Assumptions \ref{as:_rf}--\ref{as:sup}    and Definition  \ref{def:pie_lip}  Induce Different Function Classes  }

We start  by providing  a stronger condition that Assumption \ref{as:sup}. Specifically, assume that for small  enough $\epsilon$ it holds that
\begin{equation}
\label{eqn:condition}
\sum_{A \in  \mathcal{P}_{\epsilon, \emptyset}}    \text{Vol}(A) \,\underset{z \in  B_{2\epsilon}(A)  }{\sup}\, \| \nabla f_0(z)\|_1 \,\,< \,\,C,   
\end{equation}
for a constant $C$. Then we claim that $f_0$  satisfies  Assumption  \ref{as:sup}. To verify this, we note that by the mean value theorem,
\begin{equation}
\label{eqn:st1}
S_1(f_0, \mathcal{P}_{\epsilon,\emptyset}  )  \,\leq\, \underset{A \in\mathcal{P}_{\epsilon,\emptyset}  }{\sum} \,   C_2\text{Vol}(A) \,  \underset{z \in  B_{\epsilon}(A)}{\sup}  \| \nabla f_0(z)  \|_1
\,<\, C_2 C,
\end{equation}
for some constant $C_2 >0$.
Moreover,
\[
\begin{array}{lll}
T(f_0,z_A)   &\leq &  \underset{z \in B_{2\epsilon}(z_A) }{\sup}            \text{Vol}(B_{\epsilon}(0))   \frac{\left\| \nabla f_0(z)   \right\|_1}{\epsilon^d} \,d \| \nabla\phi \|_{\infty} \,\leq \,  \text{Vol}(A)   \underset{z \in  B_{2\epsilon}(A)  }{\sup}\frac{\left\| \nabla f_0(z)   \right\|_1}{\epsilon^d} \,d \| \nabla\phi \|_{\infty},
\end{array}
\]
which implies that for some constant $C_3 >0$
\begin{equation}
\label{eqn:st2}
S_2(f_0,\mathcal{P}_{\epsilon,\emptyset}  )\,\leq \,  C_3 \underset{A \in\mathcal{P}_{\epsilon,\emptyset}  }{\sum}\, \text{Vol}(A) \underset{z \in  B_{2\epsilon}(A)  }{\sup}\left\| \nabla f_0(z)   \right\|_1 \,< \,C_3 C.\,
\end{equation}

Therefore,  combining (\ref{eqn:st1})  and (\ref{eqn:st2}), we obtain that the condition defined in (\ref{eqn:condition}) implies Assumption \ref{as:sup}.

Next, we exploit  (\ref{eqn:condition})  to show that Assumptions  \ref{as:_rf}--\ref{as:sup} and   Definition  \ref{def:pie_lip} induce different classes of functions.

\begin{example}
	\label{ex1}
	Consider the function $f_0(x)  \,=\,  \sum_{j=1}^{d} \sqrt{x_j} $,  $x \in [0,1]^d$. Then  $f_0 $  satisfies Assumptions \ref{as:_rf}--\ref{as:sup}.   However,  $f_0$  is not  piecewise Lipschitz. 
\end{example}
To verify Example \ref{ex1}, we check that (\ref{eqn:condition}) holds for $f_0  \,:\,   [0,1]^d \rightarrow \mathbb{R} $ defined as  $f_0(x) \,:=\,  \sum_{j=1}^d  \sqrt{x_j} $.  To that end, notice that for some constant $C_4 > 0$, we have that 
\[
\begin{array}{lll}
\displaystyle \sum_{A \in  \mathcal{P}_{\epsilon, \emptyset}}    \text{Vol}(A) \,\underset{z \in  B_{2\epsilon}(A)  }{\sup}\, \| \nabla f_0(z)\|_1  &\leq & C_4 \left[ \displaystyle  \sum_{j_1=1}^{1/\epsilon } \ldots \sum_{j_d=1}^{1/\epsilon }  \epsilon^{d}  \left(   \frac{1}{\sqrt{\epsilon  j_1  }  } \,+\, \ldots\,+\,\   \frac{1}{\sqrt{\epsilon  j_d }) }  \right )    \right]
\\
& \leq   &  C_4  d   \sqrt{\epsilon} \displaystyle  \sum_{j_1=1}^{1/\epsilon }\frac{1}{\sqrt{ j_1  }  } \\
&  \leq  &C_4  d   \sqrt{\epsilon} \displaystyle  \sum_{j_1=1}^{1/\epsilon }    \frac{2}{\sqrt{ j_1  }  +  \sqrt{ j_1 -1 } }\\
& \leq &C_4  d   \sqrt{\epsilon} \displaystyle  \sum_{j_1=1}^{1/\epsilon }     2 \left[ \sqrt{ j_1  }  -  \sqrt{ j_1 -1 }    \right]\\
& \leq &C_4  d   \sqrt{\epsilon}   \sqrt{  \frac{1}{\epsilon}   }\\
& \leq &C_4  d. 
\end{array}
\]
Furthermore, it is straightforward to check that $f_0$ has bounded variation, and hence  $f_0$  satisfies Assumption \ref{as:_rf}. However, 
\[
\underset{t \rightarrow 0 }{\lim }  \, \left\vert  \frac{\partial f_0(t)}{\partial t_1} \right\vert \,=\,\infty,
\]
which by the argument in Section \ref{sec:not_pl}  implies that $f_0$  is not piecewise Lipschitz. 


\section{Outline of the Proof of Theorem \ref{thm:upper_bound} } 
\label{sec:grid_embedding}
We start by  discussing  how to embed a mesh on a $K$-NN   graph  generated under Assumptions \ref{as1}--\ref{as3}. 
This construction will then be used in the next sections  in order to derive an upper bound for the MSE of the $K$-NN-FL  estimator.

The embedding idea that  we will present  appeared   in the flow-based proof of Theorem  4 from  \cite{von2010hitting}  regarding commute distance on $K$-NN-graphs. There,  the authors  introduced the notion  of a valid grid (Definition 17). For a fixed set of design points,  a  grid graph  $G$  is  valid  if, among other things,  $G$ satisfies the following: (i) The grid width is not too small, in the sense that each cell of the grid contains at least one of the design
points.  (ii) The grid width is not too large: points in the same or neighboring cells of the grid
are always connected in the $K$-NN graph.

The   notion of a valid grid  was introduced for fixed design points, but through a minor modification, we construct a grid graph that with high probability satisfies the conditions of a valid grid from \cite{von2010hitting}. 
We now proceed to construct a grid embedding that for any signal will lead to a lower bound on the total variation along the $K$-NN graph. 

Given $N \in \mathbb{N}$,  in $ [0,1]^{d}$   we   construct a  $d$-dimensional grid graph $G_{ \text{lat} } = (V_{\text{lat}}, E_{\text{lat}})$, i.e., a lattice graph, with equal side lengths, and total number of nodes $|V_{\text{lat}} | \,=\, N^d$. Without loss of generality,  we assume that the nodes of the grid correspond to the points 
\begin{equation}
\label{eqn:lattice}
P_{\text{lat}}(N) \,\,=\,\,  \left\{ \left(\frac{i_1}{N}  \,-\,\frac{1}{2N},\ldots, \frac{i_d}{N}  \,-\,\frac{1}{2N}  \right)\,\,:\,\, i_1,\ldots,i_d \in \{1,\ldots,N\}  \right\}.
\end{equation}
Notice that $P_{\text{lat}}(N)$ is the  set of  centers of the  elements of the parition $\mathcal{P}_{N^{-1}}$, with the notation of Section \ref{sec:compact_support} in the main paper. Moreover,  $z,z^{\prime} \in P_{\text{lat}}(N)$  share an edge in the graph $G_{\text{lat}}(N)$ if only if $\|z \,\,-\,\,z^{\prime} \|_2\,\,=\,\,N^{-1}$. If the nodes corresponding to $z, z^{\prime}$ 
share an edge, then we will write  $(z,z^{\prime}) \in E_{\text{lat}}(N)$.
Note that the lattice $G_{\text{lat}}(N)$ is constructed in $[0,1]^d$  and not in the set $\mathcal{X}$. However, this lattice can be transformed into a mesh in the covariate space through the homeomorphism from Assumption \ref{as3}, by $I(N) = h^{-1}\{P_{\text{lat}}(N)\}$.  
We can use $I(N)$ to perform a quantization in the domain $\mathcal{X}$  by using the cells  associated with $I(N)$. See \cite{alamgir2014density} for more general aspects of quantizations under Assumptions \ref{as1}--\ref{as3}. 

Using  this mesh $I(N)$, which depend on the homeomorphism $h$,   for any signal   $\theta  \in \mathbb{R}^{n}$,  we can construct two vectors  denoted by
$\theta_I \in \mathbb{R}^n$ and   $\theta^I \in \mathbb{R}^{N^d}$. The former ($\theta_I$) is a signal vector that is constant within mesh cells. 
The latter ($\theta^I$) has coordinates corresponding to the different nodes of the mesh (centers of cells). 
The precise definitions of  $\theta_I$ and $\theta^I $  are given in  Section \ref{sec:quantization}.
Since  $\theta$  and $\theta_I$ have the same dimension, it is natural to  ask how these two relate to each other, at least for the purpose of understanding the empirical process  associated with the $K$-NN-FL estimator. Moreover,  given that $\theta^I \in \mathbb{R}^{N^d}$, one can try to relate the total variation of $\theta^I$  along a $d$-dimensional grid with $N^d$ nodes, with the total variation of the original signal $\theta$ along the $K$-NN graph. We proceed to establish these connections next.

\begin{lemma}
	\label{grid_embedding}
	Assume that $K$ is chosen such that $ K/\log n   \,\rightarrow  \, \infty $. Then 
	with high probability   the following holds. Under Assumptions \ref{as1}--\ref{as3},  there exists  an  $N  \asymp  (n/K)^{1/d}  $   such that for the   corresponding mesh $I(N)$ we have that   
	\begin{equation}
	\label{eqn:inequality1}
	\vert e^T( \theta \,-\, \theta_I)\vert  \,\,  \leq  \,\,  2\,\|e \|_{\infty}\, \| \nabla_{G_K} \theta  \|_1,  \,\,\,\,\,\,\,\, \forall \theta \in \mathbb{R}^{n},\\
	\end{equation}
	for all $e \in \mathbb{R}^n$.
	Moreover, 
	\begin{equation}
	\label{eqn:inequality2}
	\|D\, \theta^I \|_1    \,\, \leq  \,\,    \|\nabla_{G_K}\, \theta\|_1,\,\,\,\,\,\,\,\, \forall \theta \in \mathbb{R}^{n},
	\end{equation}
	where $D$ is the incidence matrix of a d-dimensional  grid graph $G_{grid} = (V_{grid},E_{grid})$ with $V_{grid} =  [N^d]$.
\end{lemma}

Lemma \ref{grid_embedding}  immediately  provides a path to control the empirical process associated with the $K$-NN-FL  estimator. In particular, by the basic inequality argument (see, for instance, \citealt{wang2016trend}), it is of  interest to bound the quantity  $\varepsilon^T\,(\hat{\theta}   \,-\theta^*)$. In our context, this can be done by noticing that 
\begin{equation}
\label{eqn:basic_inequality2}
\begin{array}{lll}
\frac{1}{n}\varepsilon^{T}( \hat{\theta} \,-\,\theta^*   )    &  =   &  \frac{1}{n}\varepsilon^{T}( \hat{\theta} \,-\,\hat{\theta}_I  )   \,+\,   \frac{1}{n}\varepsilon^{T}(\hat{\theta}_I  \,-\,\theta^*_I   )    \,+\,   \frac{1}{n}\varepsilon^{T}( \theta^*_I   \,-\,\theta^*   ).  \\
\end{array}
\end{equation}
Hence  in the proof of our main theorem stated in   Section \ref{sec:appendix_theorem2},  we proceed to bound each term in the right hand side of (\ref{eqn:basic_inequality2}).


Furthermore, Lemma \ref{grid_embedding}   provides  lower bounds involving $\theta_I$ and $\theta^I $, both of which depend on the homeomorphism $h$. However, we only need to specify $K$, not $h$. In other words,
we can avoid constructing the mesh $I(N)$, which would require knowledge of the unknown function $h$.

\section{Explicit Construction of Mesh for  $K$-NN Embeddings }
\label{sec:quantization}

We now describe in detail the  embedding  idea  from Section \ref{sec:grid_embedding}. 

Importantly, $ I(N)$ (defined in Section \ref{sec:grid_embedding})  can be thought of as a quantization in the domain  $\mathcal{X}$; see \cite{alamgir2014density} for more general aspects of quantizations under Assumptions \ref{as1}--\ref{as3}.  For our purposes, it will be crucial to understand the behavior of $\{x_i\}_{i=1}^n$ and their relationship to  $I(N)$. This is because we will use a grid embedding in order to analyze the behavior of the $K$-NN-FL estimator $\hat{\theta}$.  With that goal in mind, we define a collection of cells, $\{ C(x)\}_{x \in  I(N)}$, in  $\mathcal{X}$ as
\begin{equation}
\label{eqn:cells}
C(x)  \,\, =\,\,  h^{-1}\left(   \left\{ z   \in [0,1]^d \,\,:\,\,   h(x)  \,=\,\underset{x' \in P_{\text{lat}}(N) }{\arg\min } \| z - x' \|_{\infty}  \right\}     \right)    .
\end{equation}

Recall that the goal in this paper is to estimate $\theta^* \in \,\mathbb{R}^n$. However, the mesh construction $I(N)$  has $N^d$  elements, which we denote by $u_1,\ldots,u_{N^d}$. Hence, it is not immediately clear how to evaluate the total variation of $\theta^*$ along the graph corresponding to the mesh.

We first define $\theta_I$, a vector constructed from $\theta$ that incorporates information about the samples $\{x_i\}_{i=1}^{n}$  and the cells $\{ C(x)\}_{x \in  I(N)}$. For $x \in \mathcal{X}$  we write $P_I(x)$ as the point in $I(N)$ such that $x \in C(P_I(x))$. If there is more than one  such point, then we arbitrarily pick one. Then   we collapse  measurements corresponding to  different  observations $x_i$ that fall in the same cell  in  $\{ C(x)\}_{x \in  I(N)}$ into a single  value associated with the observation closest  to the center of the cell after mapping with the homeomorphism.   Thus,  for a given  $\theta \in \mathbb{R}^n$,  we define $\theta_{I} \in \mathbb{R}^{n}$ as
\begin{equation}
\label{eqn:first_proj}
(\theta_I)_{i} \,\,=\,\, \theta_j\,\,\,\,\,\,\,\text{where}\,\,\,\,\,\,  j \,\,\,=\,\,\,\underset{ l \in [n]}{\arg \min } \, \| h( P_I(x_i)) \,-\,h(x_l)\|_{\infty}.
\end{equation}

Next we construct $\theta^I$, a mapping from $\mathbb{R}^n$ to $\mathbb{R}^{N^d}$. For any   $ \theta \in \mathbb{R}^n$,  we can induce  a signal in $\mathbb{R}^{N^d}$ corresponding to the elements in $I(N)$. We write
\[
I_j \,=\, \{  i\in [n]\,:\,    P_I(x_i)\,=\, u_j \},\,\,\,\,j \in [N^d],
\]
where as before  $u_1,\ldots, u_{N^d}$ are the elements of  $I(N)$. If $I_j  \,\neq \, \emptyset$, then there exists $i_j \in I_j$ such that  $(\theta_I)_i \,\,=\,\, \theta_{i_j}$ for all $i \in I_j$. Here  $\theta_I$ is the vector defined in (\ref{eqn:first_proj}).  Using this notation, for a vector  $\theta \in \mathbb{R}^n$, we write  $\theta^{I} \,=\,( \theta_{i_1},\ldots, \theta_{i_{N^d} } )$. We use the convention  that $\theta_{i_j}   \,=\,0$ if   $I_j \,=\,  \emptyset$.

Hence, for a given  $\theta \in \mathbb{R}^n$,  we have  constructed two very intuitive signals $\theta_I \in \mathbb{R}^n $  and  $\theta^I \in \mathbb{R}^{N^d}$. The former  forces covariates $x_i$ in the same cell to take the same signal value. The latter has coordinates corresponding to the different nodes of the mesh (centers of cells).


\section{Auxiliary Lemmas for Proof of Theorem  \ref{thm:upper_bound} }
\label{sec:aux_thm1}

In several of the lemmas we will present next, we will implicitly  condition on one of the observations, for example $x_1$.  When doing so, we will exploit the fact that  the other observations are independent  as by our model  assumption $x_i \overset{\text{ind}}{\sim } p(x)$,  $i=1,\ldots,n$.   

The following lemma  is a well-known concentration inequality  for binomial random variables. It can  be found as Proposition 27 in \cite{von2010hitting}.

\begin{lemma}
	\label{lem:binom_concentration} Let  $m$  be a $\text{Binomial}(M,q)$  random variable. Then, for all $\delta \in (0,1]$,
	\[
	\begin{array}{lll}
	\mathbb{P}\left( m \,\leq \, (1-\delta)Mq  \right) &\leq &  \exp\left( -\frac{1}{3}\delta^2 M q  \right),\\
	\mathbb{P}\left( m \,\geq \, (1+\delta)Mq  \right) &\leq &  \exp\left( -\frac{1}{3}\delta^2 M q  \right).
	\end{array}
	\]
\end{lemma}

We now use Lemma \ref{lem:binom_concentration}   to  bound   the  distance  between  any design point and its  $K$-nearest neighbors.

\begin{lemma}
	\label{lem:controlling_dist}
	(\textbf{See Proposition 30 in \cite{von2010hitting}})  Denote by $R_K(x)$ the distance from  $x \in \mathcal{X}$  to its $K$th  nearest neighbor in the set $\{x_1,\ldots,x_n\}$. Setting
	\[
	\begin{array}{lll}
	R_{K,\max}   &   = &  \underset{1 \leq  i \leq n}{\max}\,\, R_{K}(x_i),  \\
	R_{K,\min}   &   = &  \underset{1 \leq  i \leq n}{\min}\,\, R_{K}(x_i),  \\
	\end{array}
	\] 
	we have that 
	\[
	\textup{pr}\left\{ a\left(\frac{K}{n}\right)^{1/d}    \,\,\leq  \,\, R_{K,\min} \,\,\leq \,\,  R_{K,\max} \,\,\leq \,\,\tilde{a} \left(\frac{K}{n}\right)^{1/d}   \right\}  \,\,\geq\,\, 1\,\,-\,\,  n\,\exp(-K/3)\,\,-\,\, n\,\exp(-K/12),
	\]
	under Assumptions \ref{as1}--\ref{as3}, where  $a \,=\, 1/(2\,c_{2,d}\,p_{\max})^{1/d} $,  and   $\tilde{a}    \,\,=\,\, 2^{1/d}/(p_{\min}\,c_{1,d})^{1/d}$.
\end{lemma}

\begin{proof}
	This proof closely follows that of Proposition 30 in \cite{von2010hitting}.
	
	First note that for any $x \in \mathcal{X}$,  by Assumptions \ref{as1}--\ref{as2} we have
	\[
	\begin{array}{l}
	\text{pr}\left\{ x_1  \in B_r(x) \right\}    \,\, = \,\, \int_{B_r(x)} p(t) \mu(dt)     \,\, \leq \,\,   p_{\max}\,\mu\{B_r(x)\} \,\,\leq \,\,  c_{2,d}\, r^d\,p_{\max} \,\,:=\,\,  \mu_{\max} \,\,<\,\,1, \\   
	\end{array}
	\]
	for small enough $r$. We also have that $R_K(x)   \,\leq \, r$  if and only if there are at least $K$  observations  $\{x_i\}$ in $B_r(x)$. Let  $V \,\sim\,  \text{Binomial}(n,\,\mu_{\max})$. Then, 
	\[
	\mathbb{P}\left(  R_K(x) \,\leq\,  r \right)\,\,\leq \,\,     \mathbb{P}\left( V \,\geq \,K \right)\,\,=\,\,   \mathbb{P}\left( V \,\geq \,2\,\mathbb{E}(V) \right), 
	\]
	where the last equality follows by choosing $a \,=\, 1/(2\,c_{2,d}\,p_{\max})^{1/d} $, and $r  \,=\,  a(K/n)^{1/d} $. Therefore, from Lemma \ref{lem:binom_concentration} we obtain
	\[
	\begin{array}{lll}
	\textup{pr}\left\{ R_{K,\min} \,\leq\,  a(K/n)^{1/d}  \right\}  &\leq  &   \textup{pr}\left\{ \exists  i\,:\,  R_K(x_i) \,\leq\,  a(K/n)^{1/d} \right\}  \\ 
	& \leq &   n\,\underset{1 \leq i \leq n}{\max}  \textup{pr}\left\{  R_K(x_i)  \leq r  \right\} \\
	& \leq  &  n\,\exp(-K/3).
	\end{array}
	\]
	Furthermore,
	\[
	\textup{pr}\left\{  x_1  \in B_r(x) \right\}   \,\, = \,\, \int_{B_r(x)} p(t) \mu(dt)   \,\,\geq \,\, p_{\min}\mu\{B_r(x)\}\,\,\geq \,\,  c_{1,d}\, r^d\,p_{\min} \,\,:=\,\,  \mu_{\min} \,\,>\,\,0, 
	\]
	and we arrive with a similar argument to
	\begin{equation}
	\label{eqn:rkmax}
	\textup{pr}\left\{  R_{K,\max}\,>\,  \tilde{a}(K/n)^{1/d}  \right\} \,\,\leq  \,\, n\,\exp(-K/12), 
	\end{equation}
	where $\tilde{a}    \,\,=\,\, 2^{1/d}/(p_{\min}\,c_{1,d})^{1/d}$.
\end{proof}

The upper bound in Lemma  \ref{lem:controlling_dist}   allows us  to control the maximum distance between  $x_i$  and $x_j$ whenever  they are connected in the $K$-NN  graph. This maximum distance   scales as $(K/n)^{1/d}$.  The lower bound, on the other hand,  prevents  $x_i$  from  being arbitrarily   close to its $K$th nearest neighbor. These properties are particularly  important as they will be used to   characterize the penalty  $\|  \nabla_G \theta^* \|_1$.

As explained in \cite{wang2016trend}, there are different strategies for proving convergence rates in generalized lasso problems such as  (\ref{eq:fusedlasso}) from the main paper. Our approach here is similar in spirit to \cite{padilla2016dfs}, and it is based on considering a  lower bound for the penalty function  induced by the $K$-NN graph. This lower bound will arise by constructing a signal over the grid graph induced by the cells $\{ C(x)\}_{x \in  I(N)}$ defined in \eqref{eqn:cells}. Towards that  end, we provide the following lemma characterizing the minimum and maximum  number of observations  $\{x_i\}$   that fall within each cell $C(x)$. With an abuse of notation, we write $\vert C(x) \vert  \,=\,  \vert\{i \in [n] \,:\,  x_i \in C(x)  \} \vert$.

We now  present a   result related to  Proposition  28 in  \cite{von2010hitting}.

\begin{lemma}
	\label{lem:controlling_counts}
	Assume that $N$  in the construction of $P_{\text{lat}}(N)$ defined in \eqref{eqn:lattice}  is chosen as
	\[
	N\,\,=\,\,  \ceil[\Bigg]{  \frac{   3\,\sqrt{d}\,\left( 2\,c_{2,d}\,p_{\max} \right)^{1/d} \,n^{1/d}  }{L_{\min}\,K^{1/d}} }.
	\]
	Then there exist positive constants $\tilde{b}_1$ and $\tilde{b}_2$ depending on $L_{\min}$,  $L_{\max}$, $d$, $p_{\min}$, $p_{\max}$, $c_{1,d}$, and $c_{2,d} $ defined in Assumptions \ref{as1}--\ref{as3}, such that
	\begin{equation}
	\begin{array}{lll}
	\textup{pr}\bigg\{\underset{x  \in  I(N) }{\max}\vert C(x)\vert   \,\,\geq \,\, (1+\delta)\,c_{2,d}\,\tilde{b}_1\,K  \bigg\}    &  \leq & N^d \exp\left(  -\frac{1}{3}\delta^2\,\tilde{b}_2\,c_{1,d}\,K   \right),\\
	\vspace{0.3in}
	\textup{pr}\bigg\{ \underset{x  \in  I(N) }{\min}\vert C(x)\vert   \,\,\leq \,\, (1-\delta)\,c_{1,d}\,\tilde{b}_2\,K  \bigg\}  &\leq  & N^d \exp\left(  -\frac{1}{3}\delta^2\,\tilde{b}_2\,c_{1,d}\,K   \right),
	\end{array}
	\end{equation}
	for all $\delta \in (0,1)$,  with  $a  \,\,=\,\,1/(2\,c_{2,d}\,p_{\max})^{1/d}$, and $\tilde{a}    \,\,=\,\, 2^{1/d}/(p_{\min}\,c_{1,d})^{1/d}$. Moreover,  the symmetric $K$-NN  graph     has maximum degree $d_{\max}$  satisfying
	\[
	\textup{pr}\left(  d_{\max}    \,\geq \,  \frac{3}{2}\,p_{\max}\, c_{2,d}  \,\tilde{a}^d\, K     \right) \,\,\leq \,\,    n\, \left\{ \exp\left(-\frac{K}{12}\right)  \, +\,   \exp\left(-\frac{p_{\min} c_{1,d}\tilde{a}^d  K }{24}\right)    \right\}.
	\] 
	
	Define the event $\Omega$  as:  ``If  $x_i  \in C(x_i^{\prime})$ and $x_j \in C(x_j^{\prime})$  for $x_i^{\prime},x_j^{\prime} \in I(N)$ with $\|h(x_i^{\prime}) \,-\,h(x_j^{\prime}) \|_2 \,\leq\,  \,N^{-1} $, then $x_i$ and $x_j$  are connected in the $K$-NN graph". Then,  
	$$ \textup{pr}\left(\Omega\right)  \,\, \geq \,\,  1\,\,-\,\,  n\,\exp(-K/3).$$
	
\end{lemma}

\begin{proof}
	Let  $x_i$  and $x_j$ such that $x_i  \in C(x_i^{\prime})$ and $x_j \in C(x_j^{\prime})$  for $x_i^{\prime},x_j^{\prime} \in I(N)$ with $\|h(x_i^{\prime}) \,-\,h(x_j^{\prime}) \|_2 \,\leq\,  \,N^{-1} $. Then the following holds using Assumption \ref{as3}:
	\[
	\def\arraystretch{1.8}
	\begin{array}{lll}
	d_{\mathcal{X}}(x_i,x_j)   & \leq & L_{\min}^{-1}\|  h(x_i)  \,-\,h(x_j) \|_2\\
	& \leq &    L_{\min}^{-1}  \,  \sqrt{d}\, \|  h(x_i)  \,-\,h(x_j) \|_{\infty}\\
	&\leq  &   L_{\min}^{-1}  \,  \sqrt{d} \left\{  \|  h(x_i)  \,-\,h(x_i^{\prime}) \|_{\infty} \,\,+\,\, \|  h(x_i^{\prime})  \,-\,h(x_j^{\prime}) \|_{\infty}\,\,+\,\, \|  h(x_j^{\prime})  \,-\,h(x_j) \|_{\infty} \right\}  \\
	&  < & 3\,L_{\min}^{-1}  \,  \sqrt{d}\, N^{-1} \\
	& \leq &     a\left(\frac{K}{n}\right)^{1/d}
	\end{array}
	\] 
	with   $a  \,\,=\,\,1/(2\,c_{2,d}\,p_{\max})^{1/d} $, where the first inequality follows from Assumption \ref{as3}. Therefore, as in the proof of Lemma \ref{lem:controlling_dist},
	\[
	\textup{pr}\left(\Omega\right) \,\,\geq \,\,\,\textup{pr}\left\{ a\left(\frac{K}{n}\right)^{1/d}    \leq  R_{K,\min}  \right\}\, \geq\,  1\,\,-\,\,  n\,\exp(-K/3).
	\] 
	Next we proceed to  derive an upper bound on the counts $\{C(x)\}$.	Assume that $x \in h^{-1}\{P_{\text{lat}}(N)\}$,  and $x^{\prime } \in C(x)$. Then,
	\[
	\def\arraystretch{1.8}
	\begin{array}{lll}
	d_{\mathcal{X}}(x,x^{\prime})  & \leq &  \frac{1}{L_{\min}}  \|   h(x) - h(x^{\prime}) \|_2 \\
	& \leq &  \frac{ \sqrt{d}  }{L_{\min}}  \|   h(x) - h(x^{\prime}) \|_{\infty} \\
	& \leq & \frac{\sqrt{d}}{2\,  L_{\min}\,N }\\
	&   \leq &  \frac{a}{6} \left( \frac{K}{n} \right)^{1/d}   \\ 
	& =&  \frac{(\tilde{b}_1)^{1/d} }{p_{\max} ^{1/d} }\,  \left( \frac{K}{n} \right)^{1/d},
	\end{array}
	\]
	where the first inequality follows from Assumption \ref{as3}, the second from the definition of $P_{\text{lat}}(N)$,  the third one from the choice of $N$, and $\tilde{b}_1$ is an appropriate  constant. Therefore,
	\begin{equation}
	\label{one_side}
	C(x)   \,\subset\, B_{  \frac{ (\tilde{b}_1)^{1/d}}{p_{\max} ^{1/d} } \,  \left( \frac{K}{n} \right)^{1/d}  }(x).
	\end{equation}
	On the other hand, if $\tilde{b}_2$  is such that 
	\[
	\frac{(\tilde{b}_2)^{1/d}}{p_{\min}^{1/d} }  \,\,\leq  \,\,  \frac{a}{2 \,L_{\max}}\,\frac{L_{\min} }{3\,   \sqrt{d} }\,,
	\]
	then if 
	\[
	d_{\mathcal{X}}(x,x^{\prime})\,\,\leq \,\,  \frac{  (\tilde{b}_2)^{1/d} }{p_{\min}^{1/d} }\left(\frac{K}{n}\right)^{1/d}, 
	\]
	we have that  by Assumption \ref{as3},   for large enough $n$
	\[
	\begin{array}{lll}
	\|   h(x) \,-\,  h(x^{\prime}) \|_{\infty}    &  \leq & \frac{1}{2\,N}.
	\end{array}
	\]
	And so,
	\begin{equation}
	\label{second_side}
	B_{ \frac{ (\tilde{b}_2)^{1/d }}{ p_{\min}^{1/d} }\,  \left( \frac{K}{n} \right)^{1/d}  }(x)   \,\subset\, C(x).
	\end{equation}
	
	In consequence,
	\[
	\def\arraystretch{1.8}
	\begin{array}{lll}
	\textup{pr}\left\{ \underset{x  \in   I(N) }{\max}\vert C(x)\vert   \,\,\geq \,\, (1+\delta)\,c_{2,d}\,\tilde{b}_1\,K  \right\} & \leq  & 
	\underset{x  \in  I(N)  }{\sum} \,\textup{pr}\left\{\vert C(x)\vert   \,\,\geq \,\, (1+\delta)\,c_{2,d}\,\tilde{b}_1\,K  \right\}\\
	& \leq &\underset{x  \in   I(N)}{\sum} \,\textup{pr}\left[\vert C(x)\vert   \,\,\geq \,\, (1+\delta)\,n\,  \textup{pr}\{x_1 \in C(x)\} \right]\\
	& \leq & \underset{x  \in  I(N)}{\sum} \, \exp\left(  -\frac{1}{3}\delta^2\,\tilde{b}_2\,c_{1,d}\,K   \right) \\
	& = & N^d\,\exp\left(  -\frac{1}{3}\delta^2\,\tilde{b}_2^d\,c_{1,d}\,K   \right), \\
	\end{array}
	\]	
	where the first inequality follows from a union bound, the second from (\ref{one_side}), and the third  from (\ref{second_side}) combined with Lemma \ref{lem:binom_concentration}.
	
	On the other hand, with a similar argument, we have that
	\[
	\begin{array}{lll}
	\textup{pr}\left\{ \underset{x  \in   I(N) }{\min}\vert C(x)\vert   \,\,\leq \,\, (1-\delta)\,c_{1,d}\,\tilde{b}_2\,K  \right\}  & \leq &
	\underset{x  \in  I(N)}{\sum} \, 	\textup{pr}\left\{\vert C(x)\vert   \,\,\leq \,\, (1-\delta)\,c_{1,d}\,\tilde{b}_2\,K  \right\}\\
	&\leq  &\underset{x  \in   I(N)}{\sum} \,	\textup{pr}\left[\vert C(x)\vert   \,\,\leq \,\, (1-\delta)\,n\,  \textup{pr}\{x_1 \in C(x)\} \right]\\
	&\leq  &    N^d\,\exp\left(  -\frac{1}{3}\delta^2\,\tilde{b}_2\,c_{1,d}\,K   \right).
	\end{array}
	\]
	Next  we proceed to  find an upper bound on the maximum degree of the $K$-NN  graph.  We start by defining the sets
	\[
	\begin{array}{lll}
	B_i(x)    \,\, = \,\, \left\{   j   \in [n]\backslash  \{i\}    \,\,:\,\,      x_j \in  B_{  \tilde{a} (K/n)^{1/d    }}(x)      \right\},\,\,\,\,\,\,	   B_i   \,\, = \,\, \left\{   j   \in [n]\backslash  \{i\}    \,\,:\,\,      x_j \in  B_{  \tilde{a} (K/n)^{1/d    }}(x_i)      \right\},	       
	\end{array} 
	\]
	for $i \in [n]$, and $x \in \mathcal{X}$, and  where $\tilde{a}$ is given as in Lemma  \ref{lem:controlling_dist}. Then,
	\[
	\vert B_i(x) \vert    \,\,\,\sim \,\,  \text{Binomial}\left[ n-1, \textup{pr}\{x_1 \in  B_{  \tilde{a} (K/n)^{1/d    }}(x) \}  \right],
	\]
	where
	\[
	p_{\min} \,c_{1,d}\, \tilde{a}^d   \frac{K}{n}\,\,\leq \,\,  \textup{pr}\{x_1 \in  B_{  \tilde{a} (K/n)^{1/d    }}(x) \}  \,\,\leq \,\,   p_{\max} \,c_{2,d}\, \tilde{a}^d   \frac{K}{n},   
	\]
	which implies by Lemma \ref{lem:binom_concentration}  that
	\[
	\textup{pr}\left\{    \vert B_i(x) \vert    \, \geq \, \frac{3}{2} p_{\max} \,c_{2,d}\, \tilde{a}^d \,K \right\} \, \,\leq \,\,\textup{pr}\left[    \vert B_i(x) \vert    \, \geq \,    \frac{3}{2}\textup{pr}\{x_1 \in  B_{  \tilde{a} (K/n)^{1/d    }}(x) \}(n-1)   \right]   \,\,\leq \,\,  \exp\left(  -\frac{ p_{\min} \,c_{1,d}\, \tilde{a}^d }{24} K \right).
	\]
	Hence,
	\[
	\textup{pr}\left\{   \vert B_i \vert    \, \geq \,   \frac{3}{2}p_{\max} \,c_{2,d}\, \tilde{a}^d \,K\right\}    \,\,= \,\, \underset{\mathcal{X}}{\int}    \textup{pr}\left\{\vert B_i(x) \vert    \, \geq \,   \frac{3}{2}p_{\max} \,c_{2,d}\, \tilde{a}^d \,K\right\} \,  p(x) \mu(dx)  \,\,\leq \,\,  \exp\left(  -\frac{ p_{\min} \,c_{1,d}\, \tilde{a}^d }{24} K \right).
	\]
	Therefore,  if $d_i$  is the degree associated with $x_i$  then
	\[
	\begin{array}{lll}
	\textup{pr}\left(  d_i    \, \leq \,   \frac{3}{2}p_{\max} \,c_{2,d}\, \tilde{a}^d \,K  \right)  & \geq & \textup{pr}\left[  \left\vert \{ j  
	\in [n] \backslash \{i\} \,\,:\,\,   d_{\mathcal{X}}(x_i  ,x_j)   \leq  R_{K,\max}  \}  \right\vert   \, \leq \,   \frac{3}{2}p_{\max} \,c_{2,d}\, \tilde{a}^d \,K  \right] \\
	& \geq  & \textup{pr}\big[  \left\vert \{ j  
	\in [n] \backslash \{i\} \,\,:\,\,   d_{\mathcal{X}}(x_i  ,x_j)   \leq  R_{K,\max}  \}  \right\vert   \, \leq \,   \frac{3}{2}p_{\max} \,c_{2,d}\, \tilde{a}^d \,K,\\
	& &   R_{K,\max}  \leq \tilde{a}\left(\frac{K}{n}\right)^{1/d} \big]\\
	&\geq &  \textup{pr}\big[  \left\vert \{ j  
	\in [n] \backslash \{i\} \,\,:\,\,   d_{\mathcal{X}}(x_i  ,x_j)   \leq  \tilde{a}(K/n)^{1/d} \}  \right\vert   \, \leq \,   \frac{3}{2}p_{\max} \,c_{2,d}\, \tilde{a}^d \,K,\\
	& &   R_{K,\max}  \leq \tilde{a}\left\{\frac{K}{n}\right\}^{1/d} \big]\\
	& \geq &   1 \,\,-\,\,  \textup{pr}\left\{ R_{K,\max}  > \tilde{a}\left(\frac{K}{n}\right)^{1/d}\right\}  \,\,-\,\, \textup{pr}\left\{  \vert B_i \vert    \, > \,   \frac{3}{2}p_{\max} \,c_{2,d}\, \tilde{a}^d \,K\right\},
	\end{array}
	\]
	and the claim follows from the previous inequality and  Equation (\ref{eqn:rkmax}).
\end{proof}

As stated in Lemma \ref{lem:controlling_counts}, the number of observations $x_i$ that fall within each cell $C(x)$ scales like $K$,  with high probability. We will exploit this fact  next in order to obtain an upper bound on the MSE.


\begin{lemma}
	\label{embedding}
	Assume that the event  $\Omega$ from  Lemma \ref{lem:controlling_counts} intersected with
	\begin{equation}
	\label{mesh_lower_bound}
	\underset{x  \in  I(N) }{\min}\vert C(x)\vert   \,\,\geq \,\, \frac{1}{2}\,c_{1,d}\,\tilde{b}_2\,K 
	\end{equation}
	holds.	Define $N$ as in that lemma  and let  $I(N)$ be  the corresponding mesh.
	Then for all $e \in \mathbb{R}^n$,  it holds that  
	\begin{equation}
	\label{eqn:inequality1}
	\vert e^T( \theta \,-\, \theta_I)\vert  \,\,  \leq  \,\,  2\,\|e \|_{\infty}\, \| \nabla_{G_K} \theta  \|_1,  \,\,\,\,\,\,\,\, \forall \theta \in \mathbb{R}^{n}.\\
	\end{equation}
	Moreover, 
	\begin{equation}
	\label{eqn:inequality2}
	\|D\, \theta^I \|_1    \,\, \leq  \,\,    \|\nabla_{G_K}\, \theta\|_1,\,\,\,\,\,\,\,\, \forall \theta \in \mathbb{R}^{n},
	\end{equation}
	where $D$ is the incidence matrix of a d-dimensional  grid graph $G_{grid} = (V_{grid},E_{grid})$ with $V_{grid} =  [N^d]$, where $(l,l^{\prime}) \in E_{grid} $ if and only if  
	$$\|   h\{P_I(x_{i_l})\}   \,\,-\,\,  h\{P_I(x_{i_{l ^{\prime} }})\}\|_2    \,\,=\,\,\frac{1}{\,N}.$$
	Here we use the notation from Section \ref{sec:quantization}.
\end{lemma}

\begin{proof}
	We start by introducing the notation $x_i^{\prime} \,\,=\,\,   P_I(x_i)$. 	To prove (\ref{eqn:inequality1}) we proceed in cases.
	
	\smallskip\smallskip
	\noindent
	{\bf Case 1.} If $(\theta_I)_1 \,\,=\,\, \theta_1$, then clearly $ \vert e_1 \vert \, \vert \theta_1 \,\,-\,\ (\theta_I)_1 \vert \,\,=\,\, 0 $. \\
	\smallskip\smallskip
	\noindent
	{\bf Case 2.} If $(\theta_I)_1 \,\,=\,\, \theta_i $, for $i \neq 1$, then 
	$$ \| h(x_1^{\prime}) \,\,-\,\,h(x_i ) \|_{\infty} \,\,\leq \,\,\| h(x_1^{\prime}) \,\,-\,\,h(x_1)  \|_{\infty}\,\,\leq \,\,\frac{1}{2\,N}.$$ 
	Thus, $x^{\prime}_1\,\,=\,\, x_i^{\prime}$,    and so  $(1,i) \in  E_K$ by the assumption that the event $\Omega$ holds.
	
	Therefore for every $i \in \{1,\ldots,n\}$,   there exists $j_i \in [n]$ such that  $(\theta_I)_i \,\,=\,\, \theta_{j_i}  $  and either $i = j_i$  or $(i, j_i) \in E_K$. Hence,
	\[
	\begin{array}{lll}
	\vert e^T( \theta \,-\, \theta_I)\vert  &  \leq  &   \sum_{i=1}^{n}    \,\vert e_i\vert\,\vert \theta_i \,\,-\,\, \theta_{j_i}  \vert\\
	& \leq  &   2\,\|e\|_{\infty}\, \| \nabla_{G_K} \hat{\theta}  \|_1.
	\end{array}
	\]	
	
	To verify (\ref{eqn:inequality2}), we  observe that
	\begin{equation}
	\begin{array}{lll}
	\|D\, \theta^I \|_1    & = &   \underset{(l,l^{\prime}) \in E_{grid}}{\sum }  \, \left \vert \theta_{i_l} \,-\, \theta_{i_{l^{\prime} }} \right \vert.\\
	\end{array}
	\end{equation}
	Now, if $(l,l^{\prime})  \,\in \,  E_{grid}$,   then $x_{i_l}$ and $x_{i_{l^{\prime}} }$ are in  neighboring cells in $I$. This implies that $(i_l,i_{l^{\prime}})$ is an edge in the $K$-NN graph. Thus, every edge in the grid graph $G_{grid}$  corresponds to an edge in the $K$-NN graph and the mapping is injective. Note that here we have used the fact that  (\ref{mesh_lower_bound}) ensures  that every cell has at least one point, provided that $K$ is large enough. The claim follows.
\end{proof}

\begin{lemma}
	\label{lem:empirical_process_1}
	With the notation from Lemma \ref{embedding}, we have that 
	\[
	\varepsilon^T( \hat{\theta} \,-\, \theta^*)   \leq  2\,\|\varepsilon \|_{\infty}\,\left( \| \nabla_{G_K} \theta^*  \|_1  \,+\,\| \nabla_{G_K} \hat{\theta}  \|_1 \right)\,\,+\,\, \varepsilon^T( \hat{\theta}_I \,-\, \theta_I^* ),
	\]
	on the event $\Omega$.
\end{lemma}

\begin{proof}
	Let us assume that event $\Omega$ happens. Then we observe that
	\begin{equation}
	\label{eqn:decomposition}
	\varepsilon^T( \hat{\theta} \,-\, \theta^* )\,\,\,=\,\,\,\varepsilon^T( \hat{\theta} \,-\, \hat{\theta}_I )\,\,+\,\,\varepsilon^T( \hat{\theta}_I \,-\, \theta_I^* )\,\,\,+\,\,\,\varepsilon^T( \theta_I^* \,-\, \theta^*),
	\end{equation}
	and the claim follows by Lemma \ref{embedding}. 
	
\end{proof}

\begin{lemma}
	\label{lem:empirical_process_2}
	With the notation  from Lemma \ref{lem:empirical_process_1}, on the event $\Omega$, we have that 
	\[
	\varepsilon^T( \hat{\theta}_I \,-\, \theta_I^*)   \,\,\leq \,\, \underset{u \in I}{\max } \sqrt{\vert C(u) \vert}\,\left(  \| \Pi \tilde{\varepsilon} \|_2\, \ \|\hat{\theta}  \,-\,\theta^{*}   \|_2     \,\,+\,\,   \|(D^{+})^T\,\tilde{\varepsilon}\|_{\infty}\, \Big[\|\nabla_{G_K} \hat{\theta}\|_1  \,+\,\|\nabla_{G_K}\,\theta^{*}   \|_1 \Big]   \right),
	\]
	where $\tilde{\varepsilon}   \in \mathbb{R}^{N^d}$ is a mean zero vector whose coordinates are independent and sub-Gaussian with the same constants  as in (\ref{eqn:sub_gaussian_errors}).
	Here, $\Pi$ is the orthogonal projection onto the span of $\boldsymbol{1}   \in \mathbb{R}^{N^d}$, and $D^+$  is the pseudo-inverse of the incidence matrix $D$  from Lemma \ref{embedding}.   
\end{lemma}
\begin{proof}
	Here we use the notation from the proof of  Lemma \ref{embedding}. Then,
	$$
	\begin{array}{lll}
	\varepsilon^T( \hat{\theta}_I \,-\, \theta_I^*)   &   = &   \displaystyle \sum_{j=1}^{N^d } \underset{l \in I_j}{\sum} \, \varepsilon_l \,\left(  \hat{\theta}_{i_j}\,\,-\,\, \theta^*_{i_j} \right) \,\,=\,\,\displaystyle \underset{u \in I}{\max }\,  \vert C(u) \vert^{\frac{1}{2}}   \,   \tilde{\varepsilon}^T(\hat{\theta}^I \,-\, \theta^{*,I} )   \\
	\end{array} 
	$$
	where 
	\[
	\tilde{\varepsilon}_ j  \,\,=\,\, \left\{ \underset{u \in I}{\max }  \vert C(u) \vert   \right\}^{-1/2}\, \underset{ l \, \in   I_j   }{\sum }\, \varepsilon_l.
	\]
	Clearly, the 
	$\tilde{\varepsilon}_1, \ldots, \tilde{\varepsilon}_{N^d }$  are independent given $\Omega$, and are also sub-Gaussian with the same constants as the original errors $\varepsilon_1,\ldots , \varepsilon_n$.
	
	Moreover,  let $\Pi$ be the orthogonal projection onto the span of $\boldsymbol{1}   \in \mathbb{R}^{N^d}$. Then, by H$\ddot{o}$lder's inequality, and  by the triangle inequality,
	\begin{equation}
	\label{eqn:art_one} 
	\begin{array}{lll}
	\varepsilon^T( \hat{\theta}_I \,-\, \theta_I^*)   &   \leq  & \left\{\underset{u \in I}{\max }  \vert C(u) \vert   \right\}^{1/2}\left\{  \| \Pi \tilde{\varepsilon} \|_2\, \ \|\hat{\theta}^I  \,-\,\theta^{*,I}   \|_2     \,\,+\,\,   \|(D^{+})^T\,\tilde{\varepsilon}\|_{\infty}\,\|D( \hat{\theta}^I  \,-\,\theta^{*,I}    )\|_1    \right\}\\
	&  \leq &  \left\{ \underset{u \in I}{\max }  \vert C(u) \vert   \right\}^{1/2}\left\{  \| \Pi \tilde{\varepsilon} \|_2\, \ \|\hat{\theta}^I  \,-\,\theta^{*,I}   \|_2     \,\,+\,\,   \|(D^{+})^T\,\tilde{\varepsilon}\|_{\infty}\, \Big(\|D \hat{\theta}^I\|_1  \,+\,\|D\,\theta^{*,I}   \|_1 \Big)   \right\}.
	\end{array} 
	\end{equation}
	Next we observe that 
	\begin{equation}
	\label{eqn:partA}
	\begin{array}{lll}
	\displaystyle	\|\hat{\theta}^I  \,-\,\theta^{*,I}   \|_2    \,\, =  \,\,   \left\{\sum_{j = 1}^{ N^d }  \left(\hat{\theta}_{i_j} \,-\,\theta^*_{i_j} \right)^2  \right\}^{1/2}
	\,\,  \leq \,\,\left\{\sum_{i = 1}^{n}  \left(\hat{\theta}_{i} \,-\,\theta^*_{i} \right)^2  \right\}^{1/2}.
	\end{array}
	\end{equation}
	Therefore, combining (\ref{eqn:art_one}), (\ref{eqn:partA}) and Lemma \ref{embedding},  we arrive at
	\[
	\varepsilon^T( \hat{\theta}_I \,-\, \theta_I^*)   \,\,\leq \,\, \underset{u \in I}{\max } \vert C(u) \vert^{ \frac{1}{2} }  \,\left\{  \| \Pi \tilde{\varepsilon} \|_2\, \ \|\hat{\theta}  \,-\,\theta^{*}   \|_2     \,\,+\,\,   \|(D^{+})^T\,\tilde{\varepsilon}\|_{\infty}\, \Big(\|\nabla_{G_K} \hat{\theta}\|_1  \,+\,\|\nabla_{G_K}\,\theta^{*}   \|_1 \Big)   \right\}.
	\]
	
\end{proof}


The following lemma results from a similar argument as the proof of Theorem 2 in \cite{hutter2016optimal}.

\begin{lemma}
	\label{lem:empirical_process_3}
	Let $d>1$ (recall  Assumptions  \ref{as1}--\ref{as3} ) and let $\delta  > 0$.	With the notation from the previous lemma, given the event $\Omega$, we have that
	\begin{equation}
	\begin{array}{lll}
	\varepsilon^T( \hat{\theta}_I \,-\, \theta_I^*)  &\leq  & \{(1+\delta)\,c_{2,d}\,\tilde{b}_1^d\,K\}^{\frac{1}{2} } \,\Bigg(2\,\sigma\,\left\{    2\,\log \left( \frac{e}{\delta} \right)  \right\}^{ \frac{1}{2}  }\, \|\hat{\theta}  \,-\,\theta^{*}   \|_2 \,\,+\,\,   \sigma \,C_1(d)\, \left[\, \frac{2\,\log n}{d}\,\log\left\{ \frac{ C_2(p_{\max},L_{\min},d)\, n}{K\,\delta}\right\}\right]^{ \frac{1}{2} }  \\
	& &\,\,\,\,\,\,\,\,\,\,\,\,\,\,\,\,\,\,\,\,\,\,\,\,\,\,\,\,\,\,\,\,\,\,\,\,\,\,\,\,\,\,\,\,\,\,\,\,\,\cdot \Big(\|\nabla_{G_K} \hat{\theta}\|_1  \,+\,\|\nabla_{G_K}\,\theta^{*}   \|_1 \Big) \Bigg)    \\ 
	\end{array}
	\end{equation}
	with probability at least
	\[
	1\,\,-\,\,2\delta\,\, \,\,-\,\, \frac{n}{K \, \tilde{a}^d\, L_{\max}^d  } \exp\left(  -\frac{1}{3}\delta^2\,\tilde{b}_2\,c_{1,d}\,K   \right) \,\,-\,\, n\,\exp(-K/3),  
	\]
	where $C_1(d) \,>\, 0 $ is a constant depending on $d$, and  $C_2(p_{\max},L_{\min},d)\, >\, 0 $ is another constant depending on $p_{\max}$, $L_{\min}$,  and $d$.
\end{lemma}

\begin{proof}
	First,  as in the proof of Theorem 2 from \cite{hutter2016optimal}, and our choice of $N$, we obtain that given $\Omega$,
	\begin{equation}
	\label{enq:from_hutter}
	\begin{array}{lll}
	\|  (D^{+})^{T} \tilde{\varepsilon} \|_{\infty}    & \leq &    \sigma \,C_1(d)\,\left[ \frac{2 \log n}{d}\,\log\left\{\frac{ C_2(p_{\max},L_{\min},d)\, n}{K\,\delta} \right\}  \right]^{  \frac{1}{2}  } , \\
	\| \Pi \tilde{\varepsilon} \|_2  & \leq  &  2\,\sigma\,\left\{    2\,\log \left( \frac{e}{\delta} \right)  \right\}^{ \frac{1}{2}  }
	\end{array}
	\end{equation}
	with probability at least $1 \,-\, 2\,\delta$, where  $C(d)$ is  constant depending on $d$, and $ C_2(p_{\max},L_{\min},d)$  is a constant depending on $p_{\max}$, $L_{\min}$ and $d$. 
	
	On the other hand, by Lemma \ref{lem:controlling_counts},  given the event $\Omega$ we have
	\[
	\underset{x  \in   h^{-1}(P_{\text{lat}}) }{\max} \vert C(x)\vert^{ \frac{1}{2}  }   \,\,\,\leq  \,\,\,  \{(1+\delta)\,c_{2,d}\,\tilde{b}_1\,K\}^{ \frac{1}{2}  }
	\]
	with probability at least 
	\begin{equation}
	\label{eqn:comb}
	1  \,\,-\,\, N^d \exp\left(  -\frac{1}{3}\delta^2\,\tilde{b}_2\,c_{1,d}\,K   \right) \,\,-\,\, n\,\exp(-K/3).
	\end{equation}
	Therefore,  combining (\ref{enq:from_hutter})  and (\ref{eqn:comb}),  the result follows.
\end{proof}

\section{Proof of Theorem  \ref{thm:upper_bound} }
\label{sec:appendix_theorem2}

Instead of proving Theorem  \ref{thm:upper_bound},  we will prove a more general result that holds for a general choice of $K$. This is given next. The corresponding result for $\epsilon$-NN-FL  can be obtained with a similar  argument.

\begin{theorem}
	\label{thm:knn}
	There exist  constants  $C_1(d)$ and $ C_2(p_{\max},L_{\min},d)$, depending on $d$, $p_{\max}$, and $L_{\min}$, such that with the notation from  Lemmas \ref{lem:controlling_dist}  and  \ref{lem:controlling_counts}, if  $\lambda$ is chosen as
	\[
	\lambda   \, = \,  \sigma \,C_1(d)\,\left[ (1+\delta)\,c_{2,d}\,\tilde{b}_1\,\,\frac{2\,\log n}{d}\,\log\left\{ \frac{ C_2(p_{\max},L_{\min},d)\, n}{K\,\delta} \right\}\,K\right]  \,\,+  \,\,  8\sigma \left(\log n\right)^{ \frac{1}{2}   },
	\]
	then the $K$-NN-FL  estimator $\hat{\theta}$ \eqref{eqn:fl}  satisfies 
	\[
	\begin{array}{lll}
	\| \hat{\theta}  \,-\, \theta^*  \|_n^2      & \leq &  \|\nabla_{G_K}\,\theta^{*}   \|_1\,\Big(  \frac{4\,\sigma \,C_1(d)}{n}\,\left[(1+\delta)\,c_{2,d}\,\tilde{b}_1\,\,\frac{2\,\log n}{d}\,\log\left\{ \frac{ C_2(p_{\max},L_{\min},d)\, n}{K\,\delta}\right\}\,K\right]^{1/2}    \\ 
	&&\,\,+\,\, \frac{32\, \left(\log n\right)    }{n}\Big)  \,\,+\,\,\frac{16\,\sigma^2\,(1+\delta)\,c_{2,d}\,\tilde{b}_1\,K\,\log \left( \frac{e}{\delta} \right) }{n},
	\end{array}
	\]
	with probability at least $\eta_{n}\,\{1\,\,-\,\,  n\,\exp(-K/3)\}$. Here,
	\[
	\eta_{n} \,\,=\,\,1\,\,-\,\,2\delta\,\, \,\,-\,\, N^d \exp\left(  -\frac{1}{3}\delta^2\,\tilde{b}_2\,c_{1,d}\,K   \right) \,\,-\,\, n\,\exp(-K/3)  \,\,-\,\,\frac{C}{n^7},
	\] 
	where $C$  is  the constant in  (\ref{eqn:sub_gaussian_errors}),  and  $N$  is given as in Lemma \ref{lem:controlling_counts}.
	Consequently, taking $K \,\asymp\,  \log^{1\,+\,2r} n $  we obtain the result in Theorem  \ref{thm:upper_bound}.
\end{theorem}

\begin{proof}
	We notice by the basic inequality argument,  see for instance \cite{wang2016trend}, that
	\begin{equation}
	\label{eqn:basic_inequality}
	\begin{array}{lll}
	\frac{1}{2}\| \hat{\theta}  \,-\, \theta^*  \|_n^2   & \leq & \frac{1}{n}\Bigg\{ \varepsilon^T\,(\hat{\theta}\,-\, \theta^*)\,\,+\,\, \lambda \left( -\|\nabla_{G_K} \hat{\theta}\|_1  \,+\,\|\nabla_{G_K}\,\theta^{*}   \|_1 \right)    \Bigg\}.\\
	\end{array}
	\end{equation}
	On the other hand, by (\ref{eqn:sub_gaussian_errors})  in the paper, and a union bound,
	\begin{equation}
	\label{eqn:maximal_inequality}
	\textup{pr}\left\{ \underset{1 \leq i \leq n}{\max} \vert \varepsilon_i \vert   \,>\,  4 \sigma \left(\log n\right)^{ \frac{1}{2}  }   \,\,\Big|\,\, \Omega \right\}\,=\,  \textup{pr}\left\{   \underset{1 \leq i \leq n}{\max} \vert \varepsilon_i \vert   \,>\,  4\sigma \left(\log n\right)^{ \frac{1}{2}  }   \right\}\,\leq \,  \frac{C}{n^7}.
	\end{equation}
	Therefore, combining (\ref{eqn:basic_inequality}), (\ref{eqn:maximal_inequality}), Lemma \ref{lem:empirical_process_1},  Lemma \ref{lem:empirical_process_2}, and Lemma \ref{lem:empirical_process_3}
	we obtain that conditioning on $\Omega$,
	\[
	\begin{array}{lll}
	\frac{1}{2}\| \hat{\theta}  \,-\, \theta^*  \|_n^2     & \leq &  \frac{\left\{(1+\delta)\,c_{2,d}\,\tilde{b}_1\,K\right\}^{ \frac{1}{2}   }}{n} \,\Bigg(2\,\sigma\,\left\{    2\,\log \left( \frac{e}{\delta} \right)  \right\}^{\frac{1}{2}   }\, \|\hat{\theta}  \,-\,\theta^{*}   \|_2 \,\,+\,\,   \sigma \,C_1(d)\,\left[ \frac{2\,\log n}{d}\,\log\left\{ \frac{ C_2(p_{\max},L_{\min},d)\, n}{K\,\delta} \right\}  \right]^{ \frac{1}{2}  }  \\
	& &\,\,\,\,\,\,\,\,\,\,\,\,\,\,\,\,\,\,\,\,\,\,\,\,\,\,\,\,\,\,\,\,\,\,\,\,\,\,\,\,\,\,\,\,\,\,\,\,\,\cdot \Big(\|\nabla_{G_K} \hat{\theta}\|_1  \,+\,\|\nabla_{G_K}\,\theta^{*}   \|_1 \Big) \Bigg)    \,\,+\,\,   \\
	& &   \frac{8\sigma\, \left(\log n\right)^{ \frac{1}{2}  }       }{n}\Big(\|\nabla_{G_K} \hat{\theta}\|_1  \,+\,\|\nabla_{G_K}\,\theta^{*}   \|_1 \Big)\,\,+\,\,\frac{\lambda}{n} \left( -\|\nabla_{G_K} \hat{\theta}\|_1  \,+\,\|\nabla_{G_K}\,\theta^{*}   \|_1\right) \\
	& \leq&   \|\nabla_{G_K}\,\theta^{*}   \|_1\,\Big(  \frac{\sigma \,C_1(d)}{n}\,\left[(1+\delta)\,c_{2,d}\,\tilde{b}_1\,\,\frac{2\,\log n}{d}\,\log\left\{ \frac{ C_2(p_{\max},L_{\min},d)\, n}{K\,\delta}\right\}\,K\right]^{  \frac{1}{2}    }    \\ 
	&&\,\,+\,\, \frac{8\sigma\, \log^{ \frac{1}{2}   } n    }{n}\Big)  \,\,+\,\, 2\,\sigma\frac{\left\{2\,(1+\delta)\,c_{2,d}\,\tilde{b}_1\,K\,\log \left( \frac{e}{\delta} \right)\right\}^{ \frac{1}{2}   }   }{n} \|\hat{\theta}  \,-\,\theta^{*}   \|_2,
	\end{array}
	\]
	with probability at least  $\eta_n$.
	Hence by the inequality $a\,b  - 4^{-1}\,b^2  \,\leq \, a^2 $ , we obtain that conditioning on $\Omega$,
	\[
	\begin{array}{lll}
	\frac{1}{4}\| \hat{\theta}  \,-\, \theta^*  \|_n^2      & \leq &  \|\nabla_{G_K}\,\theta^{*}   \|_1\,\Big(  \frac{\sigma \,C_1(d)}{n}\,\left[(1+\delta)\,c_{2,d}\,\tilde{b}_1\,\,\frac{2\,\log n}{d}\,\log\left\{\frac{ C_2(p_{\max},L_{\min},d)\, n}{K\,\delta} \right\} \,K\right]^{\frac{1}{2}   }    \\ 
	&&\,\,+\,\, \frac{8\sigma\, (\log n)^{\frac{1}{2}   }   }{n}\Big)  \,\,+\,\,\frac{4\,\sigma^2\,(1+\delta)\,c_{2,d}\,\tilde{b}_1\,K\,\log \left( \frac{e}{\delta} \right) }{n},
	\end{array}
	\]
	with high probability. The claim follows.
\end{proof}

\section{Auxiliary lemmas for Proof of Theorem \ref{thm:penalty} }
\label{sec:useful_lemma}

Throughout we use the notation  from Section \ref{sec:quantization}. The lemma below  resembles Lemma \ref{lem:controlling_counts}  with the difference that we now prove that an edge in the $K$-NN  graph induces an edge in a certain mesh. 

\begin{lemma}
	\label{lem:controlling_counts2}
	Assume that $N$  in the construction of $G_{\text{lat}}$ is chosen as
	\[
	N\,\,=\,\,  \floor[\Bigg]{  \frac{  \,\left( \,c_{1,d}\,p_{\min} \right)^{1/d} \,n^{1/d}  }{2^{1/d}\,L_{\max}\,K^{1/d}} }.
	\]
	Then there exist positive constants $b^{\prime}_1$ and $b^{\prime}_2$ depending on $L_{\min}$, $L_{\max}$, $d$, $p_{\min}$, $c_{1,d}$, and $c_{2,d} $, such that
	\begin{equation}
	\begin{array}{lll}
	\textup{pr}\bigg\{ \underset{x  \in   h^{-1}(P_{\text{lat}}) }{\max}\vert C(x)\vert   \,\,\geq \,\, (1+\delta)\,c_{2,d}\,b^{\prime}_1\,K  \bigg\}   &  \leq &  N^d \exp\left(  -\frac{1}{3}\delta^2\,b^{\prime}_2\,c_{1,d}\,K   \right),\\
	\textup{pr}\bigg\{ \underset{x  \in   h^{-1}(P_{\text{lat}}) }{\min}\vert C(x)\vert   \,\,\leq \,\, (1-\delta)\,c_{1,d}\,b^{\prime}_2\,K  \bigg\} &\leq  &N^d \exp\left(  -\frac{1}{3}\delta^2\,b^{\prime}_2\,c_{1,d}\,K   \right),
	\end{array}
	\end{equation}
	for all $\delta \in (0,1]$. Moreover, let  $\tilde{\Omega}$ denote the event:  ``For all  $i, j \in [n]$, if $x_i$ and $x_j$  are connected in the $K$-NN graph, then  $\|h(x_i^{\prime}) \,-\,h(x_j^{\prime}) \|_2 \,<\,  2\,N^{-1} $  where  $x_i  \in C(x_i^{\prime})$ and $x_j \in C(x_j^{\prime})$  with $x_i^{\prime},x_j^{\prime} \in I(N)$". Then 
	$$\textup{pr}\left(\tilde{\Omega} \right) \,\,  \geq  \,\,  1\,\,-\,\,  n\,\exp(-K/3). $$
	
\end{lemma}

\begin{proof}
	Let   $i, j \in [n]$ such that $x_i$ and $x_j$  are connected in the $K$-NN graph   where  $x_i  \in C(x_i^{\prime})$ and $x_j \in C(x_j^{\prime})$  with $x_i^{\prime},x_j^{\prime} \in I(N)$. Then, 
	\[
	\begin{array}{lll}
	\|  h(x_i^{\prime})  \,-\,h(x_j^{\prime}) \|_{2}  &  \leq &\|  h(x_i^{\prime})  \,-\,h(x_i) \|_{2}  \,\,+\,\, \|  h(x_i)  \,-\,h(x_j) \|_{2}  \,\,+\,\,\|  h(x_j)  \,-\,h(x_j^{\prime}) \|_{2} \\
	&\leq & \frac{1}{N} \,\,+\,\,  \|  h(x_i)  \,-\,h(x_j) \|_{2}\\
	& \leq &   \frac{1}{N} \,\,+\,\,   L_{\max}\,d_{\mathcal{X}}(x_i,x_j)  \\ 
	& \leq &   \frac{1}{N} \,\,+\,\,   L_{\max}\,R_{K,\max}
	\end{array}
	\] 
	Therefore,
	\[
	\textup{pr}\left(\tilde{\Omega}\right) \,\,\geq \,\,\,\textup{pr}\left\{  R_{K,\max}\,\leq \,  \tilde{a}(K/n)^{1/d}  \right\} \,\,\geq  \,\,1 \,\,-\,\, n\,\exp(-K/12), 
	\] 
	where $\tilde{a}$ is given as in Lemma \ref{lem:controlling_dist}.
	
	Next, we derive an upper bound on the counts of the mesh.	Assume that $x \in h^{-1}\{P_{\text{lat}}(N)\}$,  and $x^{\prime } \in C(x)$. Then,
	\[
	\begin{array}{lll}
	d_{\mathcal{X}}(x,x^{\prime})  & \leq &  \frac{1}{L_{\min}}  \|   h(x) - h(x^{\prime}) \|_2 \\
	& \leq & \frac{1}{2\,  L_{\min}\,N }\\
	& \leq & \frac{1}{2\,  L_{\min}\, } \frac{2 ^{ 1 +1/d}  \,K^{1/d}\, L_{\max}\,   }{n^{1/d}  }\\
	& =:& (b_1^{\prime})^{1/d}  \,  \left( \frac{K}{n} \right)^{1/d},
	\end{array}
	\]
	where the first inequality follows from Assumption \ref{as3}, the second from the definition of $P_{\text{lat}}(N)$, and the third one from the choice of $N$. Therefore,
	\begin{equation}
	\label{one_side_2}
	C(x)   \,\subset\, B_{ (b_1^{\prime})^{1/d}\,  \left( \frac{K}{n} \right)^{1/d}  }(x).
	\end{equation}
	On the other hand, we can find $b_2^{\prime}$  with a similar argument the proof of Lemma \ref{lem:controlling_counts}, and the proof follows the proof of that lemma.
	
\end{proof}

\begin{lemma}
	\label{thm:pred_rate}
	Suppose that Assumptions \ref{as1}--\ref{as3} hold, and choose $K \asymp \log^{1+2r} n $ for some $r >0$. 
	With  $\hat{f}(x)$ the prediction function for the $K$-NN-FL estimator as defined in (\ref{eqn:interpolation_rule}), and for an 
	appropriate choice of $\lambda$, it follows that 
	\begin{equation}
	\label{eqn:pred_rate}
	\mathbb{E}_{X  \sim p}\left\vert    f_0(X) \,-\,\hat{f}(X)  \right\vert^2	\,\,=\,\, O_{\text{pr}}\left( \frac{ \log^{1+2r} n  }{n}\,+\, \frac{\log^{1.5 + r} n}{n} \|\nabla_{G_K}\,\theta^{*}\|_1   \,\,+\,\, \text{AErr}\right),
	\end{equation}
	where $\text{AErr}$ is the approximation error,  defined as
	\[
	\text{AErr}\,\,=\,\,  \bigintssss   \,\left\{  f_0(x) \,-\,\frac{1}{ K    }  \underset{i \in  \mathcal{N}_K(x) }{\sum } f_0(x_i) \right\}^2\,\,p(x)\,  \mu(dx).
	\]
\end{lemma}


\begin{proof}
	Throughout we use the notation from Appendix  \ref{sec:quantization}.
	We start by noticing that
	\[
	\def\arraystretch{2.3}
	\begin{array}{lll}
	\mathbb{E}_{x  \sim p}\left\{    f_0(x) \,-\,\hat{f}(x)  \right\}^2  & =&         \bigints   \,\left\{  f_0(x) \,\,-\,\,\hat{f}(x)   \right\}^2   \,p(x)\,  \mu(dx)  \\
	&  =   &    \bigints   \,\Big\{  f_0(x) \,-\,\frac{1}{ \vert  \mathcal{N}_K(x)\vert    }  \underset{i \in  \mathcal{N}_K(x) }{\sum } f_0(x_i)   \\
	& & \,+\,\frac{1}{ \vert  \mathcal{N}_K(x)\vert    }  \underset{i \in  \mathcal{N}_K(x) }{\sum } f_0(x_i)\,-\,\frac{1}{ \vert  \mathcal{N}_K(x)\vert    }  \underset{i \in  \mathcal{N}_K(x) }{\sum } \hat{f}(x_i)    \Big\}^2   \,p(x)\,  \mu(dx)  \\
	& \leq &  2\, \bigints   \,\Big\{  f_0(x) \,-\,\frac{1}{ \vert  \mathcal{N}_K(x)\vert    }  \underset{i \in  \mathcal{N}_K(x) }{\sum } f_0(x_i) \Big\}^2\,\,p(x)\,  \mu(dx)    \,\,+\,\, \\
	& & 2\, \bigints   \,\Big\{ \frac{1}{ \vert  \mathcal{N}_K(x)\vert    }  \underset{i \in  \mathcal{N}_K(x) }{\sum } f_0(x_i)\,-\,\frac{1}{ \vert  \mathcal{N}_K(x)\vert    }  \underset{i \in  \mathcal{N}_K(x) }{\sum } \hat{f}(x_i)    \Big\}^2\,\,p(x)\,  \mu(dx). \\
	\end{array}
	\]  
	Therefore  we proceed to bound the second term in the last inequality. We observe that
	\begin{equation}
	\label{eqn:pred1}
	\begin{array}{l}
	\bigints   \,\Big\{ \frac{1}{ \vert  \mathcal{N}_K(x)\vert    }  \underset{i \in  \mathcal{N}_K(x) }{\sum } f_0(x_i)\,-\,\frac{1}{ \vert  \mathcal{N}_K(x)\vert    }  \underset{i \in  \mathcal{N}_K(x) }{\sum } \hat{f}(x_i)    \Big\}^2\,\,p(x)\,  \mu(dx)\\
	= \,\,\underset{ x^{\prime} \in I(N)  }{\sum }\,\,\underset{C(x^{\prime})}{ \bigints }\Big\{ \frac{1}{ \vert  \mathcal{N}_K(x)\vert    }  \underset{i \in  \mathcal{N}_K(x) }{\sum } f_0(x_i)\,-\,\frac{1}{ \vert  \mathcal{N}_K(x)\vert    }  \underset{i \in  \mathcal{N}_K(x) }{\sum } \hat{f}(x_i)    \Big\}^2\,\,p(x)\,  \mu(dx).
	\end{array}
	\end{equation}
	Let $x^{\prime}  \in \mathcal{X}$. Then there exists $u(x^{\prime})   \in  I(N)$  such that $x \in  C(u(x^{\prime}))$ and
	\[
	\|  h(x^{\prime})\,-\,h(u(x^{\prime}))   \|_2\,\,\leq \,\, \frac{1}{2\,N}.
	\]
	Moreover, by Lemma \ref{lem:controlling_counts}, with high probability, there exists $i(x^{\prime}) \in [n]$ such that $x_{i(x^{\prime})}   \in   C(u(x^{\prime}))$. Hence,
	\[
	d_{\mathcal{X}}(x^{\prime},x_{i(x^{\prime})})    \,\, \leq   \frac{1}{L_{\min}}\,\, \|  h(x^{\prime})\,-\, h(x_{i(x^{\prime})} )\|_2\,\,\leq \,\, \frac{1}{  L_{\min}\,N}.
	\]
	This implies that there exist $x_{i_1},\ldots,x_{i_K} $,  $i_1,\ldots,i_K  \in [n]$ such that
	\[
	d_{\mathcal{X}}(x^{\prime},x_{i_l}) \,\,\leq \,\,  \frac{1}{  L_{\min}\,N}\,\,+\,\,  R_{K,\max},\,\,\,\,\,l =1,\ldots, K. 
	\]
	Thus,  with high probability,  for any $x^{\prime}  \in \mathcal{X}$,
	\begin{equation}
	\def\arraystretch{2}
	\label{eqn:subset}
	\begin{array}{lll}
	\mathcal{N}_K(x^{\prime} )&\subset &   \left\{   i\,\,:\,\,  d_{\mathcal{X}}(x^{\prime},x_{i}) \,\,\leq \,\,  \frac{1}{  L_{\min}\,N}\,\,+\,\,  R_{K,\max} \right\}\\ 
	& \subset  &\left\{   i\,\,:\,\,  \| h(x^{\prime})\,-\,  h(x_{i}) \|_2 \,\,\leq \,\,  \frac{L_{\max}}{  L_{\min}\,N}\,\,+\,\,  L_{\max}\,R_{K,\max} \right\}\\
	& \subset  &\left\{   i\,\,:\,\,  \| h(u(x^{\prime}))\,-\,  h(x_{i}) \|_2 \,\,\leq \,\,       \left( \frac{1}{2} \,\,+\,\, \frac{L_{\max}}{  L_{\min}}\right)\frac{1}{N}  \,\,+\,\,  L_{\max}\,R_{K,\max} \right\}.
	\end{array}
	\end{equation}
	Hence,  we set
	\[
	\tilde{\mathcal{N}}(u(x^{\prime}))   \,=\,\left\{   i\,\,:\,\,  \| h(u(x^{\prime}))\,-\,  h(x_{i}) \|_2 \,\,\leq \,\,       \left( \frac{1}{2} \,\,+\,\, \frac{L_{\max}}{  L_{\min}}\right)\frac{1}{N}  \,\,+\,\,  L_{\max}\,R_{K,\max} \right\}.
	\]
	As a result, denoting by   $u_1,\ldots,u_{N^d}$  the elements of  $I(N)$, we have that for $j \in [N^d]$,
	\[    
	\def\arraystretch{2.3}       
	\begin{array}{l}
	\underset{C(u_j)}{\bigints}\,\Big\{ \frac{1}{ \vert  \mathcal{N}_K(x)\vert    }  \underset{i \in  \mathcal{N}_K(x) }{\sum } f_0(x_i)\,-\,\frac{1}{ \vert  \mathcal{N}_K(x)\vert    }  \underset{i \in  \mathcal{N}_K(x) }{\sum } \hat{f}(x_i)    \Big\}^2\,\,p(x)\,  \mu(dx)\\
	\leq \,\,\underset{C(u_j)}{\bigints}\, \frac{1}{ \vert  \mathcal{N}_K(x)\vert    } \,\underset{i \in  \mathcal{N}_K(x) }{\sum } \,\left\{ f_0(x_i) \,-\,  \hat{f}(x_i)  \right\}^2\,p(x)\,\mu(dx)\\
	\leq\,\,\frac{1}{ K  } \,\underset{C(u_j)}{\bigints}\, \underset{i \in \tilde{\mathcal{N}}(u_j) }{\sum } \,\left\{ f_0(x_i) \,-\,  \hat{f}(x_i)  \right\}^2\,p(x)\,\mu(dx)\\
	=\,\,\frac{1}{K    } \,\left[\underset{i \in  \tilde{\mathcal{N}}(u_j) }{\sum } \,\left\{ f_0(x_i) \,-\,  \hat{f}(x_i)  \right\}^2   \right]\,\underset{C(u_j)}{\bigints}p(x)\,\mu(dx).\\
	\end{array}
	\]
	The above combined with (\ref{eqn:pred1})  leads to
	\begin{equation}
	\label{eqn:step2}      
	\def\arraystretch{2.3} 
	\begin{array}{l}
	\displaystyle
	\int  \,\Bigg\{ \frac{1}{ \vert  \mathcal{N}_K(x) \vert    }  \underset{i \in  \mathcal{N}_K(x) }{\sum } f_0(x_i)\,-\,\frac{1}{ \vert  \mathcal{N}_K(x)\vert    }  \underset{i \in  \mathcal{N}_K(x) }{\sum } \hat{f}(x_i)    \Bigg\}^2\,\,p(x)\,  \mu(dx)\\
	\leq \,\,\displaystyle
	\sum_{j=1}^{N^d}\left(\frac{1}{ K   } \,\left[\underset{i \in  \tilde{\mathcal{N}}(u_j) }{\sum } \,\left\{f_0(x_i) \,-\,  \hat{f}(x_i)  \right\}^2   \right]\,\underset{C(u_j)}{\int}p(x)\,\mu(dx)\right)\\
	\displaystyle	\leq \,\,\, \left[\sum_{j=1}^{N^d}\frac{1}{ K   } \,\underset{i \in  \tilde{\mathcal{N}}(u_j) }{\sum } \,\left\{ f_0(x_i) \,-\,  \hat{f}(x_i)  \right\}^2   \,\right]\,\, \underset{1 \leq j \leq N^d}{\max} \underset{C(u_j)}{\int}p(x)\,\mu(dx)\\
	\displaystyle	\leq \,\,\left[ \sum_{i=1}^{n} \,\frac{1}{K}    \underset{j \in [N^d]\,:\,  \| h(x_i) \,-\, h(u_j)\|_2 \,\leq \,  \left( \frac{1}{2} \,+\, \frac{L_{\max}}{  L_{\min}}\right)\frac{1}{N}  \,+\, L_{\max}\,R_{K,\max}  }{\sum}  \,\left\{ f_0(x_i) \,-\,\hat{f}(x_i)  \right\}^2      \right]\,\, \underset{1 \leq j \leq N^d}{\max} \underset{C(u_j)}{\int}p(x)\,\mu(dx)\\
	\displaystyle	\leq \,\,\, \left[\sum_{i=1}^{n} \left\{ f_0(x_i) \,-\,  \hat{f}(x_i)  \right\}^2 \right]\,\left[  \underset{i \in [n]}{\max} \left\vert \left\{ j\,:\, \| h(x_i) \,-\, h(u_j)\|_2 \,\leq \,  \left( \frac{1}{2} \,+\, \frac{L_{\max}}{  L_{\min}}\right)\frac{1}{N}\,+\,  L_{\max}\,R_{K,\max}   \right\} \right\vert   \right]\cdot\\
	\,\,\,\,\,\,\,\,\,\,\frac{1}{K}\,\,\underset{1 \leq j \leq N^d}{\max} \underset{C(u_j)}{\int}p(x)\,\mu(dx).
	\end{array}
	\end{equation}
	Next,  for  a set $A \subset \mathbb{R}^d$  and  positive constant $r$, we define the packing and external covering numbers as
	\[
	\begin{array}{lll}
	\text{N}_{ r}^{\text{pack}}(A) & := & \max\left\{  	l \in \mathbb{N} \,:\, \exists  q_1,\ldots,q_l \in A,\,\,\,\text{such that}\,\,\, \|q_j - q_{j^{\prime}}\|_2 > r   \,\,\,\,\forall j \neq j^{\prime}   \right\},\\
	\text{N}_{ r}^{ \text{ext} }  & := & \min\left\{ l \in \mathbb{N}\,:\, \exists  q_1,\ldots,q_l \in \mathbb{R}^d,  \text{such that}\,\,\, \forall x \in A \,\,\,\,\text{there exists}\,\, l_x  \,\,\text{with}\,\,\,\|x-q_{l_x}\|_2 < r     \right\}.
	\end{array}
	\]
	Furthermore, by Lemma \ref{lem:controlling_counts}, there exists a constant $\tilde{c}$ such that  $R_{K,\max}   \leq  \tilde{c}/N$ with high probability. This implies that  with high probability, for  a positive constant $\tilde{C}$,  we have that 
	\begin{equation}
	\label{eqn:entropy_argument}
	\def\arraystretch{2.3}
	\begin{array}{l}
	\underset{i \in [n] }{\max} \left\vert \left\{ j \in [N^d] \,:\, \| h(x_i) \,-\, h(u_j)\|_{2} \,\leq \,  \left( \frac{1}{2} \,+\, \frac{L_{\max}}{  L_{\min}}\right)\frac{1}{N}\,+\,  L_{\max}\,R_{K,\max}   \right\} \right\vert \\
	\leq \,\,
	\underset{i \in [n] }{\max} \left\vert \left\{ j \in [N^d]\,:\, \| h(x_i) \,-\, h(u_j)\|_{2} \,\leq \,  \frac{ \tilde{C} }{N}   \right\} \right\vert \\
	\end{array}
	\end{equation}
	\begin{equation}
	\begin{array}{l}
	\label{eqn:entropy_argument2}
	\def\arraystretch{2.3}
	\leq \,\,\underset{i \in [n] }{\max}  \,\,\, \text{N}_{ \frac{1}{N} }^{\text{pack}}\left(  B_{\frac{ \tilde{C} }{N}   } (h(x_i)) \right)\\
	\leq   \,\,  \text{N}_{  \frac{1}{N} }^{ \text{ext} }(  B_{\frac{ \tilde{C} }{N}   } (0)   )\\
	=\,\,  \text{N}_{ 1}^{ \text{ext} }(   B_{\tilde{C}    } (0)   ) \\
	<\,\,  C^{\prime}, 
	\end{array}
	\end{equation}
	for some  positive constant  $C^{\prime}$, where the first inequality  follows from  $R_{K,\max}   \leq  \tilde{c}/N$, the second from the definition of packing number, and the remaining  inequalities from  well-known properties of packing and external covering numbers.

	Therefore, there exists a constant $C_1$ such that
	\[
	\def\arraystretch{2.3} 
	\begin{array}{l}
	\displaystyle \int   \,\Bigg\{ \frac{1}{ \vert  \mathcal{N}_K(x)\vert    }  \underset{i \in  \mathcal{N}_K(x) }{\sum } f_0(x_i)\,-\,\frac{1}{ \vert  \mathcal{N}_K(x)\vert    }  \underset{i \in  \mathcal{N}_K(x) }{\sum } \hat{f}(x_i)    \Bigg\}^2\,\,p(x)\,  \mu(dx)\\
	\displaystyle	\leq \,\,\left[\sum_{i=1}^{n} \left\{ f_0(x_i) \,-\,  \hat{f}(x_i)  \right\}^2 \right]\,
	\,\frac{C_1}{K}\,\,\underset{1 \leq j \leq N^d}{\max} \underset{C(u_j)}{\int}p(x)\,\mu(dx)\\
	\displaystyle	\leq \,\,\left[\sum_{i=1}^{n} \left\{ f_0(x_i) \,-\,  \hat{f}(x_i)  \right\}^2 \right]\,
	\,\frac{C_1\,p_{\max}}{K}\,\,\mu\left( B_{  \frac{ \tilde{b}_1}{p_{\max} ^{1/d} } \,  \left( \frac{K}{n} \right)^{1/d}  }(x) \right)\\
	\end{array}
	\]
	where the last equation follows as in (\ref{one_side}). The claim then follows.
	
\end{proof}

\begin{lemma}
	\label{thm:aerr}
	Assume that $g_0 :=  f_0\circ h^{-1}  $  satisfies  Definition  \ref{def:pie_lip}, i.e. $g_0$ is piecewise Lipschitz.
	Then, under Assumptions \ref{as1}--\ref{as3}, 
	\[
	\text{AErr}    \,\,=\,\, O_{ \text{pr}}\left(  \frac{K^{1/d}}{n^{1/d}} \right),
	\]
	provided   that  $K/\log n   \rightarrow \infty$, where $\text{AErr}$ was defined in Lemma \ref{thm:pred_rate}.
	Consequently,  with   $K \asymp \log^{1+2r} n $ for some $r >0$,  and for an appropriate choice of $\lambda$,  we have that
	\begin{equation}
	\mathbb{E}_{X  \sim p}\left\vert   f_0(X) \,-\,\hat{f}(X)  \right\vert^2	\,\,=\,\, O_{\text{pr}}\left\{   \frac{ \log^{  2.5 +3r +   (1+2r)/d } n   }{n^{1/d}}     \right\}.
	\end{equation}
\end{lemma}

\begin{proof}
	Throughout we use the notation from Lemma \ref{lem:controlling_counts}  and Section \ref{sec:quantization}. We denote by $u_1,\ldots,u_{N^d}$ the elements of $h^{-1}(P_{\text{lat}}(N))$, and so
	\[
	\def\arraystretch{2.3} 
	\begin{array}{lll}
	\text{AErr}  &=  &  \bigints   \,\Big\{ f_0(x) \,-\,\frac{1}{ K    }  \underset{i \in  \mathcal{N}_K(x) }{\sum } f_0(x_i) \Big\}^2\,\,p(x)\,  \mu(dx)\\
	&= &	 	\underset{j \in [N^d]}{\sum}\,	\underset{C(u_j)}{\bigints }	   \,\Big\{f_0(x) \,-\,\frac{1}{ K    }  \underset{i \in  \mathcal{N}_K(x) }{\sum } f_0(x_i) \Big\} ^2\,\,p(x)\,  \mu(dx)\\
	&   \leq& \underset{j \in [N^d]}{\sum}\,	\underset{C(u_j)}{\bigints } \,  \underset{i \in \mathcal{N}_K(x)  }{\sum }
	\frac{1}{K} \Big\{  f_0(x) \,-\, f_0(x_i) \Big\} ^2\,\,p(x)\,  \mu(dx)
	\\
	&  = & \underset{j \in [N^d]\,:\, \,\,  C(u_j) \cap   \mathcal{S}\, =\, \emptyset     }{\sum}\,\,\,\,	\underset{C(u_j)}{\bigints } \,  \underset{i \in \mathcal{N}_K(x)   }{\sum }
	\frac{1}{K} \Big\{  f_0(x) \,-\, f_0(x_i) \Big\}^2\,\,p(x)\,  \mu(dx)  \,\,+\,\,\\
	& & \underset{j \in [N^d]\,:\, \,\,  C(u_j) \cap  \mathcal{S}\, \neq \, \emptyset     }{\sum}\,\,\,\,	\underset{C(u_j)}{\bigints } \,  \underset{i \in \mathcal{N}_K(x)   }{\sum }
	\frac{1}{K} \Big\{  f_0(x) \,-\, f_0(x_i) \Big\}^2\,\,p(x)\,  \mu(dx),\\
	\end{array}
	\]
	where the  inequality holds by convexity.
	Therefore, 
	\[
	\def\arraystretch{2.3} 
	\begin{array}{lll}
	\text{AErr} 	&  \leq &  \,4\, \left\vert \left\{  j \in [N^d]\,:\, \,\,  C(u_j) \cap  \mathcal{S}\, \neq \, \emptyset  \right\}\right\vert\,\,\| f_0 \|_{\infty}^2\,\,\,\underset{1 \leq j \leq N^d}{\max} \,\underset{C(u_j)}{\bigints } \,\,p(x)\,  \mu(dx)  \,\,+\,\,\\
	& &   \underset{j \in [N^d]\,:\, \,\,  C(u_j) \cap  \mathcal{S}\, =\, \emptyset     }{\sum}\,\,\,\,	\underset{C(u_j)}{\bigints } \,  \underset{i \in \mathcal{N}_K(x)   }{\sum }
	\frac{1}{K} \Big\{  f_0(x) \,-\, f_0(x_i) \Big\}^2\,\,p(x)\,  \mu(dx) \\
	& \leq&\mu\left( B_{  \frac{ \tilde{b}_1^{1/d} }{p_{\max} ^{1/d} } \,  \left( \frac{K}{n} \right)^{1/d}  }(x) \right)\,\,\left( 4\, p_{\max}\,\| f_0 \|_{\infty}^2  \right)\,\left\vert \left\{  j \in [N^d]\,:\, \,\,  C(u_j) \cap   \mathcal{S}\, \neq \, \emptyset  \right\}\right\vert\\
	& & \,\,+\,\, \underset{j \in [N^d]\,:\, \,\,  C(u_j) \cap  \mathcal{S}\, =\, \emptyset     }{\sum}\,\,\,\,	\underset{C(u_j)}{\bigints } \,  \underset{i \in \mathcal{N}_K(x)  }{\sum }
	\,\,\frac{1}{K} \Big\{  g_0(h(x)) \,-\, g_0(h(x_i) )\Big\}^2\,p(x)\,  \mu(dx) \\
	& \leq& \left(c_{2,d}\,4\, \tilde{b}_1 \,\| f_0 \|_{\infty}^2\right) \frac{K}{n}\,\left\vert \left\{  j \in [N^d]\,:\, \,\,  C(u_j) \cap   \mathcal{S}\, \neq \, \emptyset  \right\}\right\vert\\
	& & \,\,+\,\, \underset{j \in [N^d]\,:\, \,\,  C(u_j) \cap \mathcal{S}\, =\, \emptyset     }{\sum}\,\,\,\,	\underset{C(u_j)}{\bigints } \,  \underset{i \in \mathcal{N}_K(x)   }{\sum }
	\,\,\frac{L_0}{K} \| h(x) \,-\, h(x_i) \|_2^2\,p(x)\,  \mu(dx) \\
	& \leq &     \left(c_{2,d}\,4\, \tilde{b}^d \,\| f_0 \|_{\infty}^2\right) \frac{K}{n}\,\left\vert \left\{  j \in [N^d]\,:\, \,\,  C(u_j) \cap   \mathcal{S}\, \neq \, \emptyset  \right\}\right\vert\\
	&&   \,\,+\,\,  L_0 \,\left\{\underset{j \in [N^d]\,:\, \,\,  C(u_j) \cap  \mathcal{S}\, =\, \emptyset     }{\sum}\,\,	\underset{C(u_j)}{\bigints } \,p(x)\,  \mu(dx) \right\} \,\left( \frac{L_{\max}}{  L_{\min}}\frac{1}{N}  \,\,+\,\,  L_{\max}\,R_{K,\max}\right)^2\\
	& \leq &     \left(c_{2,d}\,4\, \tilde{b}^d \,\| f_0 \|_{\infty}^2\right) \frac{K}{n}\,\left\vert \left\{  j \in [N^d]\,:\, \,\,  C(u_j) \cap  \ \mathcal{S}\, \neq \, \emptyset  \right\}\right\vert\\
	&&   \,\,+\,\,  L_0 \, \,\left( \frac{L_{\max}}{  L_{\min}}\frac{1}{N}  \,\,+\,\,  L_{\max}\,R_{K,\max}\right)^2,\\
	\end{array}
	\]
	
	\noindent where the first inequality holds  by elementary properties of integrals,
	the second inequality by (\ref{one_side}), the third inequality by Assumptions \ref{as2}--\ref{as3}, the fourth inequality by  the same argument as in (\ref{eqn:subset}), and the fifth inequality from  properties of integration.  The conclusion of the lemma follows from the inequality above combined with Lemma  \ref{lem:controlling_dist}  and   the proof of Proposition 23 from \cite{hutter2016optimal} which uses Lemma 8.3 from \cite{arias2012oracle}.

\end{proof}

\section{Proof of Theorem   \ref{thm:penalty}  }
\label{sec:appendix_penalty}
\begin{proof}
	Combining  Lemmas \ref{thm:pred_rate} and \ref{thm:aerr}  we obtain that
	\[
	E_{X  \sim p}\left\{\left\vert   f_0(X) \,-\,\hat{f}(X)  \right\vert^2\right\}    	\,\,=\,\, O_{ \textup{pr} }\left\{   \frac{ \log^{  \alpha} n   }{n^{1/d}}     \right\},
	\]
	provide that Assumptions \ref{as1}--\ref{as3}  hold  and  $f_0$   satisfies   Definition \ref{def:pie_lip}.
	
	Suppose  now that Assumptions   \ref{as1}--\ref{as:sup} hold. Throughout, we extend the domain of the function $g_0$  to be $\mathbb{R}^{d}$ by simply making it take the value zero in  $\mathbb{R}^{d} \backslash [0,1]^d$.  We will proceed to construct smooth approximations to $g_0$  that will allow us to obtain the desired result. To that end, for any $\epsilon   \,>\,0$   we construct the regularizer
	(or mollifier)  $g_{\epsilon} \,:\, \mathbb{R}^d   \rightarrow \mathbb{R} $ defined as
	\[
	g_{\epsilon}(z)\,\,=\,\,  \psi_{\epsilon}*g_0 (z)    \,\, =  \,\, \int \,  \psi_{\epsilon}(z^{\prime})\,g_0(z  \,-\, z^{\prime} )\,dz^{\prime},
	\]
	where  $\psi_{\epsilon}(z^{\prime})\,=\, \epsilon^{-d}\, \psi( z^{\prime}/\epsilon )$. Then  given Assumption \ref{as:_rf}, by the proof of Theorem 5.3.5   from  \cite{zhu2003semisupervised},  it follows that there exists a constant $C_2$ such that
	\[
	\underset{\epsilon   \rightarrow^+  0 }{\lim \sup}\, \,\, \underset{(0,1)^d}{\int} \| \nabla g_{\epsilon}(z)\|_1 \,dz        \,\,<\,\, C_2
	\] 
	which implies that there exists $\epsilon_1   >0$  such that 
	\begin{equation}
	\label{eqn:first_bound}
	\underset{0<\epsilon  \,<\, \epsilon_1      }{\sup}\, \,\,\underset{(0,1)^d}{\int}  \| \nabla g_{\epsilon}(z)\|_1  \,dz     \,\,<\,\, C_2.
	\end{equation}

	Next,	for $N$ as in Lemma \ref{lem:controlling_counts2}, we set $\epsilon  \,\,=\,\, N^{-1}$ and consider the event
	\begin{equation}
	\label{eqn:omega_definition}
	\Lambda_{\epsilon} \,=\, \left\{   h(x_1) \,\, \in\,\,  B_{4\epsilon}(\mathcal{S}) \cup \left[(0,1)^d\backslash \Omega_{4\epsilon}\right]     \right\},
	\end{equation}
	and note that
	\begin{equation}
	\label{eqn:prob}
	\begin{array}{lll}
	\textup{pr}( \Lambda_{\epsilon}) &=&    \underset{  h^{-1}\left( B_{4\epsilon}(\mathcal{S}) \cup (0,1)^d\backslash \Omega_{4\epsilon}\right)   }{\int }    p(z)\, \mu(dz)\\
	& \leq  &   p_{\max}\,\mu\left[ h^{-1}\left\{ B_{4\epsilon}(\mathcal{S}) \cup (0,1)^d\backslash \Omega_{4\epsilon} \right\}  \right]\\
	& \leq&     p_{\max}\,C_{\mathcal{S}}\,4\,\epsilon,\,
	\end{array}
	\end{equation}
	where the last inequality follows from Assumption \ref{as:sup}. Defining 
	\begin{equation}
	\label{eqn:J}
	J \,\,=\,\, \left\{  \,i    \in  [n] \,\,:\,  h(x_i)  \in  \Omega_{4\epsilon}\backslash  B_{4\epsilon}(\mathcal{S}) \right\},  		
	\end{equation}
	by the triangle inequality we have
	\begin{equation}
	\label{eqn:relating}
	\def\arraystretch{2.2}
	\begin{array}{l}
	\left\vert    \underset{(i,j) \in  E_K,\, i,j \in J  }{\sum} \vert \theta_i^* \,-\, \theta_j^* \vert \,-\,\underset{(i,j) \in  E_K\, i,j \in J }{\sum} \vert g_{\epsilon}\{h(x_i)\} \,-\, g_{\epsilon}\{h(x_j)\} \vert   \right\vert     
	\\ 
	=\,\,  
	\left\vert    \underset{(i,j) \in  E_K\, i,j \in J}{\sum} \vert g_0\{h(x_i)\} \,-\, g_0\{h(x_j)\} \vert \,-\,\underset{(i,j) \in  E_K\, i,j \in J}{\sum} \vert g_{\epsilon}\{h(x_i)\} \,-\, g_{\epsilon}\{h(x_j)\}\vert\right\vert\\
	\end{array}
	\end{equation}
	
	\begin{equation}
	\label{eqn:relating2}
	\def\arraystretch{2.2}
	\begin{array}{l}
	\leq  \underset{(i,j) \in  E_k\, i,j \in J}{\sum}  \Big[  \vert g_0\{h(x_i)\}\,-\,g_{\epsilon}\{h(x_i)\}  \vert   \,+\,  \vert g_0\{h(x_j)\}\,-\,g_{\epsilon}\{h(x_j)\} \vert   \Big]\\
	\leq\,\,   d_{\max}\, \underset{i \in J}{\sum}\,\, \vert g_0 \{h(x_i)\}\,-\,g_{\epsilon}\{h(x_i)\} \vert\\
	\leq  \,\,  d_{\max}\, \underset{i \in J}{\sum}\,\, \, \int  \psi_{\epsilon}\{h(x_i) -z \}\,\left \vert g_0\{ h(x_i)\}  \,-\, g_0(z)  \right\vert \,dz\\
	\leq \,\, K\,\tau_d\,   C_1\, \epsilon^{-d} \, \underset{i \in J}{\sum}\,\, \underset{   \|   h(x_i)   \,-\,  z\|_2  \leq  \epsilon   }{\int }\left \vert g_0\{ h(x_i)\}  \,-\, g_0(z) \right\vert \,dz,
	\end{array}
	\end{equation}
	where $\tau_d$ is a positive constant  and the second inequality happens  with high probability as shown in Lemma  \ref{lem:controlling_counts}.
	
	We then bound the last  term in  (\ref{eqn:relating2})  using Assumption \ref{as:sup}. Thus, 
	\begin{equation}
	\label{first_equation}
	\begin{array}{lll}
	\epsilon^{-d}\,\underset{i \in J}{\sum}\,\, \underset{   \|   h(x_i)   \,-\,  z\|_2  \leq  \epsilon   }{\int }\left \vert g_0\{ h(x_i)\}  \,-\, g_0(z)  \right\vert \,dz   & = &
	\epsilon^{-d}	\displaystyle  \sum_{A \in \mathcal{P}_{\epsilon } }\,\,\sum_{  i \in J,  h(x_i)  \in A } \,\, \underset{   \|   h(x_i)   \,-\,  z\|_2  \leq  \epsilon   }{\int }\left \vert g_0\{ h(x_i)\}  \,-\, g_0(z)  \right\vert \,dz \\
	&  \leq &  \left[\underset{z   \in  P_{\text{lat}}(N)  }{\max}  \,\vert C\{h^{-1}(z)\}\vert\right] \,\,S_1(g_0,\mathcal{P}_{N^{-1}, \mathcal{S}  })\,N^{d}\,\epsilon\\
	&\leq  & \left\{(1+\delta)\,c_{2,d}\,b^{\prime}_1\,K\right\} \,\,S_1(g_0,\mathcal{P}_{N^{-1}, \mathcal{S} })\,N^{d}\,\epsilon,
	\end{array}
	\end{equation}
	with probability at least
	\[
	1  \,\,-\,\, \frac{n}{K \, \tilde{a}^d\, L_{\max}^d  } \exp\left(  -\frac{1}{3}\delta^2\,b^{\prime}_2\,c_{1,d}\,K   \right),
	\]
	which follows from Lemma \ref{lem:controlling_counts2}.
	
	
	If $h(x_i) \notin \Omega_{\epsilon} \backslash  B_{\epsilon}(\mathcal{S}) $, then
	\begin{equation}
	\label{eqn:second_upper_bound}
	\left \vert g_0\{ h(x_i)\}  \,-\, g_0\{h(x_j)\}  \right\vert  \,\,\leq \,\,
	\,2\,\|g \|_{L_{\infty}(0,1)^d}.
	\end{equation}

	We now proceed to put the different pieces together.  Setting $\tilde{n} \,\,=\,\,  \vert [n]  \backslash  J    \vert, $ we observe that 
	\[
	\tilde{n} \,\,\sim\,\,\text{Binomial}\left\{ n\,,\,   \textup{pr}(   \Lambda_{\epsilon})  \right\}.
	\]
	If 
	\[
	n^{\prime} \,\,\sim\,\,\text{Binomial}\left\{ n\,,\,   p_{\max}\,C_{\mathcal{S}}\,\,4\epsilon \right\},
	\]
	then  by (\ref{eqn:prob}), we have 
	\begin{equation}
	\label{eqn:counting}
	\def\arraystretch{2.2}
	\begin{array}{lll}
	\textup{pr}\left( \tilde{n}  \,\geq   \,\frac{3}{2} \,n\,p_{\max}\,C_{\mathcal{S}}\,\epsilon\,4\right)  &\leq&   \textup{pr}\left( n^{\prime}  \,\geq \,\frac{3}{2}\,n\,p_{\max}\,C_{\mathcal{S}}\,\epsilon \,4\right)  \\
	& \leq&   \exp\left\{ - \frac{1}{12}\,n\,\left(  \,p_{\max}\,C_{\mathcal{S}}\,\epsilon \right) 4  \right\}\\
	&  = & \exp\left[ - \frac{\,p_{\max}\,4\,C_{\mathcal{S}}}{12}\,n\,\left\{\frac{2^{1/d}\,L_{\max}\,K^{1/d}}{ \,\left( \,c_{1,d}\,p_{\min} \right)^{1/d} \,n^{1/d}}   \right\} \right]\\
	&  = & \exp\left(  - \tilde{C}\, n^{ 1- 1/d}\,K^{1/d}   \right),
	\end{array}
	\end{equation}
	where the first inequality  follows  from \eqref{eqn:prob}, the second from Lemma \ref{lem:binom_concentration},  and $\tilde{C}$ is a positive constant that depends on $p_{\min}$, $p_{\max}$, $L_{\min}$, $L_{\max}$, $d$, and $C_{\mathcal{S}}$. Consequently, combining the above inequality with (\ref{eqn:relating}), (\ref{first_equation})  and (\ref{eqn:second_upper_bound})  we arrive at
	\[
	\def\arraystretch{2.2}
	\begin{array}{lll}
	\underset{(i,j) \in  E_K\,  }{\sum} \vert \theta_i^* \,-\, \theta_j^* \vert  &  = &  \underset{(i,j) \in  E_K,\, i,j \in J  }{\sum} \vert \theta_i^* \,-\, \theta_j^* \vert   \,\,\,+\,\,\,\underset{(i,j) \in  E_K,\, i \notin J   \,\text{or}\,\,j \notin J    }{\sum} \vert \theta_i^* \,-\, \theta_j^* \vert\\
	& \leq &\underset{(i,j) \in  E_K\, i,j \in J }{\sum} \vert g_{\epsilon}\{h(x_i)\} \,-\, g_{\epsilon}\{h(x_j)\} \vert  \,\,+\,\,\\
	& & K\,\tau_d\,   C_1\,\left\{(1+\delta)\,c_{2,d}\,b^{\prime}_1\,K\right\} \,\,S_1(g_0,\mathcal{P}_{N^{-1}})\,N^{d}\,\epsilon\\
	& &  \,+\,\,2\,\|g_0\|_{L_{\infty}(0,1)^d} K\,\tau_d\,\tilde{n}    \,\,\\
	\end{array}
	\]
	
	\[
	\def\arraystretch{2.2}
	\begin{array}{lll}
	\,\,\,\,\,\,\,\,\,	\,\,\,\,\,\,\,\,\,\,\,\,\,\,\,\,\,\,	\,\,\,\,\,\,\,\,\,\,\,\,\,\,\,\,\,\,\,\,\,\,\,\,\,\,\,\,	& <& \underset{(i,j) \in  E_K\, i,j \in J }{\sum} \vert g_{\epsilon}
	\{h(x_i)\} \,-\, g_{\epsilon}\{h(x_j)\} \vert  \,\,+\,\, C_6\,n^{1-1/d}\,K^{1+1/d}\,\\
	&  &\,\,+\,\,  2\,\|g_0\|_{L_{\infty}(0,1)^d} K\,\tau_d\,\tilde{n},    \,\,\\
	\end{array}
	\]
	for some positive constant $C_6 >0$, which happens with high probability, see  (\ref{eqn:counting}).
	Hence, from (\ref{eqn:counting}),
	\begin{equation}
	\label{eqn:ineq2}
	\begin{array}{lll}
	\underset{(i,j) \in  E_K\,  }{\sum} \vert \theta_i^* \,-\, \theta_j^* \vert   &\leq  & \underset{(i,j) \in  E_K\, i,j \in J }{\sum} \vert g_{\epsilon}\{h(x_i)\} \,-\, g_{\epsilon}\{h(x_j)\} \vert  \,\,+\,\, C_6\,n^{1-1/d}\,K^{1+1/d}\,\\
	&  &   \,\,+\,\, C_7\,n^{ 1 - 1/d}\,K^{1+1/d}, 
	\end{array}
	\end{equation}
	where $C_7 >0$ is a constant, and the last inequality holds 
	with probability approaching  one provided that $K /\log n    \rightarrow \infty$.
	
	Therefore, it remains to bound the first term in the right hand side of inequality (\ref{eqn:ineq2}). To that end, we notice that if $(i,j) \in E_K$, then 
	from the proof of Lemma \ref{lem:controlling_counts2}  we observe that
	\[
	\|  h(x_i)   \,-\,  h(x_j) \|_2  \,\,\leq \,\,  L_{\max}\, R_{K,\max}  \,\,\leq \,\,\epsilon,
	\]
	where the last inequality happens with high probability. Hence, with high probability, if $z$ is in the segment connecting $h(x_i)$ and $h(x_j)$, then  $z \notin  B_{2\epsilon}(\mathcal{S})  \cup  ((0,1)^d \backslash  \Omega_{2\epsilon}  ) $ provided that $i,j \in J$.  As a result, by the mean value theorem,  for $i , j  \in J$  there exists a $z_{i,j}  \in    \Omega_{2\epsilon}   \backslash B_{2\epsilon}(\mathcal{S}) $ such that
	
	\[
	g_{\epsilon}\{h(x_i)\} \,-\, g_{\epsilon}\{h(x_j)\}  \,\,=\,\,   \nabla g_{\epsilon}(z_{i,j})^T\{h(x_i)\,-\,h(x_j)\},
	\]
	and this holds uniformly with probability approaching one.
	
	Then,
	\begin{equation}
	\label{eqn:smooth}
	\def\arraystretch{2.3}
	\begin{array}{l}
	\underset{(i,j) \in  E_K\, i,j \in J }{\sum}\, \vert g_{\epsilon}\{h(x_i)\} \,-\, g_{\epsilon}\{h(x_j)\} \vert 
	\\
	= \,\,  \underset{(i,j) \in  E_K\, i,j \in J }{\sum}\,   \,\,\vert  \nabla g_{\epsilon}(z_{i,j})^T\{h(x_i)-h(x_j)\} \vert \\
	\leq \,\,  \left\{  \underset{(i,j) \in  E_K\, i,j \in J }{\sum}\,   \,\,\|  \nabla g_{\epsilon}(z_{i,j}) \|_1   \right\}  \frac{2}{N}\\
	=\,\,2\,\,\,\,\underset{A  \in \mathcal{P}_{\epsilon},\,\,A\cap \{ \Omega_{2\epsilon}\backslash B_{2\epsilon}(\mathcal{S} )\} \,\neq\, \emptyset    }{\sum} 
	\,\,\,\,\, \left[\underset{(i,j) \in  E_K\,  \,\text{s.t}\, i,j \in J,\,  \text{ and } \,z_{i,j} \in A }{\sum}\,\|  \nabla g_{\epsilon}(z_{i,j}) \|_1 \,\text{Vol}\{ B_{\epsilon}(z_{i,j})  \}\right] \frac{C_8\,N^d}{ N}\\
	\leq  \,\,2\,\,\,\,\underset{A  \in \mathcal{P}_{\epsilon},\,\,A\cap \{ \Omega_{2\epsilon}\backslash B_{2\epsilon}(\mathcal{S} )\} \,\neq\, \emptyset    }{\sum} 
	\,\,\,\,\, \Bigg\{\underset{(i,j) \in  E_K\,  \,\text{s.t}\, i,j \in J,\,  \text{ and } \,z_{i,j} \in A}{\sum}\,\,  \underset{B_{\epsilon}(z_{i,j}) }{\bigints}  \|  \nabla g_{\epsilon}(z) \|_1   \,dz \,\Bigg\} \frac{C_8\,N^d}{N}   \,\,+\,\,\\ 
	\,  \,\,2\,\,\,\,\underset{A  \in \mathcal{P}_{\epsilon},\,\,A\cap \{ \Omega_{2\epsilon}\backslash B_{2\epsilon}(\mathcal{S} )\} \,\neq\, \emptyset    }{\sum} 
	\,\,\,\,\, \Bigg\{\underset{(i,j) \in  E_K\,  \,\text{s.t}\, i,j \in J,\,  \text{ and } \,z_{i,j} \in A}{\sum}\,\,  \underset{B_{\epsilon}(z_{i,j}) }{\bigints}  \left\vert \|  \nabla g_{\epsilon}(z_{i,j}) \|_1  \,-\, \|  \nabla g_{\epsilon}(z) \|_1  \right\vert  \,dz \,\Bigg\} \frac{C_8\,N^d}{N}\\
	= :         T_1   \,\,+\,\, T_2.
	\end{array}
	\end{equation}
	
	Therefore  we proceed to bound $T_1$  and $T_2$. Let us assume that  $(i,j) \in  E_K\,  \,\text{with}\, i,j \in J,\,  \text{ and } \,z_{i,j} \in A$  with  $A  \in \mathcal{P}_{\epsilon}$. Then  by Lemma  \ref{lem:controlling_counts2}  we have two cases. Either  $h(x_i)$  and $h(x_j)$ are in the same cell (element of $\mathcal{P}_{\epsilon}$), or $h(x_i)$  and $h(x_j)$  are in adjacent  cells. Denoting by  $c(A^{\prime})$  the center of a cell $A^{\prime}  \in \mathcal{P}_{\epsilon}$,  then if   $z^{\prime}   \in  B_{\epsilon}(z_{i,j})  \cap  A^{\prime} $  it implies that
	\[
	\|  c(A^{\prime})  \,-\,  c(A)   \|_{\infty}  \, \,\leq\,\,  \|  c(A^{\prime})  \,-\, z^{\prime} \|_{\infty}   \,+\, \| z^{\prime}  \,-\,  z_{i,j} \|_{\infty}  \,+\,\| z_{i,j}  \,-\,  c(A) \|_{\infty}    \,\,\leq \,\,  2\epsilon.
	\]
	And if in addition $h(x_i)   \in A_i$  and  $h(x_j)   \in A_j$, then
	\begin{equation}
	\label{eqn:cA}
	\|  c(A_i)  \,-\,  c(A)    \|_{\infty}  \, \,\leq\,\,  \|  c(A_i)  \,-\,  h(x_i)\|_{\infty}  \,+\,   \|  h(x_i)  \,-\,  z_{i,j} \|_{\infty} \,+\,     \|  z_{i,j}  \,-\, c(A)\|_{\infty}   \,\,\leq \,\,  2\epsilon,
	\end{equation}	 
	and the same is true for $c(A_j)$.
	Hence,
	\[
	\def\arraystretch{2.3}
	\begin{array}{l}
	\underset{B_{\epsilon}(z_{i,j}) }{\bigints}  \|  \nabla g_{\epsilon}(z) \|_1   \,dz  \,\,\leq  \,\, \underset{ A^{\prime} \in  \mathcal{P}_{\epsilon}    \,:\,\, \|  c(A)  \,-\,  c(A^{\prime})   \|_{\infty}  \leq 2\epsilon}{\sum} \,\underset{A^{\prime} }{\bigints}  \|  \nabla g_{\epsilon}(z) \|_1   \,dz. 
	\end{array}
	\]
	Since the previous discussion was for an arbitrary $z_{i,j} $  with $i,j  \in J$,  we obtain that
	\[
	\def\arraystretch{2.3}
	\begin{array}{l}
	T_1\,\, \leq \,\,\underset{A  \in \mathcal{P}_{\epsilon},\,\,A\cap \{ \Omega_{2\epsilon}\backslash B_{2\epsilon}(\mathcal{S} )\} \,\neq\, \emptyset    }{\sum} 
	\,\,\,\,\, \Bigg\{\underset{(i,j) \in  E_K\,  \,\text{s.t}\, i,j \in J,\,  \text{ and } \,z_{i,j} \in A}{\sum}\,\,\,\,\,\,  \underset{ A^{\prime}  \in  \mathcal{P}_{\epsilon}    \,:\,\, \|  c(A)  \,-\,  c(A^{\prime})   \|_{\infty}  \leq 2\epsilon}{\sum} \,\,\underset{A^{\prime} }{\bigints}  \|  \nabla g_{\epsilon}(z) \|_1   \,dz  \,\Bigg\} \frac{2\,C_8\,N^d}{N}  \\
	\,\,\, \,\,\,\,\,\,\,\leq \,\, \frac{2\,C_8\,N^d}{N}\,\,  \underset{A  \in \mathcal{P}_{\epsilon} }{\sum} \, \left[\vert \{ A^{\prime}  \in  \mathcal{P}_{\epsilon}    \,:\,\, \|  c(A)  \,-\,  c(A^{\prime})   \|_{\infty}  \leq 2\epsilon \} \vert\right]^3  \,\,  \underset{ x  \in h^{-1}\{P_{\text{lat}}(N)\}  }{\max} \,\vert C(x)\vert^2\,\,\underset{A}{\bigints}  \|  \nabla g_{\epsilon}(z) \|_1   \,dz  \\
	\,\,\, \,\,\,\,\,\,\, \leq \,\,  C_9\, N^{d-1} \,\underset{ x  \in h^{-1}\{P_{\text{lat}}(N)\}  }{\max} \,\vert C(x)\vert^2 \underset{A  \in \mathcal{P}_{\epsilon} }{\sum} \underset{A}{\bigints}  \|  \nabla g_{\epsilon}(z) \|_1   \,dz,  \\
	\end{array}
	\]
	where  $C_9$  is positive constant depending on $d$ which can be obtained  with an entropy  argument similar to (\ref{eqn:entropy_argument}). As a result, from (\ref{eqn:first_bound})  we obtain 
	\begin{equation}
	\label{eqn:T1}
	T_1 \,\,\leq \,\,    C_2 C_9\, N^{d-1} \,\underset{ x  \in h^{-1}\{P_{\text{lat}}(N)\}  }{\max} \,\vert C(x)\vert^2.	   
	\end{equation}
	
	It remains to bound $T_2$ in (\ref{eqn:smooth}). Towards that goal, we observe that
	\[
	\begin{array}{lll}
	\underset{B_{\epsilon}(z_{i,j}) }{\int}  \left\vert \|  \nabla g_{\epsilon}(z_{i,j}) \|_1  \,-\, \|  \nabla g_{\epsilon}(z) \|_1  \right\vert  \,dz   &\leq & 	  \displaystyle \underset{B_{\epsilon}(z_{i,j}) }{\int}  \sum_{l=1}^d    \left\vert  \frac{\partial g_{\epsilon} }{\partial z_l}(z_{i,j}) \,-\, \frac{\partial g_{\epsilon}}{\partial z_l}(z)   \right\vert  \,dz    \\
	&=& \displaystyle \underset{B_{\epsilon}(z_{i,j}) }{\int}  \sum_{l=1}^d\Bigg\vert  \frac{1}{\epsilon^{d+1}}   \underset{B_{\epsilon}(0) }{\int}   \frac{\partial \psi}{\partial z_l}(z^{\prime}/\epsilon)  \left\{ g_{0}(z_{i,j}-z^{\prime})\,-\, g_{0}(z-z^{\prime})  \right\}   \,dz^{\prime} \Bigg\vert     dz \\
	& \leq  & \displaystyle \underset{B_{\epsilon}(z_{i,j}) }{\int}  \sum_{l=1}^d    \Bigg\vert  \frac{1}{\epsilon^{d}  \| z_{i,j} \,-\, z  \|_2 }   \underset{B_{\epsilon}(0) }{\int}   \frac{\partial \psi(z^{\prime}/\epsilon) }{\partial z_l} \{ g_{0}(z_{i,j}-z^{\prime})\,-\, \\
	& &\,\,\,\,\,\,\,\,\,\,\,\,\,\,\,\,\,\,\,\,\,\,\,\,\,\,\,g_{0}(z-z^{\prime})  \}   \,dz^{\prime} \Bigg\vert     dz, \\
	\end{array}
	\]
	which implies that for some $C_{10} >0$ 
	\[
	\underset{B_{\epsilon}(z_{i,j}) }{\int}  \left\vert \|  \nabla g_{\epsilon}(z_{i,j}) \|_1  \,-\, \|  \nabla g_{\epsilon}(z) \|_1  \right\vert  \,dz   \,\,\leq  \,\,    C_{10}  \,\epsilon^d  \,\underset{z   \in B_{\epsilon}(z_{i,j}) }{\sup}      \sum_{l=1}^d    \Bigg\vert     \underset{B_{\epsilon}(0) }{\int}   \frac{\partial \psi(z^{\prime}/\epsilon) }{\partial z_l} \frac{\{g_{0}(z_{i,j}-z^{\prime})\,-\, 
		g_{0}(z-z^{\prime})  \}   }{\epsilon^{d}  \| z_{i,j} \,-\, z  \|_2} \,dz^{\prime} \Bigg\vert \\
	\]
	and so
	\begin{equation}
	\label{eqn:t2}
	\def\arraystretch{2.3}
	\begin{array}{l}
	T_2    \\
	\leq \,\, \underset{A  \in \mathcal{P}_{\epsilon},\,\,A\cap ( \Omega_{2\epsilon}\backslash B_{2\epsilon}(\mathcal{S} )) \,\neq\, \emptyset    }{\sum} 
	\,\,\,\, \Bigg[\underset{(i,j) \in  E_K\, i,j \in J,\,\,z_{i,j} \in A }{\sum}\,  \,\,\,  \underset{z \in B_{\epsilon}(z_{i,j})  }{\sup} \, \displaystyle  \sum_{l=1}^d \,  \Bigg\vert  \underset{ \| z^{\prime}\|_2 \leq \epsilon }{\int}   \frac{\partial\psi(z^{\prime}/\epsilon)}{\partial z_{l}}  \left(\frac{ g_0(z_{i,j} - z^{\prime})\,-\,  g_0(z - z^{\prime})}{ \|z -z_{i,j}\|_2 \,\epsilon^{d} }\right) \,dz^{\prime}   \Bigg\vert        \,\Bigg]\\
	\,\,\,\,\,\,\cdot 2\, C_{10}\,C_8  \,\epsilon^d\,  N^{d-1}\\	
	=\,\,\underset{A  \in \mathcal{P}_{\epsilon},\,\,A\cap ( \Omega_{2\epsilon}\backslash B_{2\epsilon}(\mathcal{S} )) \,\neq\, \emptyset    }{\sum} 
	\,\,\,\, \Bigg[\underset{(i,j) \in  E_K\, i,j \in J,\,\,z_{i,j} \in A }{\sum}\,  \,\,\, T(g_0,z_{i,j})\epsilon^d     \,\Bigg]\cdot 2\, C_{10}\,C_8  \,\,  N^{d-1},\\
	\end{array}
	\end{equation}
	with $T(g_0, z)$ as in  Equation (\ref{summation3}) in the main paper. Now, if  $z_{i,j} \in A$, $h(x_i)   \in A_i$  and  $h(x_j)   \in A_j$, then just as in \eqref{eqn:cA} we obtain that
	\[
	\max\{  \|  c(A_i)  \,-\,  c(A)    \|_{\infty}, \|  c(A_j)  \,-\,  c(A)    \|_{\infty} \}  \,\leq \, 2 \epsilon.
	\]    	
	Hence, since by construction we also have $z_{i,j} \in  \Omega_{2\epsilon}\backslash B_{2\,\epsilon}(\mathcal{S} )$ then
	\[
	\begin{array}{lll}
	\underset{(i,j) \in  E_K\, i,j \in J,\,\,z_{i,j} \in A }{\sum}\,  \,\,\, T(g_0,z_{i,j}) &\leq & \left\vert  \{  \{ i,j\} \,:\,   \max\{  \|  c(A_i)  \,-\,  c(A)    \|_{\infty}, \|  c(A_j)  \,-\,  c(A)    \|_{\infty} \}  \,\leq \, 2 \epsilon    \}  \right\vert\\
	&& \, \underset{z_A \in A \cap ( \Omega_{2\epsilon}\backslash B_{2\,\epsilon}(\mathcal{S} ))  }{\sup} T(g_0,z_A),
	\end{array}
	\]
	which combined with \eqref{eqn:t2} implies that
	\begin{equation}
	\label{eqn:t2_2}
	\def\arraystretch{2.3}
	\begin{array}{lll}
	T_2  &	\leq& \,\, \underset{A  \in \mathcal{P}_{\epsilon},\,\,A\cap ( \Omega_{2\epsilon}\backslash B_{2\epsilon}(\mathcal{S} )) \,\neq\, \emptyset    }{\sum} 
	\,\,\,\, \Bigg[ \underset{z_A \in A \cap ( \Omega_{2\epsilon}\backslash B_{2\,\epsilon}(\mathcal{S} ))}{\sup}  \,  T(g_0,z_{i,j})\epsilon^d     \,\Bigg]\cdot 2\, C_{10}\,C_8  \,\,  N^{d-1}\, \cdot \\
	& &	\,\,\,\,\,\left\vert  \{  \{i,j\} \,:\,   \max\{  \|  c(A_i)  \,-\,  c(A)    \|_{\infty}, \|  c(A_j)  \,-\,  c(A)    \|_{\infty} \}  \,\leq \, 2 \epsilon    \}  \right\vert \\
	&	\leq& \,\,\,	  C_{11}\,\,  N^{d-1}  \underset{ x  \in h^{-1}(P_{\text{lat}}(N))  }{\max} \,\vert C(x)\vert^2,
	\end{array}
	\end{equation}
	for some positive constant $C_{11}$,  where the last inequality follows from Assumption \ref{as:sup}    and that fact that for every $A$  the set of pairs of cells  with centers within  distance $2\epsilon$ is constant.   
	Combining   (\ref{eqn:ineq2}), (\ref{eqn:T1}), (\ref{eqn:smooth}), (\ref{eqn:t2_2})  and Lemma \ref{lem:controlling_counts2},  we obtain  that   
	\begin{equation}
	\label{eqn:final_bound}
	\underset{(i,j) \in  E_K\,  }{\sum} \vert \theta_i^* \,-\, \theta_j^* \vert  \,=\, O_{\text{pr}}(\text{poly}(\log n)  \,n^{1-1/d}  ),
	\end{equation}
	when Assumptions \ref{as1}--\ref{as:sup}  hold.  

	To conclude the proof, we proceed to verify (\ref{eqn:final_bound})  when Assumptions \ref{as1}--\ref{as3}  hold  and $f_0$ satisfies Definition \ref{def:pie_lip}.
	Using the notation from before, we observe that  \eqref{eqn:prob}  still holds  and  for  $J$ in \eqref{eqn:J}  we have that
	\begin{equation}
	\label{eqn:pl}
	\begin{array}{lll}
	\underset{(i,j) \in  E_K\,  }{\sum} \vert \theta_i^* \,-\, \theta_j^* \vert  &  = &  \underset{(i,j) \in  E_K,\, i,j \in J  }{\sum} \vert \theta_i^* \,-\, \theta_j^* \vert   \,\,\,+\,\,\,\underset{(i,j) \in  E_K,\, i \notin J   \,\text{or}\,\,j \notin J    }{\sum} \vert \theta_i^* \,-\, \theta_j^* \vert\\
	& \leq&  \underset{(i,j) \in  E_K\, i,j \in J }{\sum} \vert f_0(x_i) \,-\, f_0(x_j) \vert   \,+\,\,2\,\|g_0\|_{L_{\infty}(0,1)^d} K\,\tau_d\,\tilde{n}.    \,\,\\
	\end{array}
	\end{equation}	
	So  the claim  follows from  combining Lemma \ref{lem:controlling_dist},  (\ref{eqn:counting}),  and the piecewise Lipschitz condition.
	
\end{proof}

\section{Proof of Theorem  \ref{thm:adaptivity_dimension}  }

\begin{proof}
	
	Recall
	\begin{equation}
	\label{mixture_model2}
	\begin{array}{lll}
	y_i   & =  &  \theta^*_i  \,\,+\,\,\varepsilon_i,\,\,\,\,\,\,i=1,\ldots,n,\\
	\theta^{*}_i   &  = & f_{0,z_i}(x_i),  \,\,\,   \\
	x_i & \sim &   p_{z_i}(x),  \\
	\text{pr}(z_i =l)& \sim& \pi_l^*,  \,\,\,\text{for}  \,\,l = 1,\ldots,\ell. 
	\end{array}
	\end{equation} 
	
	Next we let  $A_l \,=\,\{ i \in [n]\,:\, z_i = l \}$, and $n_l \,=\,\vert A_l \vert$ for $l = 1,\ldots, \ell$. Then by   Proposition 27 in \cite{von2010hitting} we have that the event
	\begin{equation}
	\label{eqn:event_n_l}
	\frac{n\pi_l^* }{2} \,\leq \, n_l  \,\leq \, \frac{3n\pi_l^*}{2},\,\,\,\,\,\,\text{for}  \,\,\,\,l = 1,\ldots, \ell,
	\end{equation} 
	happens with probability at least $1 \,-\,2l \exp(-n\min\{\pi_l^*\,:\,  l \in [\ell] \} /12)$. Therefore,  we will assume that  the event defined in (\ref{eqn:event_n_l}) holds.

	We proceed by conditioning on $\{z_i\}_{i=1}^{n}$, and for simplicity we  omit to make this conditioning explicit. 	For $i \in [n]$ we write $l(i)  := j \in [\ell]$ if $i \in A_j$. 
	We also introduce the following notation:
	\[
	\begin{array}{lll}
	N_K(x_i) &=& \left\{ \text{ set of nearest neighbors of} \, x_i\,\text{in the $K$-NN  graph constructed from  points } \,\{x_{i^\prime} \}_{i^\prime = 1,\ldots,n}  \right\},\\
	\tilde{N}_K(x_i) & = & \left\{    \text{ set of nearest neighbors of} \, x_i\,\text{in the $K$-NN  graph constructed from  points } \,\{x_{i^\prime} \}_{i^\prime \in A_{l(i)} }           \right\},\\
	\tilde{R}_{K,\max}^l &=&  \underset{ i  \in A_l}{\max}\,\underset{x \in \tilde{N}_K(x_i) }{\max} \,d_{\mathcal{X}}(x,x_i),
	\end{array}
	\]
	and we denote by $\nabla_{G_K^l}$  the incidence matrix  of the $K$-NN  graph corresponding to the points $\{x_i\}_{i \in A_l}$. 
	
	Next we observe that just as in the proof of Lemma \ref{lem:controlling_counts},
	\begin{equation}
	\label{eqn:high_prob}
	\textup{pr}\left\{\tilde{R}_{K,\max}^l\,>\,  \tilde{a}_l(K/n_l)^{1/d_l}  \right\} \,\,\leq  \,\, n_l\,\exp(-K/12),  		
	\end{equation}
	for some positive constant $\tilde{a}_l$. We write $\epsilon_l \,=\, \tilde{a}_l(K/n_l)^{1/d_l}$, and consider the sets
	\[
	\Lambda_l \,\,=\,\, \left\{ i \in A_l\,\,:\,\, \text{such that}  \,\, N_K(x_i) \,\,=\,\, \tilde{N}_K(x_i)  \right\}.
	\]
	Our goal is to use these sets in order to split the basic inequality for the $K$-NN-FL  into different processes corresponding to the different sets $\mathcal{X}_l$. To that end let us pick  $l \in [\ell]$.  We notice that  if $\partial  \mathcal{X}_l \,=\,  \emptyset $,  then  by Assumption \ref{as:partition} we have that  for all $x \in \mathcal{X}_l$  it holds that $B_{\epsilon}(x) \subset \mathcal{X}_l$ for small enough $\epsilon$. Hence, with high probability, $ N_K(x_i) \,\,=\,\, \tilde{N}_K(x_i)  $ for all $i \in A_l$.     
	
	Let  us now assume that 	$\partial \mathcal{X}_l   \,\neq \, \emptyset$. Let
	$i \in A_l$.  Then 
	\begin{equation}
	\label{eqn:lower_bound_step}
	\begin{array}{lll}
	\textup{pr}\left\{ x_i   \in  B_{\epsilon_l}(\partial \mathcal{X}_l )  \right\}    \,\,\geq \,\,  p_{l,\min} \, \mu_l\left\{   B_{\epsilon_l} \left(  \partial \mathcal{X}_l    \right)   \bigcap  \mathcal{X}_l   \right\}   \,\,\geq \,\,c^{\prime}_l\,\epsilon_l^{d_l},
	\end{array}
	\end{equation}
	where $p_{l,\min}   =  \min_{x \in \mathcal{X}_l}   p_l(x) $, and $c^{\prime}_l$ is a positive constant that exists   because 
	$\mathcal{X}_l$  satisfies  Assumption \ref{as2}. 
	On the other hand,   (\ref{eqn:measure_condition}) implies that for $i \in A_l$  we have
	\begin{equation}
	\label{eqn:upper_bound_step}
	\textup{pr}\left\{x_i   \in  B_{\epsilon_l}(\partial \mathcal{X}_l)  \right\}    \,\,\leq \,\,  p_{l,\max} \, \mu_l\left\{   B_{\epsilon_l} \left(  \partial \mathcal{X}_l     \right)   \bigcap  \mathcal{X}_l   \right\}   \,\,\leq \,\,p_{l,\max} \,\tilde{c}_l\,\epsilon_l, 
	\end{equation}
	where $p_{l,\max}   =  \max_{x \in \mathcal{X}_l}   p_l(x) $, and $\tilde{c}_l$ is a positive constant.
	
	Therefore,  combining  (\ref{eqn:connected}), (\ref{eqn:high_prob}), (\ref{eqn:lower_bound_step}) ,  and (\ref{eqn:upper_bound_step})    with  Lemma \ref{lem:binom_concentration}, we obtain that
	\begin{equation}
	\label{eqn:counts_outisde}
	\begin{array}{lll}
	\textup{pr}\left( n_l  -  \vert  \Lambda_l\vert   \,\leq \,   \frac{3}{2}p_{l,\max} \,\tilde{c}_l\,n_l\,\epsilon_l   \right)     & \geq&
	1\,\,-\,\,  p_{\max} \exp\left( -\frac{1}{12}c^{\prime}_l\,\tilde{a}_l^{d_l}\,K  \right)  \,\,-\,\,  n_l\,\exp(-K/12)\\
	& \geq &  
	1\,\,-\,\,  p_{\max} \exp\left( -\frac{1}{12}c^{\prime}_l\,\tilde{a}_l^{d_l}\,K  \right)  \,\,-\,\, \,\exp(-K/12  +  \log n ).\\
	\end{array}
	\end{equation}
	
	Next we see how the previous inequality can be used to  put an upper bound on the penalty term of the $K$-NN-FL. For any $\theta \in \mathbb{R}^n$,  we have
	\[
	\begin{array}{lll}
	2\| \nabla_{G_K}\theta\|_1  \,\,  = \,\,   \displaystyle   \sum_{l=1}^{\ell} \sum_{i=1}^{n_l} \sum_{j \in N_K(x_i)}  \,\vert \theta_i \,-\, \theta_j\vert \,\,=\,\,
	\sum_{l=1}^{\ell} \sum_{i=1}^{n_l} \sum_{j \in \tilde{N}_K(x_i)}  \,\vert \theta_i \,-\, \theta_j\vert  \,\,+\,\,  R(\theta), 
	\end{array}
	\]
	where 
	\begin{equation}
	\label{eqn:R_theta}
	\def\arraystretch{2.3} 
	\begin{array}{lll}
	\vert  R(\theta) \vert    & = &  \Bigg\vert      \displaystyle   \sum_{l=1}^{\ell} \sum_{i=1}^{n_l}\, \sum_{j \in N_K(x_i)}  \,\vert \theta_i \,-\, \theta_j\vert      \,\,-\,\,     \sum_{l=1}^{\ell} \sum_{i=1}^{n_l}\, \sum_{j \in \tilde{N}_K(x_i)}  \,\vert \theta_i \,-\, \theta_j\vert              \Bigg \vert \\
	&=&    \Bigg\vert      \displaystyle   \sum_{l=1}^{\ell} \sum_{i \in  [n_l] \backslash  \Lambda_l }\,\, \sum_{j \in N_K(x_i)}  \,\vert \theta_i \,-\, \theta_j\vert      \,\,-\,\,     \sum_{l=1}^{\ell} \sum_{i \in  [n_l] \backslash  \Lambda_l } \,\,\sum_{j \in \tilde{N}_K(x_i)}  \,\vert \theta_i \,-\, \theta_j\vert          \Bigg \vert    \\
	& \leq &   4\,\tau_d\, \|\theta\|_{\infty}\,K\, \displaystyle   \sum_{l=1}^{\ell}  \vert [n_l] \backslash  \Lambda_l\vert, 
	\end{array}
	\end{equation}
	where $\tau_d$ is a positive constant  depending only on $d$, and the second inequality follows from the well-known bound on the maximum degree of $K$-NN graphs, see Corollary 3.23 in \cite{miller1997separators}. Combining with 
	(\ref{eqn:counts_outisde}), we obtain that
	\begin{equation}
	\label{eqn:R_upper_bound}
	R(\theta  )  \,\,-\,\, R(\hat{\theta} )   \,\,\,=\,\,\, O_{\textup{pr}} \left[\left\{ (\log n)^{ \frac{1}{2}  }   \,+\,  K\,\lambda \right\}\,K^{1+1/d_l}\, \displaystyle   \sum_{l=1}^{\ell}   (  n \pi_l^* )^{1-1/d_l}     \right].
	\end{equation}
	
	Note that if $ \partial \mathcal{X}_l  \,=\,\emptyset $, then (\ref{eqn:R_upper_bound}) will still hold  as $R(\theta)   \,=\,0$ for  all $\theta \in \mathbb{R}^n$ with high probability.   
	
	To conclude the proof, we notice by the basic inequality that
	\[
	\begin{array}{lll}
	\| \theta^*  \,-\, \hat{\theta} \|_n^2   &   \leq&      \frac{1}{n} \varepsilon^T\left( \hat{\theta}  \,-\,\theta^*  \right)   \,\,+\,\, \frac{\lambda}{n} \left(  \| \nabla_{G_K}\theta^* \|_1   \,-\,   \| \nabla_{G_K}  \hat{\theta} \|_1    \right)\\
	&  = &   \displaystyle   \sum_{l=1}^{\ell}   \varepsilon_{ A_l }^T  \left( \hat{\theta}_{A_l}   \,-\,   \theta^*_{A_l} \right)    \,\,+\,\,  \sum_{l=1}^{\ell} \frac{\lambda}{2n} \left(  \| \nabla_{G_K^l}\theta^*_{A_l} \|_1   \,-\,   \| \nabla_{G_K^l}  \hat{\theta}_{A_l} \|_1    \right)   \,\,+\,\,  \frac{\lambda}{2n} \left\{ R(\theta^* )   \,-\, R(\hat{\theta})  \right\} ,      \\
	\end{array} 	      
	\]
	where the notation  $x_A$ denotes the vector $x$ with the coordinates with indices not in $A$ removed. The proof follows from  (\ref{eqn:R_upper_bound})   and from bounding each term 
	\[
	\frac{1}{n} \varepsilon_{ A_l }^T  \left( \hat{\theta}_{A_l}   \,-\,   \theta^*_{A_l} \right)    \,\,+\,\, \frac{\lambda}{n } \left( \| \nabla_{G_K^l}\theta^*_{A_l} \|_1   \,-\,   \| \nabla_{G_K^l}  \hat{\theta}_{A_l} \|_1\right),
	\]
	as in the proof of Theorem  \ref{thm:knn}.

\end{proof}

\section{Manifold Adaptivity Example  }
\label{appendix_example}

The following example suggests that the $\epsilon$-NN-FL estimator may not be  manifold adaptive.

\begin{example}
	\label{ex:non_adaptivity}
	For $\mathcal{X}=\mathcal{X}_1 \cup \mathcal{X}_2  \subset \mathbb{R}^3$, suppose that $\mathcal{X}_1   \,=\, [0,1]^2 \times \{c\}$  for some 
	$c < 0$,  and  $\mathcal{X}_2  \,=\, [0,1]^3$.  
	Note that  the sets $\mathcal{X}_1$  and $\mathcal{X}_2$  satisfy Assumption~\ref{as:partition}. 
	Let us assume that  all of the conditions of Theorem \ref{thm:adaptivity_dimension}  are met
	with $n_1  \,:=\, n \pi^*_1    \, \asymp\, n   $, and $n_2 \,:=\, n \pi^*_2  \,\asymp\, n^{3/4}$. Motivated by the scaling of $\epsilon$ in Theorem~\ref{thm:penalty}, 
	we let $\epsilon     \,\asymp\,     \text{poly}(\log n) / n^{1/t}$. We now consider two possibilities for $t$: $t \in (0,2]$, and $t>2$. 
	\begin{itemize}
		\item  For any $t \in  (0,2]$, and  for any positive constant $a \in (0,3/4)$,
		there exists a positive constant $c(a)$  such that $E(\|\hat{\theta}_\epsilon   \,-\,\theta^*   \|_ {n}^2) \geq c(a)\,n^{-1/4 - a} $
		for large enough $n$, where $\hat{\theta}_\epsilon$ is the $\epsilon$-NN-FL estimator, see proof below. 
		In other words, if $t<2$, then $\epsilon$-NN-FL does not achieve a desirable rate. 
		By contrast, with an appropriate  choice of  $\lambda$ and $K$,  the $K$-NN-FL estimator attains the rate 
		$n^{-1/2}  \,\asymp \, n_1^{1-1/2}/n + n_2^{1-1/3}/n  $  (ignoring logarithmic terms)  by Theorem \ref{thm:adaptivity_dimension}.  
		
		\item 		 If   $t >  2$,   then the minimum degree in  the $\epsilon$-graph of the observations in $\mathcal{X}_1$  
		will be at least $c(t)  n^{1-2/t}$,  with high probability, for some positive constant $c(t)$.  
		This has two consequences:
		\begin{enumerate}
			\item 
			
			Recall that the  algorithm from \cite{chambolle2009total} 
			has worst-case complexity $O(mn^2)$, where $m$ is the
			number of edges in the graph  \cite[although empirically the algorithm is typically much faster, ][]{wang2016trend}. Therefore, 
			if $t$ is too large, then the computations involved in applying the fused lasso to the $\epsilon$-NN graph
			may be too demanding. 
			\item
			In  a different  context,  \cite{el2016asymptotic}  argues that it is desirable for 
			geometric graphs (which generalize  $\epsilon$-NN  graphs)   to have degree  ${\rm poly}(\log n)$. This suggests that 
			using $t>2$ leads to an $\epsilon$-NN graph that is too dense. 
		\end{enumerate}
	\end{itemize}
\end{example}

Next, we proceed to justify the lower bound on the MSE  of $\hat{\theta}_\epsilon$  when $t \in (0,2]$ in Example \ref{ex:non_adaptivity}. 	We begin by noticing that if $i , j \in A_2$  with $i \neq j$,   then for the choice of $\epsilon  >0$ in this example,  we have that 
\[
\begin{array}{lll}
\textup{pr}\{  d_{\mathcal{X}} (x_i,x_j)  \,\leq  \,  \,\epsilon   \}  
\,\,    \leq \,\, p_{2,\max} \displaystyle \underset{ \mathcal{X}   }{\int } \, \textup{pr}\{ d_{\mathcal{X}} (x,x_j)  \,\leq  \,  \,\epsilon   \}  \mu_2(dx)     \,\,\leq \,\,  \tilde{k}_l\,\frac{ (\text{poly}(\log n))^3  }{n^{3/t}    }.    \\
\end{array}
\]
Let $  m  \asymp  n^{3/4   - a  } $,  and $j_1 <  j_2 <\ldots <  j_m$ elements of $A_2$.    Then   the event 
\[
\Lambda  \,\,=\,\,\cap_{s =1}^m     \{   d_{\mathcal{X}}(x_{j_s}, x_l )   > \epsilon,\,\,\,\forall  l  \in A_2 \backslash   \{j_s \}   \},
\]
satisfies
\[
\displaystyle   \textup{pr}\left(  \Lambda\right)   \,\,\geq \,\, 1 \,-\, 
c_2  \,n^{3/2  - 3/t  - a}\, ( \text{poly}(\log n))^{3},\,  
\]
for some positive constant $c_2$. Therefore,
\[
\displaystyle    \mathbb{E}\left\{ \sum_{i=1}^{n} (\theta_i^* \,-\, \hat{\theta}_{\epsilon,i} )^2   \right\} \,\,\geq \,\,  \mathbb{E}\left\{  \sum_{i=1}^{n} (\theta_i^* \,-\, \hat{\theta}_{\epsilon,i} )^2   \mid   \Lambda \right\}\,\textup{pr}(\Lambda)    \,\,\geq \,\,  m\,\sigma^2 [1 \,-\, 
c_s   \,n^{3/2  - 3/t  - a}\, \{ \text{poly}(\log n)\}^{3}]  \,\,\geq \,\,  C_1  \,  n^{3/4 -a}
\]
for some positive constant $C_1$  if $n$  is large enough.

\section{Proving Theorems \ref{thm:upper_bound}--\ref{thm:penalty} for $\epsilon$-NN   graphs}

We start by giving an overview  of how  the conclusion of Theorem 1 in the main paper can be obtained for $\epsilon$-NN  graphs. The general idea is described next. The first step  is to  obtain a lemma  that controls the maximum and minimum degrees of the $\epsilon$-NN  graph. 
\begin{lemma}
	\label{lem:degree_epsilon}
	\textbf{(See Proposition 29 in \cite{von2010hitting} ). } Suppose  that Assumptions \ref{as1}--\ref{as3} hold. Let  $d_{\min}$  and $d_{\max}$ denote  the minimum  and maximum degrees of an $\epsilon$-NN  graph, respectively. Then for all $\delta \in (0,1)$ we have that
	\[
	\begin{array}{lll}
	\text{pr}\left\{ d_{\max} \geq (1+\delta)n \epsilon^d c_{\max}    \right\}   & \leq&   \exp\left(   -\frac{  \delta^2 n\epsilon^d  c_{\max}    }{3}  \right),\\
	\text{pr}\left\{ d_{\max} \leq (1-\delta)n \epsilon^d c_{\min}    \right\}   & \leq&   \exp\left(   -\frac{  \delta^2 n\epsilon^d  c_{\min}    }{3}  \right),\\
	\end{array}
	\]	
	for positive constants  $c_{\min}$  and $c_{\max}$.
\end{lemma}

Next we present a lemma  controlling the maximum and minimum counts of the mesh. This is similar to Lemma \ref{lem:controlling_counts}. 

\begin{lemma}
	\label{lem:controlling_counts_epsilon}
	Let  $\epsilon \,\asymp\,  \log^{(1+2r)/d} n/ n^{1/d}   $, then  there exists $N$ such  that if
	\[
	N\,\,\asymp\,\,  \ceil[\Bigg]{\frac{1}{\epsilon} }  
	\]
	in the construction of $P_{\text{lat}}(N)$ defined in \eqref{eqn:lattice}, then the following properties hold:
	\begin{itemize}
		\item There exist positive constants $b_2^{\prime}$ and $b_2^{\prime}$ such that
		\begin{equation*}
		\begin{array}{lll}
		\textup{pr}\bigg\{\underset{x  \in  I(N) }{\max}\vert C(x)\vert   \,\,\geq \,\, (1+\delta)\,b_{2}^{\prime}\,\, \log^{1+2r} n  \bigg\}    &  \leq & N^d \exp\left(  -\frac{1}{3}\delta^2\,b_{2}^{\prime}\,\log^{1+2r} n   \right),\\
		\vspace{0.3in}
		\textup{pr}\bigg\{ \underset{x  \in  I(N) }{\min}\vert C(x)\vert   \,\,\leq \,\, (1-\delta)\,b_{1}^{\prime}\,\log^{1+2r} n  \bigg\}  &\leq  & N^d \exp\left(  -\frac{1}{3}\delta^2\,b_{1}^{\prime}\,\log^{1+2r} n   \right),
		\end{array}
		\end{equation*}
		for all $\delta \in (0,1)$.
		\item 	
		Define the event $\Omega$  as:  ``If  $x_i  \in C(x_i^{\prime})$ and $x_j \in C(x_j^{\prime})$  for $x_i^{\prime},x_j^{\prime} \in I(N)$ with $\|h(x_i^{\prime}) \,-\,h(x_j^{\prime}) \|_2 \,\leq\,  \,N^{-1} $, then $x_i$ and $x_j$  are connected in the $\epsilon$-NN graph". Then for some constant $c>0$,  
		$$ \textup{pr}\left(\Omega\right)  \,\, \geq \,\,  1\,\,-\,\,  n\,\exp(-c \log^{1+2r} n /3).$$
	\end{itemize}
\end{lemma}
The proof  of Theorem \ref{thm:upper_bound} results from  setting  $\epsilon  $ as in Lemma \ref{lem:controlling_counts_epsilon}, and then  following  the steps in the proof of Theorem 1 in the main document. Specifically,  we can follow   the proofs of Lemmas \ref{embedding}--\ref{lem:empirical_process_3}, and Theorem  \ref{thm:knn}, by replacing  $\nabla_{G_K}$ with $\nabla_{G_{\epsilon}}$, and  $K$ with  $c \log^{1+2r} n $  for a constant $c$.

To prove the conclusion of Theorem \ref{thm:penalty} for the $\epsilon$-NN  graph, we  need  lemmas  similar to those in Section \ref{sec:useful_lemma}. The first of these lemmas is given next.

\begin{lemma}
	\label{lem:controlling_counts2.2}
	Let  $\epsilon \,\asymp\,  \log^{(1+2r)/d} n/ n^{1/d}   $. Then  there exists $N$  in the construction of $G_{\text{lat}}$ with
	\[
	N\,\,\asymp\,\, \ceil{ \epsilon^{-1} },
	\]
	and positive constants $b^{\prime}_1$ and $b^{\prime}_2$ depending on $L_{\min}$, $L_{\max}$, $d$, $p_{\min}$, $c_{1,d}$, and $c_{2,d} $, such that
	\begin{equation}
	\begin{array}{lll}
	\textup{pr}\bigg\{ \underset{x  \in   h^{-1}(P_{\text{lat}}) }{\max}\vert C(x)\vert   \,\,\geq \,\, (1+\delta)\,c_{2,d}\,b^{\prime}_1\,\log^{(1+2r)} n \bigg\}   &  \leq &  N^d \exp\left(  -\frac{1}{3}\delta^2\,b^{\prime}_2\,c_{1,d}\,\log^{(1+2r)} n   \right),\\
	\textup{pr}\bigg\{ \underset{x  \in   h^{-1}(P_{\text{lat}}) }{\min}\vert C(x)\vert   \,\,\leq \,\, (1-\delta)\,c_{1,d}\,b^{\prime}_2\,\log^{(1+2r)} n  \bigg\} &\leq  &N^d \exp\left(  -\frac{1}{3}\delta^2\,b^{\prime}_2\,c_{1,d}\,\log^{(1+2r)} n   \right),
	\end{array}
	\end{equation}
	for all $\delta \in (0,1]$. Moreover, let  $\tilde{\Omega}$ denote the event:  ``For all  $i, j \in [n]$, if $x_i$ and $x_j$  are connected in the $\epsilon$-NN graph, then  $\|h(x_i^{\prime}) \,-\,h(x_j^{\prime}) \|_2 \,<\,  2\,N^{-1} $  where  $x_i  \in C(x_i^{\prime})$ and $x_j \in C(x_j^{\prime})$  with $x_i^{\prime},x_j^{\prime} \in I(N)$". Then 
	$$\textup{pr}\left(\tilde{\Omega} \right) \,\,  \geq  \,\,  1\,\,-\,\,  n\,\exp(-\log^{(1+2r)} n/3). $$
\end{lemma}

Using  Lemma \ref{lem:controlling_counts2.2},  we can obtain   the conclusions of Lemmas \ref{thm:pred_rate}--\ref{thm:aerr}
and  Theorem \ref{thm:penalty} for the $\epsilon$-NN  graph. This can be done with minor modifications to the proofs in Section \ref{sec:useful_lemma}.